\documentclass[acmsmall]{acmart}

\AtBeginDocument{%
  }

\setcopyright{acmlicensed}
\acmDOI{10.1145/3729283}
\acmYear{2025}
\acmJournal{PACMPL}
\acmVolume{9}
\acmNumber{PLDI}
\acmArticle{180}
\acmMonth{6}

\usepackage{amsmath,amsthm,amsfonts,bbold,xspace,graphicx,float,mathtools}
\usepackage{braket,subcaption,xcolor,textcomp,pifont}
\usepackage{braket}
\usepackage[mathcal]{eucal} 
\usepackage{stmaryrd} 
\usepackage{enumitem} 
\usepackage{hyperref}
\usepackage[nameinlink,capitalize]{cleveref}

\usepackage{tabularray}

\usepackage{tikz}
\usetikzlibrary{arrows, shapes.multipart}
\usepackage{algorithm}
\usepackage{algpseudocodex}
\usepackage{wrapfig}
\usepackage{epigraph}

\makeatletter
\renewcommand{\ALG@beginalgorithmic}{\small}
\makeatother

\newcommand{\hlt}[1]{{\fontfamily{lmtt}\selectfont\bfseries \color{blue!50!black} #1}}

\newfloat{algorithm}{t}{lop}

\AtEndPreamble{%
\theoremstyle{acmdefinition}

\newtheorem{condition}[theorem]{Condition}
\Crefname{condition}{Condition}{Conditions}
}

\newcommand{\Co}{\mathbb C}

\newcommand{\bbone}{\mathbb 1}
\newcommand{\Id}{\bbone}

\DeclarePairedDelimiter\parens{\lparen}{\rparen}
\DeclarePairedDelimiter\abs{\lvert}{\rvert}

\DeclarePairedDelimiter\braces{\lbrace}{\rbrace}
\DeclarePairedDelimiter\bracks{\lbrack}{\rbrack}

\newcommand{\calC}{\mathcal{C}}

\newcommand{\calF}{\mathcal{F}}
\newcommand{\calG}{\mathcal{G}}
\newcommand{\calH}{\mathcal{H}}

\newcommand{\calO}{\mathcal{O}}
\newcommand{\calP}{\mathcal{P}}
\newcommand{\calQ}{\mathcal{Q}}

\newcommand{\calV}{\mathcal{V}}

\begin{document}

\title{Quantum Register Machine: Efficient Implementation of Quantum Recursive Programs}

\author{Zhicheng Zhang}
\orcid{0000-0002-7436-0426}
\affiliation{%
  \institution{University of Technology Sydney}
  \city{Sydney}
  \country{Australia}
}
\email{zhicheng.zhang@student.uts.edu.au}

\author{Mingsheng Ying}
\orcid{0000-0003-4847-702X}
\affiliation{%
  \institution{University of Technology Sydney}
  \city{Sydney}
  \country{Australia}
}
\email{mingsheng.ying@uts.edu.au}


\begin{abstract}
    Quantum recursive programming has been recently introduced for describing sophisticated and complicated quantum algorithms in a compact and elegant way. However,  implementation of quantum recursion involves intricate interplay between quantum control flow and recursive procedure calls.
	In this paper, we aim at resolving this fundamental challenge and develop a series of techniques to efficiently implement quantum recursive programs. Our main contributions include:
	\begin{enumerate}
		\item 
			We propose a notion of \textit{quantum register machine}, 
			the first quantum architecture 
			(including an instruction set) that
			provides instruction-level support for quantum control flow and recursive procedure calls at the same time.
		\item
			Based on quantum register machine,
			we describe the first \textit{comprehensive implementation  process}  of quantum recursive programs,
		  including the compilation,
			the partial evaluation of quantum control flow, and the execution on the quantum register machine.
		\item
			As a bonus, our efficient implementation of quantum recursive programs 
            also offers \textit{automatic parallelisation} of quantum algorithms.
			For implementing certain quantum algorithmic subroutine, 
			like the widely used quantum multiplexor,
			we can even obtain exponential parallel speed-up (over the straightforward implementation) from this automatic parallelisation.
            This demonstrates that quantum recursive programming can be win-win for both modularity  of programs and efficiency of their implementation.
	\end{enumerate}
\end{abstract}



\maketitle

\section{Introduction}
\label{sec:introduction}

Recursion in classical programming languages enables programmers to
conveniently describe complicated computations as compact programs.
By allowing any procedure to call itself,
a short static program text can generate (unbounded) long dynamic program execution~\cite{EWD249}.
Examples of recursion include Hoare's quicksort algorithm~\cite{Hoare61}, various recursive data structures~\cite{Hoare75}, and divide-and-conquer algorithms. The implementation of classical recursion has been well-studied and was an important feature of the celebrated programming language ALGOL 60~\cite{ALGOL60,ALGOL60b,Dij60,vdH15}.

In the context of quantum programming, recursion has been recently studied for similar reasons (e.g., \cite{Selinger04,Ying16,XYV21,DTPW24,YZ24}).
In particular, a language $\mathbf{RQC}^{++}$ was introduced in~\cite{Ying16,YZ24} for recursively programmed quantum circuits and quantum algorithms. The  expressive power of $\mathbf{RQC}^{++}$ has been demonstrated by various examples.

The aim of this paper is to study \textbf{how quantum recursive programs can be efficiently implemented}. We choose to consider quantum recursive programs~\footnote{As this terminology suggests, the recursion in such programs has a quantum nature.} 
described by the language $\mathbf{RQC}^{++}$~\cite{YZ24}. 
But we expect that the techniques developed in this paper can work for other quantum programming languages that support recursion. 

In general, quantum recursion involves the interplay of the following two programming features: 
\begin{itemize}
	\item 
		\textit{Quantum control flow}  
		(in particular, those defined by quantum if-statements~\cite{AG05,YYF12,SVV18,BBGV20,VLRH23,YVC24})
		that allow program executions to be in quantum superposition, controlled by some external quantum coin.
	\item
		\textit{Recursive procedure calls} that allow a procedure to call itself with different classical parameters.
\end{itemize}
A good implementation of quantum recursive programs  should support the above two features harmoniously.
A better implementation should further be efficient.

\subsection{Motivating Example: Quantum Multiplexor}
\label{sub:motivating_example}

\setlength{\epigraphwidth}{0.7\textwidth}
\epigraph{\ldots increase of efficiency always comes down to exploitation of structure \ldots}{\textsc{Edsger W. Dijkstra}~\cite{EWD117}}

To illustrate the basic idea of our implementation, let us  start with an algorithmic subroutine called quantum multiplexor~\cite{SBM05},
and see how quantum recursive programs can benefit its description and implementation.
Quantum multiplexor is used in a wide range of quantum algorithms, for example,
linear combination of unitaries (LCU)~\cite{CW12,BCCKS15,BCK15,Kothari14}, Hamiltonian simulation~\cite{BBKWLA16,BGBWMPFN18,BWMMNC18,LW19}, quantum state preparation~\cite{APPFd21,ZLY22,ZY24,LKS24},
and solving quantum linear system of equations~\cite{CKS17}.
Let $N=2^n$ and $[N]=\braces*{0,1,\ldots, N-1}$. A quantum multiplexor can be described by the unitary
\begin{equation}
	U=\sum\nolimits_{x\in [N]} \ket{x}\!\bra{x}\otimes U_x.\label{eq:qmux-uni}
\end{equation}
Here, every unitary $U_x$ is described by a quantum circuit,
or more generally, a quantum program, say $C_x$.
The quantum multiplexor $U$ applies $U_x$, 
conditioned on the state $\ket{x}$ of the first $n$ qubits.

A straightforward implementation of $U$ is
by applying a sequential products of $N$ controlled-$U_x$:
\begin{equation}
    \label{eq:straightforward-imp}
	\prod\nolimits_{x\in [N]}\Big(\ket{x}\!\bra{x}\otimes U_x +\sum\nolimits_{y\neq x} \ket{y}\!\bra{y}\otimes \Id\Big).
\end{equation}
This implementation has time complexity $O\parens*{\sum_{x\in [N]} T_x}$, where $T_x$ is the time for executing $C_x$ (i.e., implementing $U_x$). 
On the other hand, there exists a more efficient parallel implementation~\cite{ZLY22,ZY24} of $U$, 
with parallel time complexity $O\parens*{n+\max_{x\in [N]}T_x}$ (measured by the quantum circuit depth), using rather involved constructions similar to the bucket-brigade quantum random access memories~\cite{GLM08,GLM08b,HZZCSGJ19,HLGJ21}.
The implementation in~\cite{ZLY22,ZY24} achieves exponential parallel speed-up (with respect to $n$)
over the straightforward one.
The price for obtaining such efficiency is the manual design of rather low-level quantum circuits.

\begin{wrapfigure}{L}{0.46\textwidth}
	\centering
	\includegraphics[width=0.45\textwidth]{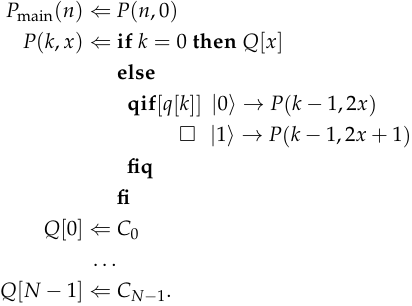}
	\caption{Quantum multiplexor as a quantum recursive program.}
	\label{fig:q-multiplexor}
\end{wrapfigure}

It is natural to ask \textit{if we can design at high-level and still obtain an efficient implementation}.
For this example of quantum multiplexor,
the intuition is as follows.
First, $U$ can be described by a high-level quantum recursive program $\calP$,
which encapsulates both the control structure in \Cref{eq:qmux-uni}
and all programs $C_x$ for describing unitaries $U_x$.
Then, by storing the program $\calP$ in a quantum memory,
we can design a \textit{quantum register machine} (to be formally defined in this paper) that automatically exploits the structure of $\calP$ and executes all $C_x$'s (i.e., implements all $U_x$'s) in quantum superposition,
thereby outperforming the straightforward implementation that only sequentially executes $C_x$'s.




Let us make the above intuition more concrete,
by describing $U$ in a quantum recursive program 
(in the language $\mathbf{RQC}^{++}$~\cite{YZ24}; 
see \Cref{sec:background})
as in \Cref{fig:q-multiplexor}.
Here, the main procedure $P_{\textup{main}}\parens*{n}$ describes $U$,
and every $Q[x]$ (or their procedure body $C_x$) describes $U_x$.
Procedure $\mathit{P}\parens*{k,x}$ 
recursively collects the control information $x$ using 
the quantum if-statement ($\mathbf{qif}$ statement)
and calls $Q[x]$ when $k=0$.
At this point, we only need to note that the program in \Cref{fig:q-multiplexor}
involves the interplay of the quantum control flow (managed by the $\mathbf{qif}$ statement)
and recursive procedure calls.
The $\mathbf{qif}$ statement in $\mathit{P}\parens*{k,x}$
creates two \textit{quantum branches} (in superposition):
when $q[k]$ is in state $\ket{0}$,
$\mathit{P}\parens*{k-1,2x}$ is called;
when $q[k]$ is in state $\ket{1}$,
$\mathit{P}\parens*{k-1, 2x+1}$ is called.

If the program in \Cref{fig:q-multiplexor} is compiled and 
stored into a quantum memory,
then a quantum register machine that supports quantum control flow and recursive procedure calls
can \textit{run through the two quantum branches
in superposition}.
The cost for executing the $\mathbf{qif}$ statement only depends on the quantum branch that takes longer running time.
This will incur
a final time complexity proportional to the maximum $\max_{x\in [N]} T_x$ (compared to the sum $\sum_{x\in [N]} T_x$ in the straightforward implementation)
and lead to an exponential parallel speed-up, similar to~\cite{ZLY22,ZY24}.

\subsection{Main Contributions}
\label{sub:contributions}


\subsubsection{Architecture: Quantum Register Machine}

We propose a notion of quantum register machine,
a quantum architecture that provides instruction-level support for quantum control flow and procedure calls at the same time.
Its storage components include a constant number of quantum registers (simply called  registers in the sequel) 
and a quantum random access memory (QRAM).
The QRAM stores both compiled quantum programs and quantum data.
The machine operates on registers like a classical CPU,
executing the compiled program by fetching instructions from the QRAM.
The machine is also accompanied with a set of low-level instructions,
each specifying operations to be carried out by the machine.
We briefly explain how the quantum register machine handles the aforementioned two features as follows:

\textbf{Handle quantum control flow}: 
Inspired by the previous work~\cite{YVC24} 
(which borrows ideas from the classical reversible architectures~\cite{Frank99,Vieri99,AGY07,TAG12}),
we put the program counter into a quantum register, which can be in quantum superposition.
However, existing techniques are insufficient to \textit{automatically} handle a challenge introduced by the quantum control flow, known as the \textit{synchronisation problem}~\cite{Deutsch85,Myers97,BV93,Ozawa98,LP98,Ozawa98b,Shi02,MO05,WY23,YVC24}.
Specifically, previous work~\cite{YVC24} circumvents this problem by manually inserting nop (no operation) into the low-level programs.
This approach changes the static program text, and is not extendable to handle quantum recursive programs, 
because the length of dynamic computation generated by quantum recursion cannot be pre-determined from the static program text (see \Cref{sub:further_synchronisation_problem} and \Cref{sec:related_work} for further discussion).

In contrast, to automatically handle the synchronisation problem (without changing the static program text), our solution is to use a partial evaluation of quantum control flow (to be explained soon) before execution,
and design a few corresponding quantum registers and mechanisms to exploit the partial evaluation result at runtime.

\textbf{Handle recursive procedure calls}: 
We allocate a call stack in the QRAM.
Stack operations are made reversible by borrowing techniques from the classical reversible computing (e.g., \cite{Axelsen11}).
Note that at runtime,
all quantum registers and the QRAM (where the dynamic call stack is stored)
can be in an entangled quantum state.

It is worth pointing out that the quantum register machine does not aim to model any existing quantum hardware (typically controlled by classical pulses to implement standard quantum circuits). Indeed, quantum register machine should better be thought of as an abstract machine (that does \textit{not} require hardware-level quantum control flow; see also \cite{YVC24}). Its execution is by repeatedly applying some fixed unitary operator per instruction cycle. Such unitary operator will be efficiently implemented by standard quantum circuits composed of one- and two-qubit gates.  

\subsubsection{Implementation: Compilation, Partial Evaluation and Execution}
\label{sub:implementation}

We propose a comprehensive process of implementing high-level quantum recursive programs 
(described in the language $\mathbf{RQC}^{++}$)
on the quantum register machine.
This includes the following three steps:
the first two are purely classical and the last is quantum.

\textbf{Step 1.\ Compilation} (\Cref{sec:compilation}):
The high-level program in $\mathbf{RQC}^{++}$ is compiled into a low-level one described by instructions,
together with a series of transformations.
The low-level instruction set is designed such that
the high-level program structure can be exploited for later execution.
This step only depends on the static program text and is independent of inputs.

\textbf{Step 2.\ Partial evaluation} (\Cref{sec:partial-evaluation}):
Given the classical inputs (typically specifying the size of quantum inputs),
the quantum control flow information of the compiled program is evaluated and stored into a data structure. In later execution,
it will be loaded into the QRAM to help address the aforementioned synchronisation problem.
This step is independent of quantum inputs.

\textbf{Step 3.\ Execution} (\Cref{sec:execution}):
With the compiled program and partial evaluation results loaded into the QRAM, 
the quantum inputs are finally considered,
and the compiled program is executed with the aid of the partial evaluation results.
The execution is done by repeatedly applying a fixed unitary (independent of the program) per cycle, 
which will be eventually implemented by standard quantum circuits with rigorously analysed complexity.

In \Cref{sec:computational_efficiency_and_algorithmic_speed_up}, we describe the theoretical complexity of \textbf{Step 2} and \textbf{3}.
More rigorous analysis can be found in \Cref{sub:complexity_partial_evaluation,sub:complexity_execution,sub:quantum_circuit_complexity_for_elementary_operations}.
The final parallel time complexity, measured by the standard asymptotic (classical and quantum) circuit depth, is $O\parens*{T_{\textup{exe}}\parens*{\calP}\cdot \parens*{T_{\textup{reg}}+T_{\textup{QRAM}}}}$.
Intuitively, $T_{\textup{exe}}\parens*{\calP}$ is the time for executing the longest quantum branch in program $\calP$; and $T_{\textup{reg}}$ and $T_{\textup{QRAM}}$ are complexities for elementary operations on registers and the QRAM, independent of the program.

\subsubsection{Bonus: Automatic Parallelisation}
\label{sub:bonus}

We show that quantum recursive programming can be \textit{win-win}
for both modularity of programs 
(demonstrated in~\cite{YZ24} via various examples)
and efficiency of their implementation (realised in this paper).
In particular, as a bonus,
the efficient implementation in \Cref{sub:implementation}
also offers \textit{automatic parallelisation}.
For implementing certain quantum algorithmic subroutine,
like the quantum multiplexor in \Cref{sub:motivating_example}, an exponential speed-up (over the straightforward implementation) can be obtained
from this automatic parallelisation,
in terms of (classical and quantum) parallel time complexity.
Here, the classical parallel time complexity is relevant
because the partial evaluation
will be performed by a classical parallel algorithm.

For implementing the quantum multiplexor, we obtain the following theorem from the automatic parallelisation,
whose proof sketch is to be shown in \Cref{sec:computational_efficiency_and_algorithmic_speed_up}.

\begin{theorem}[Automatic parallelisation of quantum multiplexor]
	\label{thm:parallel-qmux}
    Via the quantum register machine, 
	the quantum multiplexor in \Cref{eq:qmux-uni} with each $U_x$ consisting of $T_x$ elementary unitary gates
	can be implemented in (classical and quantum) parallel time complexity (i.e., circuit depth)
	$\widetilde{O}\parens*{n\cdot \max_{x\in [N]}T_x +n^2}$,
	where $\widetilde{O}\parens*{\cdot}$ hides logarithmic factors.
\end{theorem}

Although the complexity in \Cref{thm:parallel-qmux} is slightly worse than that in \cite{ZLY22,ZY24} by a factor of $\widetilde{O}(n)$,
it is worth stressing that the parallelisation in \Cref{thm:parallel-qmux} is obtained automatically.
Our framework steps towards a \textit{top-down} design of (parallel) efficient quantum algorithms: the programmer only needs to design the high-level quantum programs (like in \Cref{fig:q-multiplexor}), and the parallelisation is automatically realised by our implementation based on the quantum register machine.
Further comparison of \Cref{thm:parallel-qmux} and~\cite{ZLY22,ZY24} can be found in \Cref{sub:further_qmux}.

\subsection{Structure of the Paper}

For convenience of the reader, in \Cref{sec:background} we briefly review the language $\mathbf{RQC}^{++}$~\cite{YZ24} for describing quantum recursive programs. 
In \Cref{sec:quantum_register_machine}, we introduce the notion of quantum register machine.
In \Cref{sec:compilation}, we present the compilation of programs in $\mathbf{RQC}^{++}$ to low-level instructions.
Then, in \Cref{sec:partial-evaluation}, we present the partial evaluation of quantum control flow on the compiled program.
In \Cref{sec:execution}, we present the execution on quantum register machine.
Finally, in \Cref{sec:computational_efficiency_and_algorithmic_speed_up},
we analyse the efficiency of implementing quantum recursive programs in our framework,
and show how it offers automatic parallelisation.
In \Cref{sec:related_work} we discuss related work, and in \Cref{sec:discussion} we conclude and discuss future topics. Further details and examples are presented in the appendices. 


\section{Background on Quantum Recursive Programs}
\label{sec:background}

In this section, we briefly introduce
the high-level language $\mathbf{RQC}^{++}$ for describing quantum recursive programs,
defined in~\cite{YZ24}.
A more detailed introduction can be found in \Cref{sec:the_high_level_language_QRPL}.
Two key features of $\mathbf{RQC}^{++}$,
compared to other existing quantum programming languages,
are quantum control flow and recursive procedure calls,
which together support the quantum recursion (different from classical recursion in quantum programs as considered in e.g., \cite{VLRH23,DTPW24} and classically bounded recursion in superposition as considered in e.g., \cite{YC22,YC24}).
An additional contribution of this paper is providing further insights into $\mathbf{RQC}^{++}$ from an implementation perspective.

\subsection{Syntax}

The alphabet of $\mathbf{RQC}^{++}$ consists of:
(a) Classical variables, often denoted by $x,x_1,x_2,\ldots$;
(b) Quantum variables, often denoted by $q,q_1,q_2,\ldots$;
(c) Procedure identifiers, often denoted by $P,Q,P_1,P_2,\ldots$; and
(d) Elementary unitary gates and elementary classical arithmetic operators.
A program in $\mathbf{RQC}^{++}$ describes a \textit{parameterised unitary} without measurements (see \Cref{sec:related_work} for discussion about the unitary restriction).
Classical variables are solely for specifying the control of programs.
For example, in \Cref{fig:q-multiplexor}, $k$ and $x$ define the formal parameters of $P(k,x)$, and are used in the if-statement and the actual parameters for procedure calls. Classical variables can also store intermediate computation results (see the syntax in \Cref{fig:syntax-QRPL}).
We use $\overline{x}=x_1\ldots x_n$ to denote a list of classical variables.
Similar notations apply to quantum variables and procedure identifiers.

Variables can be simple or array variables.
The notion of array is standard, e.g., if $x$ is a $1$-indexed one-dimensional classical array,
then $x[10]$ represents the $10^{\textup{th}}$ element in $x$.
Array variables induce subscripted variables: e.g.,
for a quantum array $q$, $q[2y+z]$ is an element in $q$
with subscription $2y+z$.
For simplicity, in this paper we only consider
one-dimensional arrays, and requires that for any classical subscripted variable $x[t]$,
the expression $t$ contains no more subscripted variables.

We also consider arrays of procedure and subscripted procedure identifiers,
for which notations are similar to that for variables.
Moreover, for any procedure identifier $P$, 
we associate with it a classical variable $P.\mathit{ent}$,
storing the entry address of the declaration of $P$.
The value of $P.\mathit{ent}$ is determined after the program is compiled and loaded into the quantum memory.

The syntax of $\mathbf{RQC}^{++}$ is summarised in \Cref{fig:syntax-QRPL}.
Here, a program is specified by $\calP$, a set of procedure declarations,
with a main procedure $P_{\textup{main}}$.
Each procedure declaration has the form $P(\overline{u})\Leftarrow C$,
where $P$ is the procedure identifier, $\overline{u}$ is a list of formal parameters (which can be empty),
and $C$ is the procedure body.
The recursion is supported by that $C$ can contains $P$ itself.
A statement $C$ is inductively defined,
where $U$ represents an elementary unitary gate and $b$ represents a classical binary expression.
We further explain as follows.
\begin{itemize}
	\item
		The procedure call $P\parens*{\overline{t}}$ has a list of classical expressions $\overline{t}$
		as its actual parameters.
	\item
		The block statement $\mathbf{begin}\ \mathbf{local}\ \overline{x}:=\overline{t};C\ \mathbf{end}$
		temporarily sets classical variables $\overline{x}$ to the values of $\overline{t}$ at the beginning of the block,
		and restores their old values at the end.
	\item
		The unitary gate $U\bracks*{\overline{q}}$ applies the elementary quantum gate $U$ on quantum variables $\overline{q}$.
	\item
		The quantum if-statement $\mathbf{qif}\bracks*{q}\parens*{\ket{0}\rightarrow C_0}\square\parens*{\ket{1}\rightarrow C_1}\mathbf{fiq}$
		executes $C_i$, conditioned on the qubit variable $q$ (a.k.a., quantum coin):
		when $q$ is in state $\ket{0}$, $C_0$ is executed;
		when $q$ is in state $\ket{1}$, $C_1$ is executed.
		Unlike the classical if-statement where the control flow only runs through one of the two branches,
		the quantum control flow run through both quantum branches created by the $\mathbf{qif}$ statement,
		in superposition.
		Note that the superposition state is held in the composite system including
		$q$ and the quantum variables in $C_0,C_1$.
\end{itemize}

\begin{figure}
	\centering
	\includegraphics[width=0.7\textwidth]{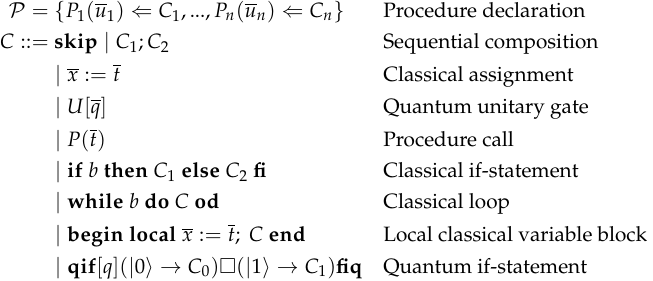}
	\caption{The syntax of quantum recursive programming language $\mathbf{RQC}^{++}$.}
	\label{fig:syntax-QRPL}
\end{figure}

\begin{figure}
	\centering
    \includegraphics[width=0.95\textwidth]{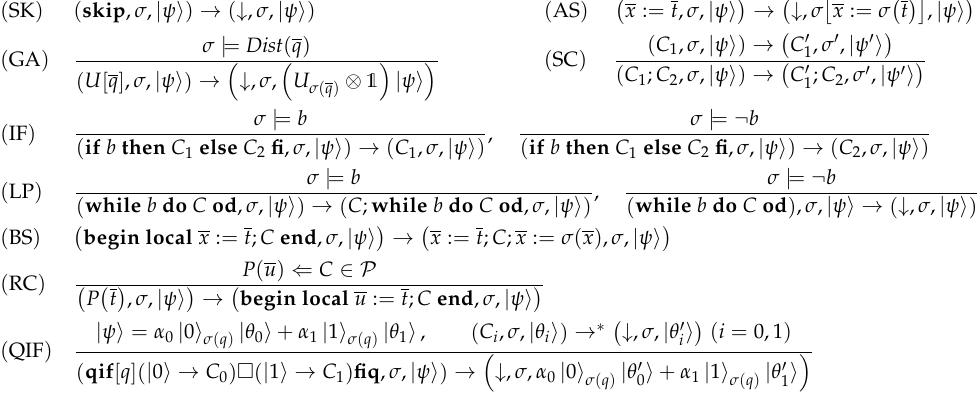}
	\caption{Transition rules for defining the operational semantics of $\mathbf{RQC}^{++}$.}
	\label{fig:semantics-QRPL}
\end{figure}

\subsection{Semantics}
\label{sub:semantics}

Now we briefly introduce the operational semantics of $\mathbf{RQC}^{++}$.
We use $\parens*{C, \sigma, \ket{\psi}}$ to denote a configuration,
where $C$ is the remaining statement to be executed or $C=\ \downarrow$ (standing for termination),
$\sigma$ is the current classical state, and $\ket{\psi}$ is the current quantum state.
The operational semantics is defined in terms of transitions between configurations of the form:
$\parens*{C,\sigma,\ket{\psi}}\rightarrow \parens*{C', \sigma',\ket{\psi'}}$.

The transition rules for defining the operational semantics of $\mathbf{RQC}^{++}$ are shown in \Cref{fig:semantics-QRPL}.
For simplicity of presentation, 
we only explain the most non-trivial (QIF) rule.
Other rules are rather standard and further explained in \Cref{sec:the_high_level_language_QRPL}.
In the (QIF) rule,
$i=0,1$ correspond to the two quantum branches,
controlled by the external quantum coin $q$.
Here, $\sigma(q)$ denotes the subsystem specified by $q$ with respect to classical state $\sigma$. 
As usual, $\rightarrow^k$ denotes the composition of $k$ copies of $\rightarrow$,
and $\rightarrow^*=\bigcup_{k=0}^{\infty}\rightarrow^k$.
The semantics of the $\mathbf{qif}$ statement is exactly a quantum multiplexor~\cite{SBM05} with one control qubit $q$:
if each $C_i$ describes a unitary $U_i$, then 
$\mathbf{qif}[q]\parens*{\ket{0}\rightarrow C_0}\square \parens*{\ket{1}\rightarrow C_1}\mathbf{fiq}$
describes the unitary $U_0\oplus U_1=\ket{0}\!\bra{0}_{\sigma(q)}\otimes U_0+\ket{1}\!\bra{1}_{\sigma(q)}\otimes U_1$.

Note that in the (QIF) rule, $\parens*{C_i,\sigma,\ket{\theta_i}}$ are required to terminate in the same classical state $\sigma$
for both branches ($i=0,1$) to \textit{prevent classical variables from being in superposition}.\footnote{For simplicity of later implementation, this requirement has been made slightly stricter than the original one (``terminating in the same $\sigma'$'', where $\sigma'$ may differ from the initial $\sigma$) described in \cite{YZ24}, but remains easy to meet in practice.}
This requirement seems inevitable to separate classical and quantum variables in the presence of quantum control flow.
As a result, only local classical variables can be arbitrarily modified in the $\mathbf{qif}$ statement. If one wishes to return different data from two quantum branches, then the data becomes intrinsically quantum and should therefore be stored in quantum variables.


\subsection{Conditions for Well-Defined Semantics}
\label{sub:conditions_for_well_defined_semantics}

We present three conditions for a program in $\mathbf{RQC}^{++}$ to have well-defined semantics, in particular,
for the (QIF) rule to be properly and easily applied.
The first condition guarantees that in every $\mathbf{qif}$ statement, 
$q$ is external to $C_0$ and $C_1$.
This is introduced for the $\mathbf{qif}$ statement to be physically meaningful.
We use $\mathit{qv}\parens*{C,\sigma}$ to denote the quantum variables in statement $C$ with respect to a given classical state $\sigma$. Its precise definition is given in \Cref{sub:details_semantics}.

\begin{condition}[External quantum coin]
	\label{cnd:qif-external-formal}
	For any procedure declaration $P(\overline{u})\Leftarrow C\in \calP$,
	and any $\mathbf{qif}[q]\parens*{\ket{0}\rightarrow C_0}\square\parens*{\ket{1}\rightarrow C_1}\mathbf{fiq}$
	appearing in $C$,
	and any classical state $\sigma$ (of concern),
	$q\notin \mathit{qv}\parens*{C_0,\sigma}\cup \mathit{qv}\parens*{C_1,\sigma}$.
\end{condition}

The second condition says that in every $\mathbf{qif}$ statement,
both $C_0$ and $C_1$ contain no free changed (classical) variables.
A classical variable is \textit{free} if it is not declared as local variable.
It is \textit{changed} if it appears on the LHS of an assignment.
We use $\mathit{fcv}\parens*{C,\sigma}$ to denote the free changed variables in $C$ with respect to $\sigma$. See \Cref{sub:details_semantics} for its precise definition.
This condition is introduced as the (QIF) rule requires $\parens*{C_i,\sigma,\ket{\psi}}$ to terminate
in the same classical state $\sigma$ for both branches $i=0,1$.

\begin{condition}[No free changed variables in $\mathbf{qif}$ statements]
	\label{cnd:no-fcv-qif}
	For any $P\parens*{\overline{u}}\Leftarrow C \in \calP$,
	any $\mathbf{qif}[q]\parens*{\ket{0}\rightarrow C_0}\square\parens*{\ket{1}\rightarrow C_1}\mathbf{fiq}$
	appearing in $C$,
	and any classical state $\sigma$ (of concern),
	$\mathit{fcv}\parens*{C_0,\sigma}=\mathit{fcv}\parens*{C_1,\sigma}=\emptyset$.
\end{condition}
%

The third condition says that every procedure body contains no free changed variables.
This condition is introduced to simplify the process of compilation,
as it allows the procedure calls to be arbitrarily used together with the $\mathbf{qif}$ statements
without violating \Cref{cnd:no-fcv-qif}.

\begin{condition}[No free changed variables in procedure bodies]
	\label{cnd:no-fcv-proc-body-formal}
	For any $P\parens*{\overline{u}}\Leftarrow C\in \calP$
	and any classical state $\sigma$ (of concern),
	$\mathit{fcv}\parens*{C,\sigma}=\emptyset$.
\end{condition}

\section{Quantum Register Machine}
\label{sec:quantum_register_machine}

Now we start to consider how to implement quantum recursive programs defined in the previous section. As the basis, let us    
introduce the notion of quantum register machine,
an architecture that provides instruction-level support for quantum control flow and recursive procedure calls at the same time.
Unlike most existing quantum architectures that use classical controllers to implement quantum circuits,
the quantum register machine stores quantum programs and data in 
a quantum random access memory (QRAM) and executes on quantum registers.
As aforementioned in \Cref{sec:introduction}, since existing quantum hardware is typically controlled by classical pulses, it would be better to think quantum register machine as an abstract machine (that does not require hardware-level quantum control flow).
Like a classical CPU, 
the machine works by repeatedly applying a fixed unitary $U_{\textup{cyc}}$ (independent of the program)
per instruction cycle, which consists of several stages, 
including fetching an instruction from the QRAM,
decoding it and executing it by performing corresponding operations.
To support quantum control flow, additional stages related to the partial evaluation are also needed.
The unitary $U_{\textup{cyc}}$ will be eventually implemented by standard quantum circuits, as described in \Cref{sec:execution} and visualised in \Cref{sub:details_uni_cyc_and_exe}.

In the following, we first explain quantum registers and QRAM, and 
then describe a low-level instruction set $\mathbf{QINS}$ (quantum instructions) for the quantum register machine.

\subsection{Quantum Registers}
\label{sub:quantum_registers}

The quantum register machine has a constant number of quantum registers (or simply, registers),
each storing a quantum word composed of $L_{\textup{word}}$ (called word length) qubits.
Registers are directly accessible.
The machine can perform a series of elementary operations on registers,
including word-level arithmetic operations (see also \Cref{sub:the_low_level_language_qins}),
each assumed to take time $T_{\textup{reg}}$.
The precise definition of elementary operations are presented in \Cref{sub:elementary_operations_on_registers}.

Registers are grouped into two types: system and user registers.
There are eight system registers. The first five 
are rather standard and borrowed from the classical reversible architectures~\cite{Vieri99,Frank99,AGY07,TAG12},
as quantum unitaries are intrinsically reversible.
We describe their classical effects as follows.
\begin{itemize}
	\item 
		Program counter  $\mathit{pc}$ records the address of the current instruction.
	\item 
		Instruction $\mathit{ins}$ records the current instruction.
	\item
		Branching offset $\mathit{br}$ records the offset of the address of the next instruction to go
		from $\mathit{pc}$.
		More specifically, if $\mathit{br}=0$, 
		then the address of the next instruction will be $\mathit{pc}+1$.
		Otherwise, the address of the next instruction will be $\mathit{pc}+\mathit{br}$.
	\item
		Return offset $\mathit{ro}$ records the offset for $\mathit{br}$ in the return of a procedure call.
	\item
		Stack pointer $\mathit{sp}$ records the current topmost location of the call stack.
\end{itemize}
In contrast, the last three system registers are novelly introduced
to support an efficient implementation of the $\mathbf{qif}$ statements.
They are related to the \textit{qif table}, a data structure
generated by the partial evaluation of quantum control flow
and used during execution to address the aforementioned synchronisation problem.
We briefly describe their classical effects as follows,
and will explain further details in \Cref{sec:partial-evaluation,sec:execution}.
\begin{itemize}
	\item
		Qif table pointer $\mathit{qifv}$ records the current node in the qif table.
	\item
		Qif wait counter $\mathit{qifw}$ records the number of instruction cycles to wait at the current node in the qif table.
	\item
		Qif wait flag $\mathit{wait}$ records whether the current instruction cycle needs to be skipped.
\end{itemize}

We also set the initial values of theses registers:
$\mathit{pc}$, $\mathit{sp}$ and $\mathit{qifv}$ are initialised to $\ket{j}$,
where $j$ is the starting addresses of the main program, the call stack and the qif table, respectively.
Other system and user registers are initialised to $\ket{0}$.

\subsection{Quantum Random Access Memory}
\label{sub:quantum_random_access_memory}

The quantum register machine has a quantum random access memory (QRAM)\footnote{
In particular, the QRAM considered here is quantum random access quantum memory (QRAQM). Readers are referred to~\cite{JR23} for a review of QRAM.}
composed of $N_{\textup{QRAM}}$ memory locations,
each storing a quantum word.
The QRAM is not directly accessible.
Like a classical memory, access to QRAM is by providing an address register
specifying the address,
and a target register to hold the information retrieved from the specified location.
Unlike the classical case, the address register can be in quantum superposition,
and registers can be entangled with the QRAM.
In this paper, we assume the following two types of elementary QRAM accesses.

\begin{definition}[Elementary QRAM accesses]
	\label{def:access_QRAM}
	\qquad
	\begin{itemize}
		\item 
			QRAM (swap) load. This access performs the unitary $U_{\textup{ld}}(r,a,\mathit{mem})$
			defined by the mapping:
			\begin{equation}
				\ket{x}_{r}\ket{i}_{a}\ket{M}_{\mathit{mem}}\mapsto \ket{M_i}_{r}\ket{i}_{a}\ket{M_0,\ldots, M_{i-1},x,M_{i+1},\ldots,M_{N_{\textup{QRAM}}-1}}_{\mathit{mem}},\label{eq:QRAM-load}
			\end{equation}
			for all $x$, $i$ and $M=(M_0, \ldots, M_{N_{\textup{QRAM}}-1})$.
			Here, $r$ is the target register,
			$a$ is the address register,
			and $\mathit{mem}$ is the QRAM.
		\item
			QRAM (xor) fetch. This access performs the unitary $U_{\textup{fet}}(r,a,\mathit{mem})$
			defined by the mapping:
			\begin{equation}
				\ket{x}_{r}\ket{i}_{a}\ket{M}_{\mathit{mem}}\mapsto \ket{x\oplus M_i}_{r}\ket{i}_{a}\ket{M}_{\mathit{mem}},
				\label{eq:QRAM-fetch}
			\end{equation}
			for all $x,i,M$.
	\end{itemize}
	Moreover, the controlled versions (controlled by a register) of elementary QRAM accesses are also considered elementary, since the number of registers is constant and the control only incurs a constant overhead. Suppose every elementary QRAM access takes time $T_{\textup{QRAM}}$.
\end{definition}

Some readers might notice that the physical realisation of QRAM is not yet near-term, a challenge shared by most works leveraging QRAM (e.g., \cite{Amb07,BJLM13,ACL+20,YC22}). Nevertheless, there are ongoing efforts towards feasible QRAM implementations (e.g., \cite{HZZCSGJ19,HLGJ21,XLD25}). Importantly, the final complexity of our implementation of quantum recursive programs is measured by the standard circuit depth and unaffected by the near-term feasibility of QRAM.

It is also worth pointing out that managing entanglement between registers and the QRAM is crucial, as improper handling can result in incorrect output states~\cite{LWS+23}. 
To this end, the instruction set $\mathbf{QINS}$ (\Cref{sub:the_low_level_language_qins}) and the compilation process (\Cref{sec:compilation}) are carefully designed to ensure that, after the execution, quantum variables are \textit{disentangled} from other registers and the remaining part of the QRAM. A key of the design is \textit{proper uncomputation} of intermediate results. The idea traces back to~\cite{Landauer61,Bennett73}, and has been applied in~\cite{Frank99,Vieri99,AGY07,TAG12,YVC24}.
Moreover, in our design, the creation and removal of entanglement during the execution align with the program structure (in $\mathbf{RQC}^{++}$).

\subsubsection{Layout of the QRAM}
\label{sub:QRAM-layout}

The QRAM in the quantum register machine stores both programs and data.
In particular, it contains the following sections.
\begin{enumerate}
	\item 
		Program section stores the compiled program in a  low-level language $\mathbf{QINS}$ (to be defined in \Cref{sub:the_low_level_language_qins}).
	\item
		Symbol table section stores the name of every variable and its corresponding address.
		Here, unlike in the classical case, the symbol table is used at runtime instead of compile time (see also \Cref{sub:symbol_table_and_memory_allocation_of_variables}),
		because arrays in $\mathbf{RQC}^{++}$ are not declared with fixed size.
	\item
		Variable section stores the classical and quantum variables.
	\item
		Qif table section stores the qif table (to be defined in \Cref{sub:qif_table}).
	\item
		Stack section stores the call stack to handle the procedure calls.
		The stack is composed of multiple stack frames,
		each storing the actual parameters and return offset (from the caller to the callee),
		and the local data used by the callee, in a procedure call.
\end{enumerate}

We visualise the quantum register machine and the layout of the QRAM in~\Cref{fig:QRM}.

\begin{figure}
	\centering
	\includegraphics[width=0.8\textwidth]{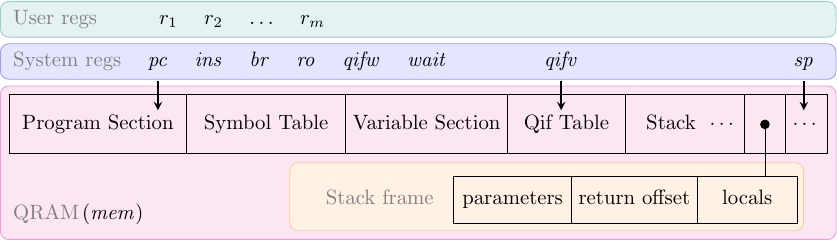}
	\caption{Storage components of the quantum register machine and the layout of the QRAM.
	All components can be together in a quantum superposition state.
	}
	\label{fig:QRM}
\end{figure}

\subsection{The Low-Level Language QINS}
\label{sub:the_low_level_language_qins}

Now we present $\mathbf{QINS}$,
an instruction set for describing the compiled programs.
Each instruction specifies a series of elementary operations to be carried out 
by the quantum register machine.
There are 22 instructions in $\mathbf{QINS}$,
which are listed with their classical effects in \Cref{fig:table-ins}.
Here, we leave the explanation of instructions \hlt{qif} and \hlt{fiq}
to \Cref{sec:execution}.
The classical effects of other instructions are lifted to quantum 
in the standard way
when being executed by the quantum register machine.

\begin{figure}
	\centering
	\begin{subfigure}[t]{0.45\textwidth}
		\centering
		\includegraphics[width=\textwidth]{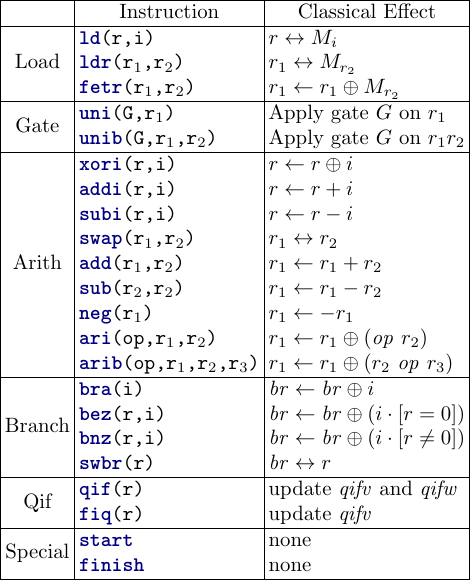}
		\caption{Instructions and corresponding classical effects.
                Here, $\oplus$ denotes the XOR operator; 
                $[b]=1$ if $b$ is true and $[b]=0$ otherwise.}
		\label{fig:table-ins}
	\end{subfigure}
	~
	\begin{subfigure}[t]{0.33\textwidth}
		\centering
		\includegraphics[width=\textwidth]{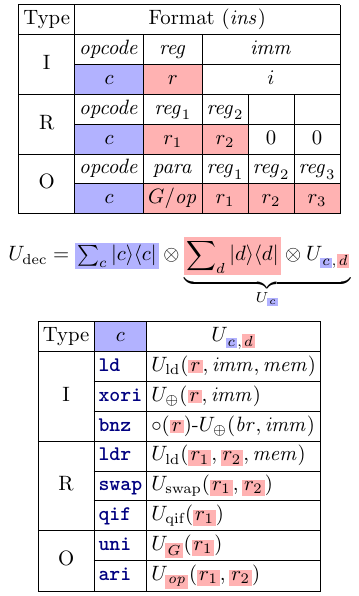}
		\caption{Instruction formats, the decoding unitary $U_{\textup{dec}}$, and selected examples of instruction implementations.}
		\label{fig:table-type-eg}
	\end{subfigure}
	\caption{The low-level language $\mathbf{QINS}$ and selected examples.}
	\label{fig:QINS}
\end{figure}

The design of $\mathbf{QINS}$ is inspired by 
the existing classical reversible instruction sets~\cite{Frank99,Vieri99,AGY07,TAG12}.
Nevertheless, several  instructions in $\mathbf{QINS}$ are essentially new.
The most important are instructions \hlt{qif} and \hlt{fiq},
which are designed for a \textit{structured management} of quantum control flow 
(generated by the $\mathbf{qif}$ statements in $\mathbf{RQC}^{++}$),
in particular, aiding the partial evaluation and execution.
Instructions \hlt{uni} and \hlt{unib} are designed for quantum unitary gates.

We group the instructions into three types:
I (immediate-type), R (register-type) and O (other-type),
according to their formats, as shown in \Cref{fig:table-type-eg}.
During the execution (to be described in \Cref{sec:execution}),
we decode the instruction in register $\mathit{ins}$ by performing a unitary $U_{\textup{dec}}$ (see \Cref{fig:table-type-eg}).
$U_{\textup{dec}}$ is a quantum multiplexor,
with section $\mathit{opcode}$ as its first part of control,
and other sections (depending on the type I/R/O) in $\mathit{ins}$ as its second part of control.
Let $c$ be a computational basis in the first part
and $d$ in the second part,
then the unitary being controlled is denoted by $U_{c,d}$.

For illustration, selected instructions and corresponding $U_{c,d}$ are presented in \Cref{fig:table-type-eg}.
Here, $U_{\textup{ld}}$ is the QRAM access in \Cref{def:access_QRAM}.
$U_{\textup{qif}}$ (and similarly $U_{\textup{fiq}}$ for \hlt{fiq}) will be defined in \Cref{sec:execution}.
Other unitaries are elementary operations on registers:
(a) $U_{\oplus}$ performs the mapping $\ket{x}\ket{y}\mapsto\ket{x\oplus y}\ket{y}$;
(b) $\circ\parens*{r}$-$U$ stands for the controlled version $\ket{0}\!\bra{0}_r\otimes U+\sum_{x\neq 0}\ket{x}\!\bra{x}_r\otimes \Id$ of unitary $U$;
(c) $U_{\textup{swap}}$ performs the mapping $\ket{x}\ket{y}\mapsto \ket{y}\ket{x}$;
(d) $U_G$ applies the elementary gate $G$ (chosen from a fixed set $\calG$ of size $O(1)$);
(e) $U_{\mathit{op}}$ performs the mapping $\ket{x}\ket{y}\mapsto \ket{x\oplus \parens*{\mathit{op}\ y}}\ket{y}$ for unary operator $\mathit{op}$ (chosen from a fixed set $\calO\calP$ of size $O(1)$).

Further details of $\mathbf{QINS}$ are provided in \Cref{sub:details_qins}.

\section{Compilation}
\label{sec:compilation}

\begin{wrapfigure}{R}{0.5\textwidth}
	\centering
	\includegraphics[width=0.5\textwidth]{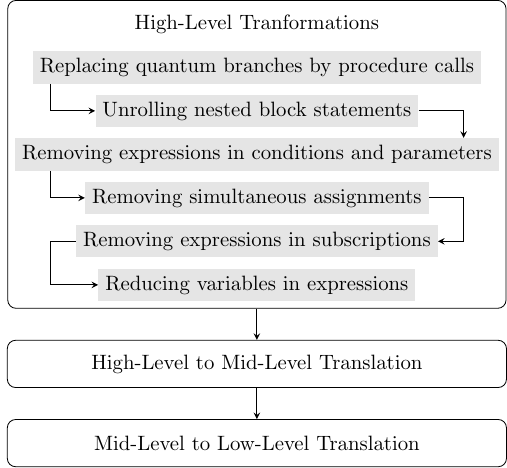}
	\caption{The compilation process.}
	\label{fig:compdiagram}
\end{wrapfigure}

As usual, the first step in the implementation of quantum recursive programs is their compilation. The compilation of a program $\calP$ in $\mathbf{RQC}^{++}$
consists of the following passes.
\begin{enumerate}
	\item 
		First, a series of high-level transformations are performed on the original $\calP$ to obtain $\calP_h$,
		which simplify the program structure and make it easier to be further compiled. 
	\item
		Then, the transformed program $\calP_h$ is translated into an intermediate program $\calP_{m}$ in the mid-level language
		composed of instructions similar to those in $\mathbf{QINS}$ but more flexible.
	\item
		Finally, the mid-level $\calP_m$ is translated into a program $\calP_l$ in the low-level language $\mathbf{QINS}$.
\end{enumerate}

The compilation process is visualised in \Cref{fig:compdiagram}.
The remainder of this section is devoted to describe these passes carefully.  
In the sequel, 
we always assume the source program $\calP$ to be compiled 
satisfies \Cref{cnd:qif-external-formal,cnd:no-fcv-qif,cnd:no-fcv-proc-body-formal} in \Cref{sub:conditions_for_well_defined_semantics}
and \textit{has well-defined semantics}.
We will not bother checking the syntax and semantics of $\calP$.

\subsection{High-Level Transformations}
\label{sub:high-level-transformations}
In this section, we describe the first pass of high-level transformation from $\calP$ to $\calP_h$. 
The major target of this pass is to simplify the 
automatic uncomputation of classical variables in later passes.

\begin{wrapfigure}{R}{0.23\textwidth}
	\centering
	\includegraphics[width=0.22\textwidth]{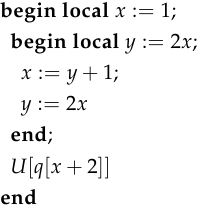}
	\caption{An example of nested block statement.}
	\label{fig:eg-nested-block}
\end{wrapfigure}

A program $\calP$ in $\mathbf{RQC}^{++}$ may contain irreversible classical statements,
e.g., assignment $x:=1$.
Reversibly implementing these statements
introduces garbage data,
e.g., through the standard Landauer~\cite{Landauer61} and Bennett~\cite{Bennett73} methods.
For the overall correctness of the quantum computation,
these garbage data should be properly uncomputed. 
Moreover, the block statement in $\mathbf{RQC}^{++}$ explicitly requires uncomputation of local variables
at the end of the block.

In the execution of a program, when should we perform uncomputation?
First, we realise a difficulty from the uncomputation of local variables
in nested block statements.
Consider the example in \Cref{fig:eg-nested-block}.
The inner block modifies $x$,
which is used by the outer block.
If one tries to uncompute the local variable $y$ at the end of the inner block,
the change on $x$ (by the inner block) is also uncomputed,
which is an undesirable side effect.

To overcome this difficulty,
we will perform a series of transformations on the source program $\calP$,
such that the transformed $\calP_h$ no longer contains nested block statements.
Along the way, we also simplify the structure of the program.
Consequently, for $\calP_h$,
we only need to perform uncomputation at the end of every procedure body (of procedure declarations),
which will be automatically done in the high-level to mid-level translation (in \Cref{sub:high_level_to_mid_level_translation}).

An overview of high-level transformations is already shown in \Cref{fig:compdiagram}.
In the following we only select the first two steps for explanation,
while other steps are rather standard (see e.g., the textbook~\cite{AMSJ07})
and presented in \Cref{sub:details_high_level_transformations}.

\subsubsection{Replacing Quantum Branches by Procedure Calls}
\label{sub:replace_qif_proc}

In this step, we replace the program in every quantum branch of every $\mathbf{qif}$ statement
by a procedure call.
More specifically, for every $\mathbf{qif}$ statement, if $C_0,C_1$ are not procedure identifiers or $\mathbf{skip}$ statements, 
then we introduce fresh procedure identifiers $P_0,P_1$, perform the replacement:
\begin{equation*}
	\mathbf{qif}[q](\ket{0}\rightarrow C_0)\square (\ket{1}\rightarrow C_1)\mathbf{fiq}
	\qquad\Rightarrow\qquad
	\mathbf{qif}[q](\ket{0}\rightarrow P_0)\square (\ket{1}\rightarrow P_1)\mathbf{fiq},
\end{equation*} 
and add new procedure declarations $P_i\Leftarrow C_i$ (for $i\in \braces*{0,1}$) to $\calP$.
If only one of $C_0,C_1$ is procedure identifier or $\mathbf{skip}$,
then the replacement is performed only for the other branch.
It is easy to see the above transformation does not violate  \Cref{cnd:qif-external-formal,cnd:no-fcv-qif,cnd:no-fcv-proc-body-formal}.

\subsubsection{Unrolling Nested Block Statements}
\label{sub:unroll_nested_blocks}

In this step, we unroll all nested block statements.
The program after this step is promised to no more contain block statements,
but uncomputation of classical variables needs to be done at the end of every procedure body
when the program is implemented,
for it to preserve its original semantics.
To do this, 
for any $P\parens*{\overline{u}}\Leftarrow C'\in \calP$
and every block statement appearing in $C'$,
we perform the replacement:
\begin{equation*}
	\mathbf{begin}\ \mathbf{local}\ \overline{x}:=\overline{t};\ C\ \mathbf{end}
	\qquad
	\Rightarrow
	\qquad
	\overline{x'}:=\overline{t};
	C\bracks*{x'/x},
\end{equation*}
where $\overline{x'}$ is a list of fresh variables, and $\bracks*{x'/x}$ stands for replacing variable $x$ by $x'$.
Also, we append $\overline{x'}:=\overline{0}$ 
at the beginning of $C'$.
The above transformation keeps \Cref{cnd:qif-external-formal,cnd:no-fcv-qif,cnd:no-fcv-proc-body-formal} too.

Note that after this step,
the program is technically in some new language with the same syntax as $\mathbf{RQC}^{++}$,
but whose semantics requires the uncomputation of classical variables at the end of every procedure body.

\subsubsection{After the High-Level Transformations}

We observe that the program $\calP_h=\braces*{P\parens*{\overline{u}}\Leftarrow C}_P$ after the high-level transformations in \Cref{fig:compdiagram}
has the following simplified syntax:
\begin{equation*}
	\begin{split}
		C::= &\ \mathbf{skip} \mid x:=t\mid U\bracks*{\overline{q}}\mid C_0; C_1\mid P(\overline{x})\mid \mathbf{if}\ x\ \mathbf{then}\ C_0\ \mathbf{else}\ C_1\ \mathbf{fi}\mid \mathbf{while}\ x\ \mathbf{do}\ C\ \mathbf{od}\\
				 &\mid \mathbf{qif}\bracks*{q}(\ket{0}\rightarrow P_0)\square (\ket{1}\rightarrow P_1)\mathbf{fiq},
	\end{split}
\end{equation*}
where every subscripted variable and procedure identifier has a basic classical variable as its subscription (e.g., $x[y]$),
and every expression has the form $t\equiv \mathit{op}\ x$ or $t\equiv x\ \mathit{op}\ y$.
As aforementioned, the semantics is slightly changed:
uncomputation of classical variables are needed at the end of every procedure body in the implementation,
(which will be automatically done in the high-level to mid-level translation in \Cref{sub:high_level_to_mid_level_translation}).

For illustration, on the LHS of \Cref{fig:q-multiplexor-after},
we show an after-the-high-level-transformations version  of the quantum multiplexor program in \Cref{fig:q-multiplexor}.

\begin{figure}
	\centering
	\includegraphics[width=\textwidth]{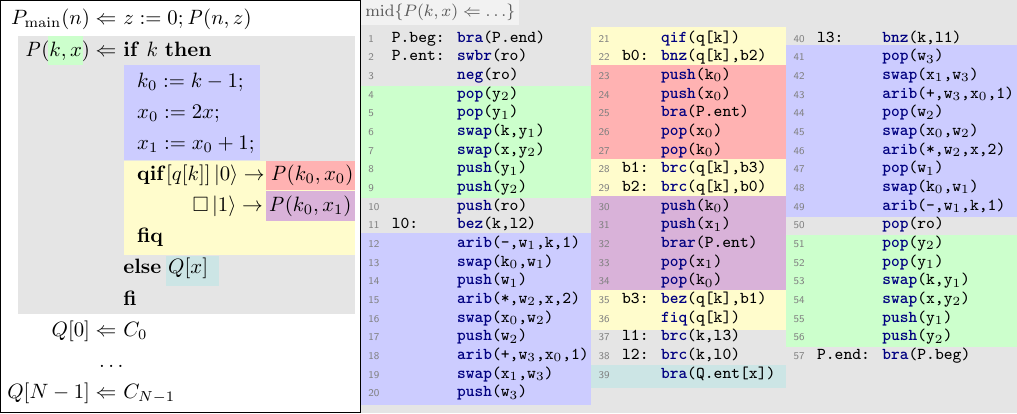}
	\caption{Example of high-level transformations and high-to-mid-level translation. 
		The original program is the quantum multiplexor program in \Cref{fig:q-multiplexor}.
		Here, on the LHS is the program after the high-level transformations. 
	On the RHS is the high-to-mid-level translation of the procedure $P(k,x)$.
	Their connections are highlighted in colors.
    Note that new variables $x_0,x_1,k_0$ are introduced by the rather standard third step of high-level transformations (see also \Cref{sub:details_high_level_transformations}), and
some transformations have no effect on this example.}
	\label{fig:q-multiplexor-after}
\end{figure}


\subsection{High-Level to Mid-Level Translation}
\label{sub:high_level_to_mid_level_translation}

Now we translate the transformed high-level program $\calP_h$ obtained in the previous subsection into $\calP_m$ in a mid-level language,
which is different from the low-level language $\mathbf{QINS}$ (defined in \Cref{sub:the_low_level_language_qins}) in the following aspects:
\begin{itemize}
	\item
		We do not consider the memory allocation.
		Thus, instructions \hlt{ld}, \hlt{ldr} and \hlt{fetr} are not needed at this stage.
	\item
		Beyond registers and numbers,
		instructions can also take variables and labels as input.
		Here, like in the classical assembly language, a label is an identifier for the address of an instruction.
		(When the program is further translated into $\mathbf{QINS}$, in the next section,
		every label $l$ will be replaced by the offset of the address of where $l$ is defined
		from the address of where $l$ is used.)

	\item 
		We have additional instructions \hlt{push} and \hlt{pop}
		for stack operations.
        Also, an additional branching instruction \hlt{brc} will be used 
        in pair with \hlt{bez} (or \hlt{bnz}).
        In particular, \hlt{brc}\texttt{(x,l)}, compared to \hlt{bra}\texttt{(l)},
        has the additional information of some variable $x$.
\end{itemize}
The high-to-mid-level translation also automatically handles the initialisation of formal parameters
and the uncomputation of classical variables at the end of procedure bodies (see \Cref{sub:high-level-transformations}). 

Let us use $\mathrm{mid}\braces*{D}$ to denote the high-to-mid-level translation 
of a statement (or declaration) $D$ in $\mathbf{RQC}^{++}$.
In \Cref{fig:HightoMid}, we present selected examples of the high-to-mid-level translation,
and more details are shown in \Cref{sub:further_details_of_high_level_to_mid_level_translation}.
Here, $\mathrm{init}\braces*{\cdot}$ and $\mathrm{uncp}\braces*{\cdot}$ denote the initialisation of formal parameters 
and uncomputation of classical variables, respectively.
We further explain them as follows.
\begin{itemize}
	\item
		To reversibly implement $\mathbf{while}\ x\ \mathbf{do}\ C\ \mathbf{od}$,
		a fresh variable $y$ is introduced to count the number of loops.
		Similar to the classical reversible architectures~\cite{Vieri99,Frank99,AGY07,TAG12}, we use a pair of branching instructions
		(e.g., \hlt{bnz} and \hlt{brc}) to realise reversible (conditional) branching.
	\item
		In the translation of $\mathbf{qif}[q]\parens*{\ket{0}\rightarrow C_0}\square \parens*{\ket{1}\rightarrow C_1}\mathbf{fiq}$,
		we have a pair of instructions \hlt{qif}\texttt{(q)} and \hlt{fiq}\texttt{(q)},
		which indicate the creation and join of quantum branching, respectively.
		They will be used in the partial evaluation of quantum control flow and the final execution.
	\item
		The translations of procedure call $P\parens*{\overline{x}}$ and declaration $P\parens*{\overline{u}}\Leftarrow C$
		are inspired by their counterparts in classical reversible computing~\cite{Axelsen11}.
		Here, the biggest difference is our design of the \textit{automatic uncomputation} of classical variables,
		performed by the program $\mathrm{uncp}\braces*{C}$
		at the end of procedure body $C$ (see also \Cref{sub:unroll_nested_blocks}),
		which reverses the changes on classical variables in $C$.
		The uncomputation $\mathrm{uncp}\braces*{\cdot}$ is also recursively defined,
		where $\mathrm{uncp}\braces*{P\parens*{\overline{x}}}$ and $\mathrm{uncp}\braces*{\mathbf{qif}\ldots\mathbf{fiq}}$
		are set to empty, due to \Cref{cnd:no-fcv-qif,cnd:no-fcv-proc-body-formal}.
\end{itemize}

\begin{figure}
	\centering
	\includegraphics[width=\textwidth]{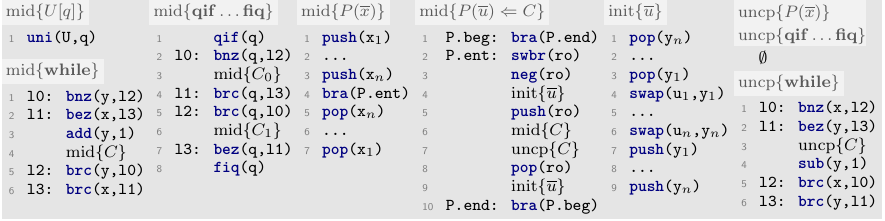}
	\caption{Selected examples of the high-to-mid-level translation.
	Here, $\mathrm{init}\braces*{\cdot}$ and $\mathrm{uncp}\braces*{\cdot}$
	stand for the initialisation of formal parameters and uncomputation of classical variables,
    respectively.
	$y,y_1,\ldots, y_n$ are all fresh variables.
    Also, $\mathrm{mid}\braces*{\mathbf{while}}$ and $\mathrm{uncp}\braces*{\mathbf{while}}$ use the same fresh variable $y$.}
	\label{fig:HightoMid}
\end{figure}

For illustration, on the RHS of \Cref{fig:q-multiplexor-after},
we present an example of the high-to-mid-level translation
of the quantum multiplexor program.
For simplicity, we only show the translation of the recursive procedure $P(k,x)$.
The full translation can be found in \Cref{sub:further_details_of_high_level_to_mid_level_translation}. 

%
%
%

\subsection{Mid-Level to Low-Level Translation}
\label{sub:mid_level_to_low_level_translation}

Now we are ready to describe the last pass in which the mid-level program $\calP_m$ obtained in the previous sections is translated 
into a program $\calP_l$ in the low-level language $\mathbf{QINS}$ and thus executable on the quantum register machine. In this pass, 
instructions that take variables and labels as inputs will be translated 
to instructions that only take registers and immediate numbers as inputs.
The additional instructions \hlt{push}, \hlt{pop} and \hlt{brc} also need to be translated.
To do this, we need the load instructions \hlt{ld}, \hlt{ldr} and \hlt{fetr}.
Let us use $\mathrm{low}\braces*{i}$ to denote the mid-to-low-level translation 
of an instruction $i$.
In \Cref{fig:MidtoLow}, we present selected examples of the mid-to-low-level translation,
and more details are shown in \Cref{sub:further_details_mid_low_translation}.
We further explain them as follows.
\begin{itemize}
	\item 
		The translation of \hlt{uni}\texttt{(U,q[x])}
		shows how to handle inputs containing subscripted quantum variables.
		We use $@ x$ to denote the address of the name $x$
		(in the symbol table section of the QRAM; see \Cref{sub:QRAM-layout}).
		The word at $@ x$ stores the address $\& x$ of the variable $x$
		(in the variable section).
		Lines 1--2 load the value of $x$
		into free register $r_2$.
		To obtain the address of $q[x]$,
		we add the address $\& q=\& \parens*{q[0]}$ and the value of $x$,
		in Lines 3--4.
		Line 5 loads the value of $q[x]$ into free register $r_4$,
		on which the instruction \hlt{uni}\texttt{(U,r$_4$)} is executed.
		Lines 7--11 reverse the effects of Lines 1-5.
        Further details of the symbol table and memory allocation of variables can be found in \Cref{sub:symbol_table_and_memory_allocation_of_variables}.
	\item
		In the translation of \hlt{bra}\texttt{(P.ent)},
		recall that the classical variable $P.\mathit{ent}$
		corresponds to some procedure identifier $P$.
        Here, Lines 1--3 loads the value of $P.\mathit{ent}$ into free register $r_2$.
        Note that Line 2 uses \hlt{fetr} instead of \hlt{ldr} to preserve the copy of $P.\mathit{ent}$ in the QRAM for recursive procedure calls.
        Lines 4--5 calculate in $r_2$ the offset of $P.\mathit{ent}$ from the address of Line 6.
        When the branching occurs after Line 6, note that registers $r_1$ and $r_2$ are cleared.
        Lines 7--12 are similar. 
	\item
		The translations of \hlt{push} and \hlt{pop} are rather simple.
		Note that they are reversible, e.g., if an element is pushed into the stack,
		the original register $r$ will be cleared.
    \item 
        The translations of \hlt{bez} and \hlt{brc} are related
        when they are used in pairs: they use the same free registers $r_1$ and $r_2$. 
\end{itemize}

\begin{figure}
	\centering
	\includegraphics[width=0.68\textwidth]{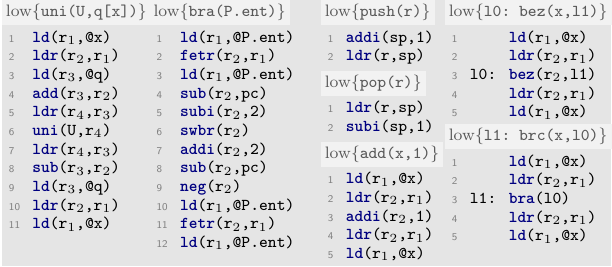}
	\caption{Selected examples of the mid-to-low-level translation.
	Here, all registers $r_i$ are free registers.}
	\label{fig:MidtoLow}
\end{figure}

To end the compilation,
we need to replace every label $l$ in the compiled program
by the offset of the address of where $l$ is defined from
where $l$ is used.
The compiled program is not yet loaded into the QRAM,
but stored classically for later partial evaluation in \Cref{sec:partial-evaluation}.
We present the example of mid-to-low-level translation of the quantum multiplexor program in \Cref{sub:further_details_mid_low_translation}.

\section{Partial Evaluation of Quantum Control Flow}
\label{sec:partial-evaluation}

At the end of the last section, 
a compiled program
in the low-level language $\mathbf{QINS}$ is obtained.
For its execution on the quantum register machine, 
we need to first perform a partial evaluation of quantum control flow to generate a data structure called qif table, 
to be loaded into the QRAM. In this section, we carefully describe this partial evaluation.

\subsection{The Synchronisation Problem}
\label{sub:the_synchronisation_problem}

Programs with only classical control flow
can be straightforwardly executed without partial evaluation.
However, for programs with quantum control flow,
there is an obstruction known as the synchronisation problem~\cite{Myers97,BV93,Ozawa98,LP98,Ozawa98b,Shi02,MO05,WY23,YVC24} 
(see further discussion in \Cref{sub:further_synchronisation_problem}).
In our case, it means in executing
the statement
$\mathbf{qif}\bracks*{q}\parens*{\ket{0}\rightarrow C_0}\square \parens*{\ket{1}\rightarrow C_1}\mathbf{fiq}$,
$C_0$ and $C_1$ can take different numbers of instruction cycles to terminate.
Consequently, the arrival times of two control flows (corresponding to two branches, in superposition) at the $\mathbf{fiq}$ are asynchronous,
and hence they cannot be correctly merged into one control flow, in the same cycle.
Another way to view the synchronisation problem 
is from the (QIF) rule in \Cref{fig:semantics-QRPL}.
The problem occurs when
$\parens*{C_0,\sigma,\ket{\theta_0}}\rightarrow^{k_0} \parens*{\downarrow, \sigma',\ket{\theta_0'}}$
and $\parens*{C_1,\sigma,\ket{\theta_1}}\rightarrow^{k_1} \parens*{\downarrow,\sigma',\ket{\theta_1'}}$
for some $k_0\neq k_1$.

The synchronisation problem becomes more complicated for
general quantum recursive programs.
Note that $C_0$ and $C_1$ can further contain quantum recursion,
and the number of nested procedure calls involved cannot be determined before hand.
The program might not even terminate.
How to deal with the probably unbounded quantum recursion?

Our solution is by partial evaluation of the quantum control flow.
When the classical inputs are given (while the quantum inputs remained unknown),
we can check whether the  compiled program $\calP$
terminates in some practical (manually set) running time $T_{\textup{prac}}\parens*{\calP}$.
If the program terminates in $T_{\textup{prac}}\parens*{\calP}$ cycles,
for every $\mathbf{qif}$ statement,
we can count the number of cycles for executing 
$C_0$ and $C_1$,
as well as determine the structure of nested quantum branching induced by nested procedure calls.
These can be gathered into
a classical data structure called \textit{qif table},
which will be used later in quantum superposition at runtime
to synchronise two quantum branches in every $\mathbf{qif}$ statement.
Note that this process is only dependent on the classical inputs but \textit{independent} of the quantum inputs, and it does not change the static program text (compared to \cite{YVC24}).
Also, our partial evaluation is different from those (e.g., \cite{JGS93,LVHPWH22}) that aim at optimising the programs.

Along with generating the qif table, 
given the classical inputs, we can also determine the sizes of all arrays
and allocate the addresses for variables (including determining the symbol table).
This task is simple and we will not describe its details.

\subsection{Qif Table}
\label{sub:qif_table}

Now let us introduce the notion of qif table,
storing the history information of quantum branching 
for an execution of the compiled program $\calP$, 
within a given practical running time $T_{\textup{prac}}\parens*{\calP}$.
The qif table is a classical data structure that will be used \textit{in quantum superposition} at runtime.

\subsubsection{Nodes and Links in Qif Table}
\label{sub:nodes_and_links_in_qif_table}

\begin{definition}[Qif table]
    \label{def:qif-table}
	A qif table is composed of linked nodes.
	There are two types of nodes in the qif table.
	Each node of type $\bullet$
	represents an instantiation of $\mathbf{qif}\ldots \mathbf{fiq}$;
	i.e., an execution running through the \hlt{qif} to the corresponding \hlt{fiq} once.
	Nodes of type $\circ$ are ancilla nodes for the qif table to be \textit{reversibly used}.
	Each node $v$ of type $\bullet$ records the following information:
	\begin{enumerate}
		
		\item
			(Next link $v.\mathit{nx}$):
			If $v$ has a \textit{continuing} non-nested instantiation $v'$ of $\mathbf{qif}\ldots \mathbf{fiq}$,
			then $v.\mathit{nx}=v'$.
			Otherwise, we set $v.\mathit{nx}=v''$ for some node $v''$ of type $\circ$, enabling the qif table to be reversibly used (no matter whether $v$ has a continuing instantiation of $\mathbf{qif}\ldots \mathbf{fiq}$) in later execution.
		\item 
			(First children links $v.\mathit{fc}_i$) and (Last children links $v.\mathit{lc}_i$) for $i\in \braces*{0,1}$:
			If $v$ has \textit{enclosed} nested instantiations of $\mathbf{qif}\ldots \mathbf{fiq}$,
			then $v.\mathit{fc}_0$ and $v.\mathit{fc}_1$ links to the first two children nodes,
			representing the first two enclosed instantiations of $\mathbf{qif}\ldots \mathbf{fiq}$ 
			(corresponding to branches $\ket{0}$ and $\ket{1}$ from $v$, respectively).
			Moreover, $v.\mathit{lc}_0$ and $v.\mathit{lc}_1$ links to the last two children nodes,
			which are the two next nodes 
			(specified by the next link $\mathit{nx}$ and of type $\circ$)
			of the last two enclosed instantiations of $\mathbf{qif}\ldots\mathbf{fiq}$
			(corresponding to branches $\ket{0}$ and $\ket{1}$ from $v$, respectively).

			Otherwise, $v.\mathit{fc}_0=v.\mathit{lc}_0=v'$ and $v.\mathit{fc}_1=v.\mathit{lc}_1=v''$
			for some nodes $v',v''$ of type $\circ$.
	\end{enumerate}

	Further, each node $v$ of either type $\bullet$ or $\circ$ records the following information:
	\begin{enumerate}
         \item
			(Wait counter $v.w$):
			It stores the number of cycles to wait at node $v$.
		\item 
			($v.\mathit{pr}$):
			$\mathit{pr}$ is the inverse link of $\mathit{nx}$.
			If $v.\mathit{nx}=v'$,
			then $v'.\mathit{pr}=v$.
		\item
			($v.\mathit{cf}$):
			$\mathit{cf}$ is the inverse link of $\mathit{fc}_0$ and $\mathit{fc}_1$.
			If $v.\mathit{fc}_0=v_0$ and $v.\mathit{fc}_1=v_1$,
			then $v_0.\mathit{cf}=v_1.\mathit{cf}=v$.
		\item
			($v.\mathit{cl}$):
			$\mathit{cl}$ is the inverse link of $\mathit{lc}_0$ and $\mathit{lc}_1$.
			If $v.\mathit{lc}_0=v_0$ and $v.\mathit{lc}_1=v_1$,
			then $v_0.\mathit{cl}=v_1.\mathit{cl}=v$.
	\end{enumerate}
\end{definition}

\begin{figure}
	\centering
	\includegraphics[width=0.9\textwidth]{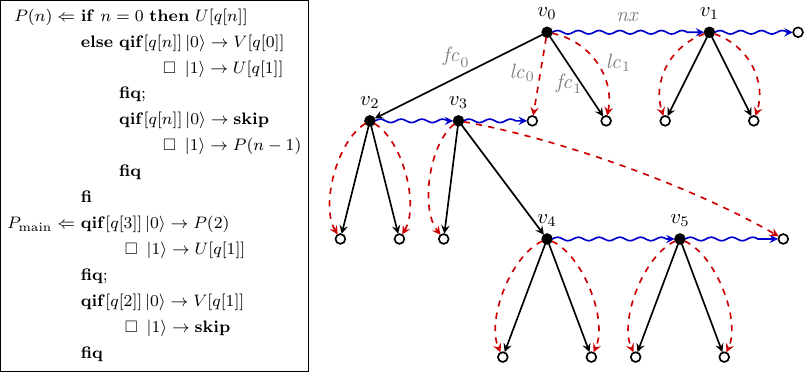}
	\caption{Example of a program in $\mathbf{RQC}^{++}$ and its corresponding qif table.
	In the qif table, 
	we only show the links $\mathit{fc}_i$ (colored in black), $\mathit{lc}_i$ (colored in red and dashed) 
	and $\mathit{nx}$ (colored in blue and squiggled),
	while $w$, $\mathit{cf}$, $\mathit{cl}$ and $\mathit{pr}$ are omitted for simplicity.
    The correspondence between the nodes on the RHS and the instantiations of $\mathbf{qif}$ statements on the LHS is as follows:
    (1) $v_0$: instantiation of the first $\mathbf{qif}\ldots \mathbf{fiq}$ in $P_{\textup{main}}$.
    (2) $v_1$: instantiation of the second $\mathbf{qif}\ldots \mathbf{fiq}$ in $P_{\textup{main}}$.
    (3) $v_2$ and $v_4$: the first and second instantiations of the first $\mathbf{qif}\ldots\mathbf{fiq}$ in $P(n)$.
    (4) $v_3$ and $v_5$: the first and second instantiations of the second $\mathbf{qif}\ldots\mathbf{fiq}$ in $P(n)$.
	}
	\label{fig:eg-qif-table}
\end{figure}

In \Cref{fig:eg-qif-table}, we give an example of a program and its corresponding qif table.
We only show the links $\mathit{fc}_i$, $\mathit{lc}_i$ and $\mathit{nx}$,
and omit $\mathit{w}$, $\mathit{cf}$, $\mathit{cl}$ and $\mathit{pr}$ for simplicity of presentation.
The partial evaluation should be done on the compiled program,
but for clarity we present the original program written in $\mathbf{RQC}^{++}$.
We also show the correspondence between nodes in the qif table and the instantiations of $\mathbf{qif}$ statements in the program.
It is easy to verify that the links are consistent with \Cref{def:qif-table}.

Additionally, we remark that to store the qif table in the QRAM,
we need to encode all links and counter recorded at a node.
The simplest way is to store them into a tuple,
where links like $v.\mathit{nx}$ records the base address of the tuple of the corresponding node.
Further discussion can be found in \Cref{sub:memory_allocation_of_qif_table,sub:qif_table_qmux}.

\subsubsection{Generation of Qif Table}
\label{sub:generation_of_qif_table}

For a compiled program $\calP$,
we fix a practical running time $T_{\textup{prac}}\parens*{\calP}$.
The partial evaluation is performed by multiple parallel processes.
We classically emulate the execution of the compiled program,
neglecting all quantum inputs and unitary gates.
Whenever a \hlt{qif} is met,
the current process forks into two sub-processes,
each continuing the evaluation of the corresponding quantum branch.
Whenever a \hlt{fiq} is met,
the current process waits for its pairing sub-process,
and collects information from both sub-processes to merge into one process.
Every process only goes into a single quantum branch and therefore contains no quantum superposition.

For each process, we maintain the following classical information.
We have $6$ system registers
$\mathit{pc}$, $\mathit{ins}$, $\mathit{br}$, $\mathit{sp}$, $v$, $t$ and a constant number of user registers.
Here, $v$ points to the current node in the qif table,
and $t$ is a counter that records the number of instructions already executed.
We also have a classical memory $M$ storing classical variables and the stack.
Let $M_i$ be the value stored at the memory location $i$.

The algorithm for partial evaluation of quantum control flow
and generation of the qif table is presented as \Cref{alg:partial-evaluation}.
The major part of the function \textsc{QEva}
is the loop between Lines~\ref{alstp:qeva-while-start}--\ref{alstp:qeva-while-end},
which consists of three stages that also appear in the execution in \Cref{sec:execution}.
The first and the last stages are similar to their classical counterparts.
The handling of instructions \hlt{qif} or \hlt{fiq} in the stage \textbf{(Decode \& Execute)} is highlighted.
\Cref{alg:partial-evaluation} returns a timeout error if $t$ exceeds the practical running time $T_{\textup{prac}}\parens*{\calP}$.
Otherwise, we obtain the actual running time $t=T_{\textup{exe}}\parens*{\calP}$ for later use in \Cref{sec:execution}.
More detailed explanation of \Cref{alg:partial-evaluation} is provided in \Cref{sub:details_qif_table_generation}.

\begin{algorithm}
	\caption{Partial evaluation of quantum control flow and generation of qif table.}
	\label{alg:partial-evaluation}
	\begin{algorithmic}[1]
        \Function{QEva}{}
			\State Initialise $\mathit{pc}\gets \text{starting address of the compiled main program}$ and $t\gets 0$
			\label{alstp:qeva-tgets0}
			\State Create an initial node $v$
			\While{$t\leq T_{\textup{prac}}\parens*{\calP}$}
            \BeginBox[fill=blue!5!white]
			\State \textbf{(Fetch)}: Let $\mathit{ins}\gets M_{\mathit{pc}}$ and $t\gets t+1$ \label{alstp:qeva-while-start}
			\label{alstp:qeva-insgets}
            \EndBox
            \BeginBox[fill=red!5!white]
            \State \textbf{(Decode \& Execute)}: 
            \If {$\mathit{ins}=$ \hlt{qif}\texttt{(q)}}
            \Comment{creation of quantum branching}
            \State Create nodes $v_0$ and $v_1$. Set
            $v.\mathit{fc}_i,v.\mathit{lc}_i\gets v_i$ and $v_i.\mathit{cf}, v_i.\mathit{cl}\gets v$
			\label{alstp:qeva-create-qif}
            \Comment{for the enclosed branches}
            \State Fork into two sub-processes \textsc{QEva}$_0$ and \textsc{QEva}$_1$.
            For \textsc{QEva}$_i$, set $v\gets v_i$ 
			\label{alstp:qeva-fork}
            \ElsIf {$\mathit{ins}=$ \hlt{fiq}\texttt{(q)}}
            \Comment{join of quantum branching}
            \State Wait for the pairing sub-process \textsc{QEva}$'$ with $v'.\mathit{cl}=v.\mathit{cl}=\hat{v}$ for some parent node $\hat{v}$
			\label{alstp:qeva-wait}
			\State $\hat{t}\gets \max\{t,t'\}$, $v.w\gets \hat{t}-t$ and $v'.w\gets \hat{t}-t'$
			\Comment{$v',t'$ are corresponding $v,t$ in \textsc{QEva}$'$.}
            \State Merge with the pairing sub-process \textsc{QEva}$'$ by letting $t\gets \hat{t}$ and $v\gets \hat{v}$
			\label{alstp:qeva-merge}
            \State Create node $u$. Set $v.\mathit{nx}\gets u$ and $u.\mathit{pr}\gets v$.
            \Comment{for the continuing branch}
            \State
			Suppose $v.\mathit{cl}=\hat{u}$ and $\hat{u}.\mathit{lc}_x=v$ for some $\hat{u}$ and $x$. 
			Let $u.\mathit{cl}\gets \hat{u}$, $\hat{u}.\mathit{lc}_x\gets u$ and $v.\mathit{cl}\gets 0$
			\State Update $v\gets u$
			\ElsIf {$\mathit{ins}=$ \hlt{finish}}
            \Comment{termination}
			\State \Return $t$
            \ElsIf {$\mathit{ins}\notin \{\textup{\hlt{uni}\texttt{(G,r)}},\textup{\hlt{unib}\texttt{(G,r$_1$,r$_2$)}}\}$}
            \Comment{neglect quantum gates}
			\State Update registers and $M$ according to \Cref{fig:table-ins}
            \EndBox
            \EndIf
            \BeginBox[fill=green!5!white]
            \State \textbf{(Branch)}:
            \If {$\mathit{br}\neq 0$}
            \State Let $\mathit{pc}\gets \mathit{pc}+\mathit{br}$
			\Else \State Let $\mathit{pc}\gets \mathit{pc}+1$ \label{alstp:qeva-while-end}
            \EndBox
            \EndIf
    		\EndWhile
			\State \Return Timeout error
        \EndFunction
	\end{algorithmic}
\end{algorithm}

\section{Execution on Quantum Register Machine}
\label{sec:execution}

Now we are ready to describe how the compiled program $\calP$ is executed,
with the aid of partial evaluation results (including symbol table and qif table),
on the quantum register machine.
Let us load all these instructions and data into the QRAM,
according to the layout described in \Cref{sub:QRAM-layout}.

\subsection{Unitary \texorpdfstring{$U_{\textup{cyc}}$}{U\_{cyc}} and Unitary \texorpdfstring{$U_{\textup{exe}}$}{U\_{exe}}}
\label{sub:cycle_unitary_u_cyc_and_execution_unitary_u_exe}

\Cref{alg:exe-qrm} presents the execution on quantum register machine,
which consists of repeated cycles,
each performing the unitary $U_{\textup{cyc}}$.
We fix the number of repetitions to be $T_{\textup{exe}}=T_{\textup{exe}}\parens*{\calP}$,
obtained from \Cref{alg:partial-evaluation}.

In $U_{\textup{cyc}}$, 
we need to decide whether to wait (i.e., skip the current cycle) or execute,
according to the wait counter information stored in the current node of the qif table.
To reversibly implement this procedure, $U_{\textup{cyc}}$ consists of three stages and exploits the registers $\mathit{qifw}$ and $\mathit{wait}$.
As a subroutine of $U_{\textup{cyc}}$,
the unitary $U_{\textup{exe}}$ consists of four stages,
inspired by the design of classical reversible processor (e.g., \cite{TAG12}).
In the \textbf{(Decode \& Execute)} stage,
the unitary $U_{\textup{dec}}$ is defined in \Cref{fig:table-type-eg}.
We provide more detailed discussion on \Cref{alg:exe-qrm} and its visualisation as quantum circuits in \Cref{sub:details_uni_cyc_and_exe}.

\algrenewcommand\algorithmicprocedure{\textbf{Unitary}}

\begin{algorithm}
	\caption{Execution on quantum register machine.}
	\label{alg:exe-qrm}
	\begin{algorithmic}[1]

	\Procedure {$U_{\textup{main}}$}{}
		\State Initialise registers according to \Cref{sub:quantum_registers}
		\For {$t=1,\ldots,T_{\textup{exe}}$}
		\State Apply the unitary \BoxedString[fill=blue!5!white]{$U_{\textup{cyc}}$} (defined below)
		\EndFor
	\EndProcedure

	\BeginBox[fill=blue!5!white]
	\Procedure {$U_{\textup{cyc}}$}{}
    \State \textbf{(Set wait flag)}:
        Conditioned on $\mathit{qifw}$, set the wait flag in $\mathit{wait}$:
    \Statex
        \qquad \qquad Perform $\sum_{w,z}\ket{w}\!\bra{w}_{\mathit{qifw}}\otimes\ket{z\oplus \bracks*{w>0}}\!\bra{z}_{\mathit{wait}}$
		\label{alstp:exe-set-wait-flag}
    \State 
        \textbf{(Execute or wait)}:
        Conditioned on $\mathit{wait}$,
		apply the unitary \BoxedString[fill=red!5!white]{$U_{\textup{exe}}$} (defined below),
        or wait and decrement the value in $\mathit{qifw}$:
    \Statex
        \qquad \qquad Perform $\ket{0}\!\bra{0}_{\mathit{wait}}\otimes U_{\textup{exe}} 
			+ \sum_{z\neq 0,w}\ket{z}\!\bra{z}_{\mathit{wait}}\otimes\ket{w-1}\!\bra{w}_{\mathit{qifw}}\otimes \Id$
		\label{alstp:exe-exe-or-wait}
    \State
        \textbf{(Clear wait flag)}:
        Conditioned $\mathit{qifw}$ and $\mathit{qifv}$,
		uncompute the wait flag in $\mathit{wait}$:
    \Statex
	\qquad \qquad Perform $\sum_{w,v,z}\ket{w}\!\bra{w}_{\mathit{qifw}}\otimes \ket{v}\!\bra{v}_{\mathit{qifv}}
            \ket{z\oplus \bracks*{w<v.w}}\!\bra{z}_{\mathit{wait}}$
		\label{alstp:exe-clear-wait-flag}
	\EndProcedure
	\EndBox

	\BeginBox[fill=red!5!white]
    \Procedure {$U_{\textup{exe}}$}{}
        \State \textbf{(Fetch)}:
        Apply the unitary $U_{\textup{fet}}\parens*{\mathit{ins},\mathit{pc},\mathit{mem}}$ 
		\label{alstp:exe-fetch}
		\Comment {$U_{\textup{fet}}$ is defined in \Cref{def:access_QRAM}.}
        \State \textbf{(Decode \& Execute)}:
        Apply the unitary $U_{\textup{dec}}$
		\label{alstp:exe-dec-exe}
		\Comment {$U_{\textup{dec}}$ is defined in \Cref{fig:table-type-eg}.}
        \State \textbf{(Unfetch)}:
		\label{alstp:exe-unfetch}
        Apply the unitary $U_{\textup{fet}}\parens*{\mathit{ins},\mathit{pc},\mathit{mem}}$ again
        \State \textbf{(Branch)}:
        Update $\mathit{pc}$, conditioned on $\mathit{br}$:
		\label{alstp:exe-branch}
        \Statex \qquad \qquad Apply $U_{+}\parens*{\mathit{pc},\mathit{br}}$
		\Comment {$U_{+}$ performs the mapping $\ket{x}\ket{y}\mapsto \ket{x+y}\ket{y}$.}
		\Statex \qquad \qquad Apply $\circ\parens*{\mathit{br}}$-$\sum_{x}\ket{x+1}\!\bra{x}_{\mathit{pc}}$
		\Comment {$\circ\parens*{\cdot}$-$U$ is defined in \Cref{sub:the_low_level_language_qins}.}
    \EndProcedure
	\EndBox
	\end{algorithmic}
\end{algorithm}

\subsection{Unitaries for Executing Qif Instructions}
\label{sub:unitaries_for_executing_qif_instructions}

It remains to define the unitaries $U_{\textup{qif}}$ and $U_{\textup{fiq}}$
that are unspecified in \Cref{fig:table-ins}.
We present their constructions in \Cref{alg:U-qif-fiq}.
Additional remarks are as follows.
\begin{itemize}
	\item 
		For $U_{\textup{qif}}$,
		note that we are promised that $\mathit{qifw}$ is initially in state $\ket{0}$,
		because $U_{\textup{qif}}$ is used as a subroutine in $U_{\textup{exe}}$,
		which will only be called by $U_{\textup{cyc}}$ when $\mathit{qifw}$ is in state $\ket{0}$.
	\item
		The information $\mathit{fc}_x$, $\mathit{cf}$, 
		$\mathit{lc}_x$, $\mathit{cl}$, $\mathit{nx}$, $\mathit{pr}$ and $w$ are stored in the qif node,
		and need to be fetched using $U_{\textup{fet}}$ into free registers before being used,
		of which details are omitted for simplicity.
\end{itemize}
Further explanation of \Cref{alg:U-qif-fiq} is provided in \Cref{sub:details_uni_qif_and_fiq}.

\begin{algorithm}
	\caption{The unitaries $U_{\mathit{qif}}$ and $U_{\mathit{fiq}}$ in \Cref{fig:table-type-eg}.}
	\label{alg:U-qif-fiq}
	\begin{algorithmic}[1]
    \Procedure {$U_{\textup{qif}}$}{$q$}
        \State 
            Conditioned on $q$,
            move $\mathit{qifv}$ to its first children node
            in the qif table via the links $\mathit{fc}_{0}$ and $\mathit{fc}_1$;
            i.e., perform the following series of unitaries:
        \Statex \qquad\qquad $V_{\mathit{fc}}=\sum_{v,x,u}\ket{v}\!\bra{v}_{\mathit{qifv}}\otimes \ket{x}\!\bra{x}_q \otimes \ket{u\oplus v.\mathit{fc}_x}\!\bra{u}_r$, where $r$ is a free register
        \Statex \qquad \qquad $V_{\mathit{cf}}=\sum_{v,u}\ket{v\oplus u.\mathit{cf}}\!\bra{v}_{\mathit{qifv}}\otimes \ket{u}\!\bra{u}_r$
        \Statex \qquad \qquad $U_{\textup{swap}}\parens*{r,\mathit{qifv}}$, which also clears register $r$
		\Comment {$U_{\textup{swap}}$ is defined in \Cref{sub:the_low_level_language_qins}.}
        \State
            Update $\mathit{qifw}$ with the wait counter information corresponding to $\mathit{qifv}$;
            i.e., perform:
        \Statex \qquad\qquad $\sum_{w,v} \ket{v}\!\bra{v}_{\mathit{qifv}}\otimes \ket{w\oplus v.w}\!\bra{w}_{\mathit{qifw}}$
    \EndProcedure

    \Procedure {$U_{\textup{fiq}}$}{$q$}
        \State
            Conditioned on $q$,
            move $\mathit{qifv}$ to its parent node in the qif table
            via the inverse link $\mathit{cl}$;
            i.e., perform the following series of unitaries:
        \Statex \qquad\qquad $V_{\mathit{cl}}=\sum_{v,u}\ket{v}\!\bra{v}_{\mathit{qifv}} \otimes \ket{u\oplus v.\mathit{cl}}\!\bra{u}_r$, where $r$ is a free register
        \Statex \qquad\qquad $V_{\mathit{lc}}=\sum_{v,x,u}\ket{v\oplus u.\mathit{lc}_q}\!\bra{v}_{\mathit{qifv}}\otimes \ket{x}\!\bra{x}_q\otimes \ket{u}\!\bra{u}_r$
        \Statex \qquad\qquad $U_{\textup{swap}}\parens*{r,\mathit{qifv}}$, which also clears register $r$
        \State Move $\mathit{qifv}$ to the next node in the qif table
        via the link $\mathit{nx}$; i.e., perform the following series of unitaries:
        \Statex \qquad\qquad $V_{\mathit{nx}}=\sum_{v,u}\ket{v}\!\bra{v}_{\mathit{qifv}}\otimes \ket{u\oplus v.\mathit{nx}}\!\bra{u}_r$, where $r$ is a free register
        \Statex \qquad\qquad $V_{\mathit{pr}}=\sum_{v,u}\ket{v\oplus u.\mathit{pr}}\!\bra{v}_{\mathit{qifv}}\otimes \ket{u}\!\bra{u}_r$
        \Statex \qquad\qquad $U_{\textup{swap}}\parens*{r,\mathit{qifv}}$, which also clears register $r$
    \EndProcedure
    \end{algorithmic}
\end{algorithm}

Now we remark on how the qif table as a classical data structure is used \textit{in quantum superposition} during the execution on quantum register machine.
Recall that at runtime, the value in register $\mathit{qifv}$ indicates the current node in the qif table. Register $\mathit{qifv}$ can be in a quantum superposition state,
in particular, entangled with the quantum coin $q$ (as well as other register and the QRAM) when instruction \hlt{qif}\texttt{(q)} is executed
(see \Cref{alg:U-qif-fiq}). 
For example, after the unitary $U_\textup{qif}(q)$ is performed, 
the state of the quantum register machine can be $\frac{1}{\sqrt{2}}\ket{0}_q\ket{v_1}_{\mathit{qifv}}\ket{\psi_0}
+ \frac{1}{\sqrt{2}}\ket{1}_q\ket{v_2}_{\mathit{qifv}}\ket{\psi_1}$,
where $\ket{\psi_0}$ and $\ket{\psi_1}$ are states of the remaining quantum registers and the QRAM. 
In this way, the information in the qif table is used in quantum superposition.

\section{Efficiency and Automatic Parallelisation}
\label{sec:computational_efficiency_and_algorithmic_speed_up}

An implementation of quantum recursive programs has been presented in the previous sections.
In this section, we analyse its efficiency,
and further show that as a bonus, such implementation also offers \textit{automatic parallelisation}.
For implementing certain algorithmic subroutine, like the quantum multiplexor introduced in \Cref{sub:motivating_example},
we can even obtain exponential parallel speed-up (over the straightforward implementation) from this automatic parallelisation.
This steps towards a \textit{top-down} design of efficient quantum algorithms:
we only need to design high-level quantum recursive programs,
and let the machine automatically realise the parallelisation 
(whose quality, of course, still depends on the program structure).
The intuition for the automatic parallelisation was already pointed out in \Cref{sub:motivating_example}:
(1) with quantum control flow, the quantum register machine can go through quantum branches in superposition; and (2) with recursive procedure calls, the program can generate exponentially many quantum branches (as each instantiation of the $\mathbf{qif}$ statement creates two quantum branches). 


In the following, we briefly describe the complexity of implementing quantum recursive programs; in particular, for partial evaluation and execution. We first describe the complexity in terms of elementary operations on registers and the QRAM, and then refine it into parallel time complexity measured by the standard (classical and quantum) circuit depth.
The full analysis can be found in \Cref{sub:complexity_partial_evaluation,sub:complexity_execution,sub:quantum_circuit_complexity_for_elementary_operations}.
Recall that, intuitively, $T_{\textup{exe}}\parens*{\calP}$ correspond to the time for executing the longest quantum branch in program $\calP$.
\begin{enumerate}
	\item 
		\Cref{alg:partial-evaluation} takes $O\parens*{T_{\textup{exe}}\parens*{\calP}}$
		classical parallel elementary operations.
		Here, ``elementary'' means the operation only involves
		a constant number of memory locations in the classical RAM 
		(as the partial evaluation is performed classically).
		``Parallel'' means multiple elementary operations
		performed simultaneously are counted as one parallel elementary operation,
		like in the standard parallel computing.
        The intuition for this complexity is that in the partial evaluation, each of the classical parallel processes only evaluate one quantum branch.

	\item
		\Cref{alg:exe-qrm} takes $O\parens*{T_{\textup{exe}}\parens*{\calP}}$ 
		quantum elementary operations, including on registers and QRAM accesses (see \Cref{def:access_QRAM}).
        The intuition was already presented in \Cref{sub:motivating_example}.
        
\end{enumerate}

The above complexities are in terms of elementary operations.
As mentioned in \Cref{sec:introduction},
the implementation will be eventually quantum circuits, so we need to translate elementary operations into quantum circuits.
The overall (classical and quantum) parallel time complexity 
of \Cref{alg:partial-evaluation,alg:exe-qrm} will be
\begin{equation*}
    O\parens*{T_{\textup{exe}}\parens*{\calP}\cdot \parens*{T_{\textup{reg}}+T_{\textup{QRAM}}}},
\end{equation*}
where $T_{\textup{reg}}$ and $T_{\textup{QRAM}}$ are parallel time complexities for elementary operations on registers 
and QRAM accesses, as aforementioned.
Here, we assume that classical elementary operations are cheaper than their quantum counterparts.

For concreteness, let us return to the example of quantum multiplexor program $\calP$ in \Cref{fig:q-multiplexor}.
Recall that in \Cref{fig:q-multiplexor-after} the programs 
after high-level transformations and high-to-mid-level translation are already presented.
Now we provide a proof sketch of \Cref{thm:parallel-qmux}, whose full proof can be found in \Cref{sub:proof-main-theorem}.

\begin{proof}[Proof Sketch of \Cref{thm:parallel-qmux}]
	Since each $C_x$ only consists of $T_x$ quantum unitary gates,
	the number of instructions in the compiled program of $P[x]\Leftarrow C_x$ will be $O\parens*{T_x}$.
	As a result, the whole compiled quantum multiplexor program (presented in \Cref{sub:further_details_mid_low_translation})
	contains $\Theta\parens*{\sum_{x\in [N]}T_x}$ instructions.
	It is easy to verify that $T_{\textup{exe}}\parens*{\calP}=O\parens*{\max_{x\in [N]} T_x+n}$.

	Let us determine $T_{\textup{reg}}$ and $T_{\textup{QRAM}}$
	by implementing the quantum register machine in the more common quantum circuit model.
	We can calculate the size $N_{\textup{QRAM}}$ of the QRAM and the word length $L_{\textup{word}}$
	for implementing $\calP$.
	In particular, taking $N_{\textup{QRAM}}=\Theta\parens*{\sum_x T_x}$ is sufficient.
	To see this, we can calculate that the sizes of the program, symbol table and variable sections
	are upper bounded by $\Theta\parens*{\sum_x T_x}$.
	The size of the qif table is $\Theta\parens*{2^n}$.
	The size of the stack is upper bounded by $\Theta\parens*{\sum_x T_x}+\Theta\parens*{n}$.
	To store an address in such QRAM, taking $L_{\textup{word}}=\Theta\parens*{\log N_{\textup{QRAM}}}$ is sufficient.

	By lifting results from classical parallel circuits for elementary arithmetic~\cite{Ofman63,Reif83,BCH86},
	we have $T_{\textup{reg}}=O\parens*{\log^2 L_{\textup{word}}}$.
	By extending existing circuit QRAM constructions (e.g., \cite{GLM08,HLGJ21}),
	we have $T_{\textup{QRAM}}=O\parens*{\log N_{\textup{QRAM}}+\log L_{\textup{word}}}$.
	The above calculations are carried out in terms of parallel time complexity, i.e., quantum circuit depth.
	Combining the above arguments leads to \Cref{thm:parallel-qmux}.
\end{proof}

\section{Related Work}
\label{sec:related_work}

\paragraph{Low-level quantum instructions} Several quantum instruction set architectures have been proposed in the literature, e.g., OpenQASM~\cite{ALJJ17}, Quil~\cite{RMW17}, eQASM~\cite{eQASM19}. Only the architecture introduced in \cite{YVC24}, called quantum control machine,  supports program counter in superposition (and hence quantum control flow), and the others do not support quantum control flow at the instruction level.
Quantum control machine also supports conditional jumps, but different from quantum register machine defined in this paper, it does not support arbitrary procedure calls.

\paragraph{Automatic parallelisation}
Numerous efforts have been devoted to parallelisation of quantum circuits of specific patterns, e.g., \cite{CW00,MN02,GHMP02,TD04,HS05,BGH07,TT13,JSTWWZ20,ZYY21,ZWY24,Rosenthal23,STYYZ23,ZLY22,YZ23}. 
Other than the quantum circuit model, the measurement-based quantum computing~\cite{RB01} is also shown to provide certain benefits for parallelisation~\cite{Josza05,DKP07,BK09,BKP11,Pius10,dSPK13}.
These techniques of parallelisation are at the low level. 
In comparison, the automatic parallelisation from our implementation is at the high level: the quantum register machine automatically exploits parallelisation opportunities in the  structures of the high-level quantum recursive programs.

\paragraph{Automatic uncomputation}
Silq~\cite{BBGV20} is the first quantum programming language that supports automatic uncomputation,
which was further investigated in \cite{PBSV21,YC22,VLRH23,YC24,PBV24,VGDV24}.
Silq's uncomputation is for quantum programs lifted from classical ones,
or in their terminology, lifted functions 
(whose semantics can be described classically and preserves the input).
Later works like~\cite{VLRH23,VGDV24} also considered uncomputation of quantum programs 
but they do not support quantum recursion.
In comparison, $\mathbf{RQC}^{++}$ supports quantum recursion,
where classical variables are used solely for specifying the control (not data).
The automatic uncomputation in our implementation is of these classical variables in quantum recursive programs.

\paragraph{Classical reversible languages}
There are extensive works in classical reversible programming languages, 
including the high-level language Janus~\cite{LD86,YG07,YAG08},
low-level instruction set architectures PISA~\cite{Vieri99,Frank99,AGY07} and BobISA~\cite{TAG12}.
Some of these reversible languages support local variables,
specified by a pair of local-delocal statements, which have explicitly reversible semantics.
In $\mathbf{RQC}^{++}$, irreversible classical computation can be done on local variables,
but their translations into low-level instructions become reversible.

\paragraph{Quantum control flow and data structures}
    Many works~\cite{AG05,YYF12,BP15,SVV18,BBGV20,VLRH23,YVC24,YC24} on quantum programming languages include quantum control flow as a feature. 
    Some~\cite{BP15,YVC24} discuss the limitation of quantum control flow.
    For example, it is shown in \cite{BP15} that the semantics of quantum recursion cannot be defined using Tarski's fixpoint theorem, when quantum measurements are involved. However, $\mathbf{RQC}^{++}$ considered in this paper only describes unitary operators and therefore circumvents this issue.
    A similar unitary restriction is used in \cite{YVC24} to support instruction-level quantum control flow.
    
    Another related topic is data structures in superposition~\cite{YC22,YC24}. The language Tower in~\cite{YC22} can describe recursive programs in superposition, which allows a \textit{single} layer of interleaving between quantum control flow and recursion. However, their syntax does not contain unitary gates and quantum if-statement, which cannot express the most general form of quantum recursion.
    In contrast, $\mathbf{RQC}^{++}$ allows \textit{arbitrary} interleaving between quantum if-statements and recursive procedure calls. Such expressive power of $\mathbf{RQC}^{++}$ also makes our implementation non-trivial.

\section{Conclusion}
\label{sec:discussion}

We propose the notion of quantum register machine,
an architecture that provides instruction-level support for quantum control flow and recursive procedure calls at the same time.
We design a comprehensive process of implementing quantum recursive programs on the quantum register machine,
including compilation, partial evaluation of quantum control flow and execution.
As a bonus, our implementation offers automatic parallelisation, from which we can even obtain exponential parallel speed-up (over the straightforward implementation)
for implementing some important quantum algorithmic subroutines like the quantum multiplexor.

To conclude this paper, let us list several topics for future research.
Firstly, an immediate next step is to develop a software 
that realises our implementation of quantum recursive programs for actual execution on future quantum hardware.
Moreover, one can consider certifying such software implementation, like in recent verified quantum compilers, e.g., \cite{ARS17,RPLZ19,HRH+21,TSY+22,LVHPWH22}.
Secondly, our implementation is designed to be simple for clarity.
It is worth extending the features of the quantum register machine
and further optimise the steps in the compilation, partial evaluation and execution.
Thirdly, it is interesting to see what other quantum algorithms (except those considered in \cite{YZ24} and this paper) can be written in quantum recursive programs
and benefit (with possible speed-up) from the efficient implementation of the quantum register machine.

\begin{acks}
    Zhicheng Zhang thanks Qisheng Wang for helpful discussions about the halting schemes of quantum Turing machine related to the synchronisation problem in \Cref{sub:the_synchronisation_problem} (see also \Cref{sub:further_synchronisation_problem}).
    We thank the anonymous reviewers for their valuable comments.
    This work was partly supported by the Australian Research Council (Grant Number: DP250102952).
    Zhicheng Zhang was supported by the Sydney Quantum Academy, NSW, Australia.
\end{acks}

\bibliographystyle{ACM-Reference-Format}
\bibliography{mybib}


\begin{thebibliography}{105}


\ifx \showCODEN    \undefined \def \showCODEN     #1{\unskip}     \fi
\ifx \showISBNx    \undefined \def \showISBNx     #1{\unskip}     \fi
\ifx \showISBNxiii \undefined \def \showISBNxiii  #1{\unskip}     \fi
\ifx \showISSN     \undefined \def \showISSN      #1{\unskip}     \fi
\ifx \showLCCN     \undefined \def \showLCCN      #1{\unskip}     \fi
\ifx \shownote     \undefined \def \shownote      #1{#1}          \fi
\ifx \showarticletitle \undefined \def \showarticletitle #1{#1}   \fi
\ifx \showURL      \undefined \def \showURL       {\relax}        \fi
\providecommand\bibfield[2]{#2}
\providecommand\bibinfo[2]{#2}
\providecommand\natexlab[1]{#1}
\providecommand\showeprint[2][]{arXiv:#2}

\bibitem[Aaronson et~al\mbox{.}(2020)]%
        {ACL+20}
\bibfield{author}{\bibinfo{person}{Scott Aaronson}, \bibinfo{person}{Nai-Hui Chia}, \bibinfo{person}{{Han-Hsuan} Lin}, \bibinfo{person}{Chunhao Wang}, {and} \bibinfo{person}{Ruizhe Zhang}.} \bibinfo{year}{2020}\natexlab{}.
\newblock \showarticletitle{On the quantum complexity of closest pair and related problems}. In \bibinfo{booktitle}{\emph{35th Computational Complexity Conference (CCC 2020)}}, Vol.~\bibinfo{volume}{169}. \bibinfo{pages}{16:1--16:43}.
\newblock
\href{https://doi.org/10.4230/LIPIcs.CCC.2020.16}{doi:\nolinkurl{10.4230/LIPIcs.CCC.2020.16}}


\bibitem[Alfred et~al\mbox{.}(2007)]%
        {AMSJ07}
\bibfield{author}{\bibinfo{person}{V.~Aho Alfred}, \bibinfo{person}{S.~Lam Monica}, \bibinfo{person}{Ravi Sethi}, {and} \bibinfo{person}{D.~Ullman Jeffrey}.} \bibinfo{year}{2007}\natexlab{}.
\newblock \bibinfo{booktitle}{\emph{Compilers: principles, techniques \& tools}}.
\newblock \bibinfo{publisher}{Pearson Education}.
\newblock


\bibitem[Altenkirch and Grattage(2005)]%
        {AG05}
\bibfield{author}{\bibinfo{person}{Thorsten Altenkirch} {and} \bibinfo{person}{Jonathan Grattage}.} \bibinfo{year}{2005}\natexlab{}.
\newblock \showarticletitle{A functional quantum programming language}. In \bibinfo{booktitle}{\emph{20th Annual IEEE Symposium on Logic in Computer Science (LICS'05)}}. \bibinfo{pages}{249--258}.
\newblock
\href{https://doi.org/10.1109/LICS.2005.1}{doi:\nolinkurl{10.1109/LICS.2005.1}}


\bibitem[Ambainis(2007)]%
        {Amb07}
\bibfield{author}{\bibinfo{person}{Andris Ambainis}.} \bibinfo{year}{2007}\natexlab{}.
\newblock \showarticletitle{Quantum walk algorithm for element distinctness}.
\newblock \bibinfo{journal}{\emph{SIAM J. Comput.}} \bibinfo{volume}{37}, \bibinfo{number}{1} (\bibinfo{year}{2007}), \bibinfo{pages}{210--239}.
\newblock
\href{https://doi.org/10.1137/S0097539705447311}{doi:\nolinkurl{10.1137/S0097539705447311}}


\bibitem[Amy et~al\mbox{.}(2017)]%
        {ARS17}
\bibfield{author}{\bibinfo{person}{Matthew Amy}, \bibinfo{person}{Martin Roetteler}, {and} \bibinfo{person}{Krysta~M. Svore}.} \bibinfo{year}{2017}\natexlab{}.
\newblock \bibinfo{booktitle}{\emph{Verified compilation of space-efficient reversible circuits}}.
\newblock \bibinfo{pages}{3–21}.
\newblock
\href{https://doi.org/10.1007/978-3-319-63390-9_1}{doi:\nolinkurl{10.1007/978-3-319-63390-9_1}}


\bibitem[Araujo et~al\mbox{.}(2021)]%
        {APPFd21}
\bibfield{author}{\bibinfo{person}{Israel~F. Araujo}, \bibinfo{person}{Daniel~K. Park}, \bibinfo{person}{Francesco Petruccione}, {and} \bibinfo{person}{Adenilton~J. {da}~Silva}.} \bibinfo{year}{2021}\natexlab{}.
\newblock \showarticletitle{A divide-and-conquer algorithm for quantum state preparation}.
\newblock \bibinfo{journal}{\emph{Scientific reports}} \bibinfo{volume}{11}, \bibinfo{number}{1} (\bibinfo{year}{2021}), \bibinfo{pages}{6329}.
\newblock
\href{https://doi.org/10.1038/s41598-021-85474-1}{doi:\nolinkurl{10.1038/s41598-021-85474-1}}


\bibitem[Axelsen(2011)]%
        {Axelsen11}
\bibfield{author}{\bibinfo{person}{Holger~Bock Axelsen}.} \bibinfo{year}{2011}\natexlab{}.
\newblock \showarticletitle{Clean translation of an imperative reversible programming language}. In \bibinfo{booktitle}{\emph{International Conference on Compiler Construction}}. \bibinfo{pages}{144--163}.
\newblock
\href{https://doi.org/10.1007/978-3-642-19861-8_9}{doi:\nolinkurl{10.1007/978-3-642-19861-8_9}}


\bibitem[Axelsen et~al\mbox{.}(2007)]%
        {AGY07}
\bibfield{author}{\bibinfo{person}{Holger~Bock Axelsen}, \bibinfo{person}{Robert Gl{\"u}ck}, {and} \bibinfo{person}{Tetsuo Yokoyama}.} \bibinfo{year}{2007}\natexlab{}.
\newblock \showarticletitle{Reversible machine code and its abstract processor architecture}. In \bibinfo{booktitle}{\emph{Computer Science--Theory and Applications: Second International Symposium on Computer Science in Russia (CSR 2007)}}. \bibinfo{pages}{56--69}.
\newblock
\href{https://doi.org/10.1007/978-3-540-74510-5_9}{doi:\nolinkurl{10.1007/978-3-540-74510-5_9}}


\bibitem[{B}abbush et~al\mbox{.}(2016)]%
        {BBKWLA16}
\bibfield{author}{\bibinfo{person}{Ryan {B}abbush}, \bibinfo{person}{Dominic~W. {B}erry}, \bibinfo{person}{Ian~D. {K}ivlichan}, \bibinfo{person}{Annie~Y. {Wei}}, \bibinfo{person}{Peter~J. {L}ove}, {and} \bibinfo{person}{Al{\'{a}}n {A}spuru{-G}uzik}.} \bibinfo{year}{2016}\natexlab{}.
\newblock \showarticletitle{Exponentially more precise quantum simulation of fermions in second quantization}.
\newblock \bibinfo{journal}{\emph{New Journal of Physics}}  \bibinfo{volume}{18} (\bibinfo{year}{2016}), \bibinfo{pages}{033032}.
\newblock
\href{https://doi.org/10.1088/1367-2630/18/3/033032}{doi:\nolinkurl{10.1088/1367-2630/18/3/033032}}


\bibitem[Babbush et~al\mbox{.}(2018)]%
        {BGBWMPFN18}
\bibfield{author}{\bibinfo{person}{Ryan Babbush}, \bibinfo{person}{Craig Gidney}, \bibinfo{person}{Dominic~W. Berry}, \bibinfo{person}{Nathan Wiebe}, \bibinfo{person}{Jarrod {McClean}}, \bibinfo{person}{Alexandru Paler}, \bibinfo{person}{Austin Fowler}, {and} \bibinfo{person}{Hartmut Neven}.} \bibinfo{year}{2018}\natexlab{}.
\newblock \showarticletitle{Encoding electronic spectra in quantum circuits with linear {T} complexity}.
\newblock \bibinfo{journal}{\emph{Physical Review X}} \bibinfo{volume}{8}, \bibinfo{number}{4} (\bibinfo{year}{2018}), \bibinfo{pages}{041015}.
\newblock
\href{https://doi.org/10.1103/PhysRevX.8.041015}{doi:\nolinkurl{10.1103/PhysRevX.8.041015}}


\bibitem[{B}abbush et~al\mbox{.}(2018)]%
        {BWMMNC18}
\bibfield{author}{\bibinfo{person}{Ryan {B}abbush}, \bibinfo{person}{Nathan {W}iebe}, \bibinfo{person}{Jarrod {M}c{C}lean}, \bibinfo{person}{James {M}c{C}lain}, \bibinfo{person}{Hartmut {N}even}, {and} \bibinfo{person}{Garnet~Kin{-L}ic {C}han}.} \bibinfo{year}{2018}\natexlab{}.
\newblock \showarticletitle{Low{-}depth quantum simulation of materials}.
\newblock \bibinfo{journal}{\emph{Physical Review X}}  \bibinfo{volume}{8} (\bibinfo{year}{2018}), \bibinfo{pages}{011044}.
\newblock
Issue 1.
\href{https://doi.org/10.1103/PhysRevX.8.011044}{doi:\nolinkurl{10.1103/PhysRevX.8.011044}}


\bibitem[Backus et~al\mbox{.}(1960)]%
        {ALGOL60}
\bibfield{author}{\bibinfo{person}{John~W. Backus}, \bibinfo{person}{Friedrich~L. Bauer}, \bibinfo{person}{Julien Green}, \bibinfo{person}{Charles Katz}, \bibinfo{person}{John McCarthy}, \bibinfo{person}{Alan~J. Perlis}, \bibinfo{person}{Heinz Rutishauser}, \bibinfo{person}{Klaus Samelson}, \bibinfo{person}{Bernard Vauquois}, \bibinfo{person}{Joseph~Henry Wegstein}, \bibinfo{person}{Adriaan {van Wijngaarden}}, {and} \bibinfo{person}{Michael Woodger}.} \bibinfo{year}{1960}\natexlab{}.
\newblock \showarticletitle{Report on the algorithmic language {ALGOL 60}}.
\newblock \bibinfo{journal}{\emph{Commun. ACM}} \bibinfo{volume}{3}, \bibinfo{number}{5} (\bibinfo{year}{1960}), \bibinfo{pages}{299--311}.
\newblock
\href{https://doi.org/10.1145/367236.367262}{doi:\nolinkurl{10.1145/367236.367262}}


\bibitem[Backus et~al\mbox{.}(1963)]%
        {ALGOL60b}
\bibfield{author}{\bibinfo{person}{John~W. Backus}, \bibinfo{person}{Friedrich~L. Bauer}, \bibinfo{person}{Julien Green}, \bibinfo{person}{Charles Katz}, \bibinfo{person}{John McCarthy}, \bibinfo{person}{Alan~J. Perlis}, \bibinfo{person}{Heinz Rutishauser}, \bibinfo{person}{Klaus Samelson}, \bibinfo{person}{Bernard Vauquois}, \bibinfo{person}{Joseph~Henry Wegstein}, \bibinfo{person}{Adriaan {van Wijngaarden}}, {and} \bibinfo{person}{Michael Woodger}.} \bibinfo{year}{1963}\natexlab{}.
\newblock \showarticletitle{Revised report on the algorithmic language {ALGOL 60}}.
\newblock \bibinfo{journal}{\emph{Commun. ACM}} \bibinfo{volume}{6}, \bibinfo{number}{1} (\bibinfo{year}{1963}), \bibinfo{pages}{1--17}.
\newblock
\href{https://doi.org/10.1145/366193.366201}{doi:\nolinkurl{10.1145/366193.366201}}


\bibitem[Beame et~al\mbox{.}(1986)]%
        {BCH86}
\bibfield{author}{\bibinfo{person}{Paul~W. Beame}, \bibinfo{person}{Stephen~A. Cook}, {and} \bibinfo{person}{H.~James Hoover}.} \bibinfo{year}{1986}\natexlab{}.
\newblock \showarticletitle{Log depth circuits for division and related problems}.
\newblock \bibinfo{journal}{\emph{SIAM J. Comput.}} \bibinfo{volume}{15}, \bibinfo{number}{4} (\bibinfo{year}{1986}), \bibinfo{pages}{994--1003}.
\newblock
\href{https://doi.org/10.1109/SFCS.1984.715894}{doi:\nolinkurl{10.1109/SFCS.1984.715894}}


\bibitem[{B}ennett(1973)]%
        {Bennett73}
\bibfield{author}{\bibinfo{person}{Charles~H. {B}ennett}.} \bibinfo{year}{1973}\natexlab{}.
\newblock \showarticletitle{Logical reversibility of computation}.
\newblock \bibinfo{journal}{\emph{IBM Journal of Research and Development}} \bibinfo{volume}{17}, \bibinfo{number}{6} (\bibinfo{year}{1973}), \bibinfo{pages}{525--532}.
\newblock
\href{https://doi.org/10.1147/rd.176.0525}{doi:\nolinkurl{10.1147/rd.176.0525}}


\bibitem[Bera et~al\mbox{.}(2007)]%
        {BGH07}
\bibfield{author}{\bibinfo{person}{Debajyoti Bera}, \bibinfo{person}{Frederic Green}, {and} \bibinfo{person}{Steven Homer}.} \bibinfo{year}{2007}\natexlab{}.
\newblock \showarticletitle{Small depth quantum circuits}.
\newblock \bibinfo{journal}{\emph{ACM SIGACT News}} \bibinfo{volume}{38}, \bibinfo{number}{2} (\bibinfo{year}{2007}), \bibinfo{pages}{35--50}.
\newblock
\href{https://doi.org/10.1145/1272729.1272739}{doi:\nolinkurl{10.1145/1272729.1272739}}


\bibitem[Bernstein et~al\mbox{.}(2013)]%
        {BJLM13}
\bibfield{author}{\bibinfo{person}{Daniel~J. Bernstein}, \bibinfo{person}{Stacey Jeffery}, \bibinfo{person}{Tanja Lange}, {and} \bibinfo{person}{Alexander Meurer}.} \bibinfo{year}{2013}\natexlab{}.
\newblock \showarticletitle{Quantum algorithms for the subset-sum problem}. In \bibinfo{booktitle}{\emph{Post-Quantum Cryptography: 5th International Workshop, PQCrypto 2013}}. \bibinfo{pages}{16--33}.
\newblock
\href{https://doi.org/10.1007/978-3-642-38616-9_2}{doi:\nolinkurl{10.1007/978-3-642-38616-9_2}}


\bibitem[Bernstein and Vazirani(1993)]%
        {BV93}
\bibfield{author}{\bibinfo{person}{Ethan Bernstein} {and} \bibinfo{person}{Umesh Vazirani}.} \bibinfo{year}{1993}\natexlab{}.
\newblock \showarticletitle{Quantum complexity theory}. In \bibinfo{booktitle}{\emph{Proceedings of the twenty-fifth annual ACM symposium on Theory of computing}}. \bibinfo{pages}{11--20}.
\newblock
\href{https://doi.org/10.1145/167088.167097}{doi:\nolinkurl{10.1145/167088.167097}}


\bibitem[{B}erry et~al\mbox{.}(2015b)]%
        {BCCKS15}
\bibfield{author}{\bibinfo{person}{Dominic~W. {B}erry}, \bibinfo{person}{Andrew~M. {C}hilds}, \bibinfo{person}{Richard {C}leve}, \bibinfo{person}{Robin {K}othari}, {and} \bibinfo{person}{Rolando~D. {S}omma}.} \bibinfo{year}{2015}\natexlab{b}.
\newblock \showarticletitle{Simulating {H}amiltonian dynamics with a truncated {T}aylor series}.
\newblock \bibinfo{journal}{\emph{Physical Review Letters}}  \bibinfo{volume}{114} (\bibinfo{year}{2015}), \bibinfo{pages}{090502}.
\newblock
Issue 9.
\href{https://doi.org/10.1103/PhysRevLett.114.090502}{doi:\nolinkurl{10.1103/PhysRevLett.114.090502}}


\bibitem[{B}erry et~al\mbox{.}(2015a)]%
        {BCK15}
\bibfield{author}{\bibinfo{person}{Dominic~W. {B}erry}, \bibinfo{person}{Andrew~M. {C}hilds}, {and} \bibinfo{person}{Robin {K}othari}.} \bibinfo{year}{2015}\natexlab{a}.
\newblock \showarticletitle{Hamiltonian simulation with nearly optimal dependence on all parameters}. In \bibinfo{booktitle}{\emph{Proceedings of the 56th Annual IEEE Symposium on Foundations of Computer Science}} \emph{(\bibinfo{series}{FOCS '15})}. \bibinfo{pages}{792--809}.
\newblock
\href{https://doi.org/10.1109/FOCS.2015.54}{doi:\nolinkurl{10.1109/FOCS.2015.54}}


\bibitem[Bichsel et~al\mbox{.}(2020)]%
        {BBGV20}
\bibfield{author}{\bibinfo{person}{Benjamin Bichsel}, \bibinfo{person}{Maximilian Baader}, \bibinfo{person}{Timon Gehr}, {and} \bibinfo{person}{Martin Vechev}.} \bibinfo{year}{2020}\natexlab{}.
\newblock \showarticletitle{Silq: A high-level quantum language with safe uncomputation and intuitive semantics}. In \bibinfo{booktitle}{\emph{Proceedings of the 41st ACM SIGPLAN Conference on Programming Language Design and Implementation}}. \bibinfo{pages}{286--300}.
\newblock
\href{https://doi.org/10.1145/3385412.3386007}{doi:\nolinkurl{10.1145/3385412.3386007}}


\bibitem[Broadbent and Kashefi(2009)]%
        {BK09}
\bibfield{author}{\bibinfo{person}{Anne Broadbent} {and} \bibinfo{person}{Elham Kashefi}.} \bibinfo{year}{2009}\natexlab{}.
\newblock \showarticletitle{Parallelizing quantum circuits}.
\newblock \bibinfo{journal}{\emph{Theoretical computer science}} \bibinfo{volume}{410}, \bibinfo{number}{26} (\bibinfo{year}{2009}), \bibinfo{pages}{2489--2510}.
\newblock
\href{https://doi.org/10.1016/j.tcs.2008.12.046}{doi:\nolinkurl{10.1016/j.tcs.2008.12.046}}


\bibitem[Browne et~al\mbox{.}(2011)]%
        {BKP11}
\bibfield{author}{\bibinfo{person}{Dan Browne}, \bibinfo{person}{Elham Kashefi}, {and} \bibinfo{person}{Simon Perdrix}.} \bibinfo{year}{2011}\natexlab{}.
\newblock \showarticletitle{Computational depth complexity of measurement-based quantum computation}. In \bibinfo{booktitle}{\emph{Theory of Quantum Computation, Communication, and Cryptography: 5th Conference, TQC 2010}}. \bibinfo{pages}{35--46}.
\newblock
\href{https://doi.org/10.1007/978-3-642-18073-6_4}{doi:\nolinkurl{10.1007/978-3-642-18073-6_4}}


\bibitem[B\u{a}descu and Panangaden(2015)]%
        {BP15}
\bibfield{author}{\bibinfo{person}{Costin B\u{a}descu} {and} \bibinfo{person}{Prakash Panangaden}.} \bibinfo{year}{2015}\natexlab{}.
\newblock \showarticletitle{Quantum alternation: prospects and problems}. In \bibinfo{booktitle}{\emph{Proceedings of the 12th International Workshop on Quantum Physics and Logic}} \emph{(\bibinfo{series}{EPTCS}, Vol.~\bibinfo{volume}{195})}. \bibinfo{pages}{33--42}.
\newblock
\href{https://doi.org/10.4204/EPTCS.195.3}{doi:\nolinkurl{10.4204/EPTCS.195.3}}


\bibitem[{C}hilds et~al\mbox{.}(2017)]%
        {CKS17}
\bibfield{author}{\bibinfo{person}{Andrew~M. {C}hilds}, \bibinfo{person}{Robin {K}othari}, {and} \bibinfo{person}{Rolando~D. {S}omma}.} \bibinfo{year}{2017}\natexlab{}.
\newblock \showarticletitle{Quantum algorithm for systems of linear equations with exponentially improved dependence on precision}.
\newblock \bibinfo{journal}{\emph{SIAM J. Comput.}} \bibinfo{volume}{46}, \bibinfo{number}{6} (\bibinfo{year}{2017}), \bibinfo{pages}{1920--1950}.
\newblock
\href{https://doi.org/10.1137/16M1087072}{doi:\nolinkurl{10.1137/16M1087072}}


\bibitem[{C}hilds and {W}iebe(2012)]%
        {CW12}
\bibfield{author}{\bibinfo{person}{Andrew~M. {C}hilds} {and} \bibinfo{person}{Nathan {W}iebe}.} \bibinfo{year}{2012}\natexlab{}.
\newblock \showarticletitle{Hamiltonian simulation using linear combinations of unitary operations}.
\newblock \bibinfo{journal}{\emph{Quantum Information \& Computation}} \bibinfo{volume}{12}, \bibinfo{number}{11--12} (\bibinfo{year}{2012}), \bibinfo{pages}{901--924}.
\newblock
\href{https://doi.org/10.26421/QIC12.11-12-1}{doi:\nolinkurl{10.26421/QIC12.11-12-1}}


\bibitem[{C}leve and {Watrous}(2000)]%
        {CW00}
\bibfield{author}{\bibinfo{person}{Richard {C}leve} {and} \bibinfo{person}{John {Watrous}}.} \bibinfo{year}{2000}\natexlab{}.
\newblock \showarticletitle{Fast parallel circuits for the quantum {F}ourier transform}. In \bibinfo{booktitle}{\emph{Proceedings of the 41st Annual IEEE Symposium on Foundations of Computer Science}} \emph{(\bibinfo{series}{FOCS '00})}. \bibinfo{pages}{526--536}.
\newblock
\href{https://doi.org/10.1109/SFCS.2000.892140}{doi:\nolinkurl{10.1109/SFCS.2000.892140}}


\bibitem[Cross et~al\mbox{.}(2017)]%
        {ALJJ17}
\bibfield{author}{\bibinfo{person}{Andrew~W. Cross}, \bibinfo{person}{Lev~S. Bishop}, \bibinfo{person}{John~A. Smolin}, {and} \bibinfo{person}{Jay~M. Gambetta}.} \bibinfo{year}{2017}\natexlab{}.
\newblock \bibinfo{title}{Open quantum assembly language}.
\newblock
\showeprint[arxiv]{1707.03429}~[quant-ph]


\bibitem[{da Silva} et~al\mbox{.}(2013)]%
        {dSPK13}
\bibfield{author}{\bibinfo{person}{Raphael~Dias {da Silva}}, \bibinfo{person}{Einar Pius}, {and} \bibinfo{person}{Elham Kashefi}.} \bibinfo{year}{2013}\natexlab{}.
\newblock \bibinfo{title}{Global quantum circuit optimization}.
\newblock
\showeprint[arxiv]{1301.0351}~[quant-ph]


\bibitem[Danos et~al\mbox{.}(2007)]%
        {DKP07}
\bibfield{author}{\bibinfo{person}{Vincent Danos}, \bibinfo{person}{Elham Kashefi}, {and} \bibinfo{person}{Prakash Panangaden}.} \bibinfo{year}{2007}\natexlab{}.
\newblock \showarticletitle{The measurement calculus}.
\newblock \bibinfo{journal}{\emph{Journal of the ACM (JACM)}} \bibinfo{volume}{54}, \bibinfo{number}{2} (\bibinfo{year}{2007}), \bibinfo{pages}{8--es}.
\newblock
\href{https://doi.org/10.1145/1219092.1219096}{doi:\nolinkurl{10.1145/1219092.1219096}}


\bibitem[Deng et~al\mbox{.}(2024)]%
        {DTPW24}
\bibfield{author}{\bibinfo{person}{Haowei Deng}, \bibinfo{person}{Runzhou Tao}, \bibinfo{person}{Yuxiang Peng}, {and} \bibinfo{person}{Xiaodi Wu}.} \bibinfo{year}{2024}\natexlab{}.
\newblock \showarticletitle{A case for synthesis of recursive quantum unitary programs}.
\newblock \bibinfo{journal}{\emph{Proceedings of the ACM on Programming Languages}}  \bibinfo{volume}{8} (\bibinfo{year}{2024}), \bibinfo{pages}{1759--1788}.
\newblock
\href{https://doi.org/10.1145/3632901}{doi:\nolinkurl{10.1145/3632901}}


\bibitem[Deutsch(1985)]%
        {Deutsch85}
\bibfield{author}{\bibinfo{person}{David Deutsch}.} \bibinfo{year}{1985}\natexlab{}.
\newblock \showarticletitle{Quantum theory, the {Church}--{Turing} principle and the universal quantum computer}.
\newblock \bibinfo{journal}{\emph{Proceedings of the Royal Society of London. A. Mathematical and Physical Sciences}} \bibinfo{volume}{400}, \bibinfo{number}{1818} (\bibinfo{year}{1985}), \bibinfo{pages}{97--117}.
\newblock
\href{https://doi.org/10.1098/rspa.1985.0070}{doi:\nolinkurl{10.1098/rspa.1985.0070}}


\bibitem[Dijkstra(1960)]%
        {Dij60}
\bibfield{author}{\bibinfo{person}{Edsger~W. Dijkstra}.} \bibinfo{year}{1960}\natexlab{}.
\newblock \showarticletitle{Recursive programming}.
\newblock \bibinfo{journal}{\emph{Numer. Math.}} \bibinfo{volume}{2}, \bibinfo{number}{1} (\bibinfo{year}{1960}), \bibinfo{pages}{312--318}.
\newblock
\href{https://doi.org/10.1007/bf01386232}{doi:\nolinkurl{10.1007/bf01386232}}


\bibitem[Dijkstra(1970)]%
        {EWD249}
\bibfield{author}{\bibinfo{person}{Edsger~W. Dijkstra}.} \bibinfo{year}{1970}\natexlab{}.
\newblock \bibinfo{title}{Notes on structured programming}.  (\bibinfo{year}{1970}).
\newblock
\urldef\tempurl%
\url{http://www.cs.utexas.edu/users/EWD/ewd02xx/EWD249.PDF}
\showURL{%
\tempurl}
\newblock
\shownote{circulated privately}.


\bibitem[Dijkstra(1979)]%
        {EWD117}
\bibfield{author}{\bibinfo{person}{Edsger~W. Dijkstra}.} \bibinfo{year}{1979}\natexlab{}.
\newblock \showarticletitle{Programming considered as a human activity}.
\newblock In \bibinfo{booktitle}{\emph{Classics in software engineering}}. \bibinfo{pages}{1--9}.
\newblock
\urldef\tempurl%
\url{https://www.cs.utexas.edu/~EWD/ewd01xx/EWD117.PDF}
\showURL{%
\tempurl}


\bibitem[Frank(1999)]%
        {Frank99}
\bibfield{author}{\bibinfo{person}{Michael~Patrick Frank}.} \bibinfo{year}{1999}\natexlab{}.
\newblock \emph{\bibinfo{title}{Reversibility for efficient computing}}.
\newblock \bibinfo{thesistype}{Ph.\,D. Dissertation}. \bibinfo{school}{Massachusetts Institute of Technology}.
\newblock


\bibitem[Fu et~al\mbox{.}(2019)]%
        {eQASM19}
\bibfield{author}{\bibinfo{person}{X. Fu}, \bibinfo{person}{L. Riesebos}, \bibinfo{person}{M.~A. Rol}, \bibinfo{person}{Jeroen van Straten}, \bibinfo{person}{J. van Someren}, \bibinfo{person}{N. Khammassi}, \bibinfo{person}{I. Ashraf}, \bibinfo{person}{R.~F.~L. Vermeulen}, \bibinfo{person}{V. Newsum}, \bibinfo{person}{K.~K.~L. Loh}, \bibinfo{person}{J.~C. de Sterke}, \bibinfo{person}{W.~J. Vlothuizen}, \bibinfo{person}{R.~N. Schouten}, \bibinfo{person}{C.~G. Almudever}, \bibinfo{person}{L. DiCarlo}, {and} \bibinfo{person}{K. Bertels}.} \bibinfo{year}{2019}\natexlab{}.
\newblock \showarticletitle{eQASM: An executable quantum instruction set architecture}. In \bibinfo{booktitle}{\emph{2019 IEEE International Symposium on High Performance Computer Architecture (HPCA)}}. \bibinfo{pages}{224--237}.
\newblock
\href{https://doi.org/10.1109/HPCA.2019.00040}{doi:\nolinkurl{10.1109/HPCA.2019.00040}}


\bibitem[Giovannetti et~al\mbox{.}(2008a)]%
        {GLM08b}
\bibfield{author}{\bibinfo{person}{Vittorio Giovannetti}, \bibinfo{person}{Seth Lloyd}, {and} \bibinfo{person}{Lorenzo Maccone}.} \bibinfo{year}{2008}\natexlab{a}.
\newblock \showarticletitle{Architectures for a quantum random access memory}.
\newblock \bibinfo{journal}{\emph{Physical Review A}} \bibinfo{volume}{78}, \bibinfo{number}{5} (\bibinfo{year}{2008}), \bibinfo{pages}{052310}.
\newblock
\href{https://doi.org/10.1103/PhysRevA.78.052310}{doi:\nolinkurl{10.1103/PhysRevA.78.052310}}


\bibitem[Giovannetti et~al\mbox{.}(2008b)]%
        {GLM08}
\bibfield{author}{\bibinfo{person}{Vittorio Giovannetti}, \bibinfo{person}{Seth Lloyd}, {and} \bibinfo{person}{Lorenzo Maccone}.} \bibinfo{year}{2008}\natexlab{b}.
\newblock \showarticletitle{Quantum random access memory}.
\newblock \bibinfo{journal}{\emph{Physical Review Letters}} \bibinfo{volume}{100}, \bibinfo{number}{16} (\bibinfo{year}{2008}), \bibinfo{pages}{160501}.
\newblock
\href{https://doi.org/10.1103/PhysRevLett.100.160501}{doi:\nolinkurl{10.1103/PhysRevLett.100.160501}}


\bibitem[{G}reen et~al\mbox{.}(2002)]%
        {GHMP02}
\bibfield{author}{\bibinfo{person}{Frederic {G}reen}, \bibinfo{person}{Steven {H}omer}, \bibinfo{person}{Cristopher {M}oore}, {and} \bibinfo{person}{Christopher {P}ollett}.} \bibinfo{year}{2002}\natexlab{}.
\newblock \showarticletitle{Counting, fanout and the complexity of quantum {ACC}}.
\newblock \bibinfo{journal}{\emph{Quantum Information \& Computation}} \bibinfo{volume}{2}, \bibinfo{number}{1} (\bibinfo{year}{2002}), \bibinfo{pages}{35--65}.
\newblock
\href{https://doi.org/10.26421/QIC2.1-3}{doi:\nolinkurl{10.26421/QIC2.1-3}}


\bibitem[Greenberger et~al\mbox{.}(1989)]%
        {GHZ89}
\bibfield{author}{\bibinfo{person}{Daniel~M. Greenberger}, \bibinfo{person}{Michael~A. Horne}, {and} \bibinfo{person}{Anton Zeilinger}.} \bibinfo{year}{1989}\natexlab{}.
\newblock \showarticletitle{Going beyond {Bell}’s theorem}.
\newblock In \bibinfo{booktitle}{\emph{Bell’s theorem, quantum theory and conceptions of the universe}}. \bibinfo{publisher}{Springer}, \bibinfo{pages}{69--72}.
\newblock


\bibitem[Hann et~al\mbox{.}(2021)]%
        {HLGJ21}
\bibfield{author}{\bibinfo{person}{Connor~T. Hann}, \bibinfo{person}{Gideon Lee}, \bibinfo{person}{S.~M. Girvin}, {and} \bibinfo{person}{Liang Jiang}.} \bibinfo{year}{2021}\natexlab{}.
\newblock \showarticletitle{Resilience of quantum random access memory to generic noise}.
\newblock \bibinfo{journal}{\emph{PRX Quantum}} \bibinfo{volume}{2}, \bibinfo{number}{2} (\bibinfo{year}{2021}), \bibinfo{pages}{020311}.
\newblock
\href{https://doi.org/10.1103/PRXQuantum.2.020311}{doi:\nolinkurl{10.1103/PRXQuantum.2.020311}}


\bibitem[Hann et~al\mbox{.}(2019)]%
        {HZZCSGJ19}
\bibfield{author}{\bibinfo{person}{Connor~T. Hann}, \bibinfo{person}{Chang-Ling Zou}, \bibinfo{person}{Yaxing Zhang}, \bibinfo{person}{Yiwen Chu}, \bibinfo{person}{Robert~J. Schoelkopf}, \bibinfo{person}{Steven~M. Girvin}, {and} \bibinfo{person}{Liang Jiang}.} \bibinfo{year}{2019}\natexlab{}.
\newblock \showarticletitle{Hardware-efficient quantum random access memory with hybrid quantum acoustic systems}.
\newblock \bibinfo{journal}{\emph{Physical Review Letters}} \bibinfo{volume}{123}, \bibinfo{number}{25} (\bibinfo{year}{2019}), \bibinfo{pages}{250501}.
\newblock
\href{https://doi.org/10.1103/PhysRevLett.123.250501}{doi:\nolinkurl{10.1103/PhysRevLett.123.250501}}


\bibitem[{H}arrow et~al\mbox{.}(2009)]%
        {HHL09}
\bibfield{author}{\bibinfo{person}{Aram~W. {H}arrow}, \bibinfo{person}{Avinatan {H}assidim}, {and} \bibinfo{person}{Seth {L}loyd}.} \bibinfo{year}{2009}\natexlab{}.
\newblock \showarticletitle{Quantum algorithm for linear systems of equations}.
\newblock \bibinfo{journal}{\emph{Physical Review Letters}} \bibinfo{volume}{103}, \bibinfo{number}{15} (\bibinfo{year}{2009}), \bibinfo{pages}{150502}.
\newblock


\bibitem[Hietala et~al\mbox{.}(2021)]%
        {HRH+21}
\bibfield{author}{\bibinfo{person}{Kesha Hietala}, \bibinfo{person}{Robert Rand}, \bibinfo{person}{Shih-Han Hung}, \bibinfo{person}{Xiaodi Wu}, {and} \bibinfo{person}{Michael Hicks}.} \bibinfo{year}{2021}\natexlab{}.
\newblock \showarticletitle{A verified optimizer for quantum circuits}.
\newblock \bibinfo{journal}{\emph{Proceedings of the ACM on Programming Languages}} \bibinfo{volume}{5}, \bibinfo{number}{POPL} (\bibinfo{year}{2021}), \bibinfo{pages}{1--29}.
\newblock
\href{https://doi.org/10.1145/3434318}{doi:\nolinkurl{10.1145/3434318}}


\bibitem[Hoare(1961)]%
        {Hoare61}
\bibfield{author}{\bibinfo{person}{Charles Antony~Richard Hoare}.} \bibinfo{year}{1961}\natexlab{}.
\newblock \showarticletitle{Algorithm 64: quicksort}.
\newblock \bibinfo{journal}{\emph{Commun. ACM}} \bibinfo{volume}{4}, \bibinfo{number}{7} (\bibinfo{year}{1961}), \bibinfo{pages}{321}.
\newblock
\href{https://doi.org/10.1145/366622.366644}{doi:\nolinkurl{10.1145/366622.366644}}


\bibitem[Hoare(1975)]%
        {Hoare75}
\bibfield{author}{\bibinfo{person}{Charles Antony~Richard Hoare}.} \bibinfo{year}{1975}\natexlab{}.
\newblock \showarticletitle{Recursive data structures}.
\newblock \bibinfo{journal}{\emph{International Journal of Computer \& Information Sciences}} \bibinfo{volume}{4}, \bibinfo{number}{2} (\bibinfo{year}{1975}), \bibinfo{pages}{105--132}.
\newblock
\href{https://doi.org/10.1007/BF00976239}{doi:\nolinkurl{10.1007/BF00976239}}


\bibitem[{H\o}yer and {\v S}palek(2005)]%
        {HS05}
\bibfield{author}{\bibinfo{person}{Peter {H\o}yer} {and} \bibinfo{person}{Robert {\v S}palek}.} \bibinfo{year}{2005}\natexlab{}.
\newblock \showarticletitle{Quantum fan-out is powerful}.
\newblock \bibinfo{journal}{\emph{Theory of Computing}} \bibinfo{volume}{1}, \bibinfo{number}{5} (\bibinfo{year}{2005}), \bibinfo{pages}{81--103}.
\newblock
\href{https://doi.org/10.4086/toc.2005.v001a005}{doi:\nolinkurl{10.4086/toc.2005.v001a005}}


\bibitem[Jaques and Rattew(2023)]%
        {JR23}
\bibfield{author}{\bibinfo{person}{Samuel Jaques} {and} \bibinfo{person}{Arthur~G. Rattew}.} \bibinfo{year}{2023}\natexlab{}.
\newblock \bibinfo{title}{QRAM: A survey and critique}.
\newblock
\showeprint[arxiv]{2305.10310}~[quant-ph]


\bibitem[{J}iang et~al\mbox{.}(2020)]%
        {JSTWWZ20}
\bibfield{author}{\bibinfo{person}{Jiaqing {J}iang}, \bibinfo{person}{Xiaoming {S}un}, \bibinfo{person}{Shang-Hua {T}eng}, \bibinfo{person}{Bujiao {W}u}, \bibinfo{person}{Kewen {W}u}, {and} \bibinfo{person}{Jialin {Z}hang}.} \bibinfo{year}{2020}\natexlab{}.
\newblock \showarticletitle{Optimal space{-}depth trade{-}off of {CNOT} circuits in quantum logic synthesis}. In \bibinfo{booktitle}{\emph{Proceedings of the 31st Annual ACM-SIAM Symposium on Discrete Algorithms}} \emph{(\bibinfo{series}{SODA '20})}. \bibinfo{pages}{213--229}.
\newblock
\href{https://doi.org/10.1137/1.9781611975994.13}{doi:\nolinkurl{10.1137/1.9781611975994.13}}


\bibitem[Jones et~al\mbox{.}(1993)]%
        {JGS93}
\bibfield{author}{\bibinfo{person}{Neil~D. Jones}, \bibinfo{person}{Carsten~K. Gomard}, {and} \bibinfo{person}{Peter Sestoft}.} \bibinfo{year}{1993}\natexlab{}.
\newblock \bibinfo{booktitle}{\emph{Partial evaluation and automatic program generation}}.
\newblock \bibinfo{publisher}{Prentice Hall}.
\newblock


\bibitem[Jozsa(2005)]%
        {Josza05}
\bibfield{author}{\bibinfo{person}{Richard Jozsa}.} \bibinfo{year}{2005}\natexlab{}.
\newblock \bibinfo{title}{An introduction to measurement based quantum computation}.
\newblock
\showeprint[arxiv]{quant-ph/0508124}~[quant-ph]


\bibitem[{K}erenidis and {P}rakash(2017)]%
        {KP17}
\bibfield{author}{\bibinfo{person}{Iordanis {K}erenidis} {and} \bibinfo{person}{Anupam {P}rakash}.} \bibinfo{year}{2017}\natexlab{}.
\newblock \showarticletitle{Quantum recommendation systems}. In \bibinfo{booktitle}{\emph{8th Innovations in Theoretical Computer Science Conference (ITCS 2017)}}, Vol.~\bibinfo{volume}{67}. \bibinfo{pages}{49:1--49:21}.
\newblock


\bibitem[{K}othari(2014)]%
        {Kothari14}
\bibfield{author}{\bibinfo{person}{Robin {K}othari}.} \bibinfo{year}{2014}\natexlab{}.
\newblock \emph{\bibinfo{title}{Efficient algorithms in quantum query complexity}}.
\newblock \bibinfo{thesistype}{Ph.\,D. Dissertation}. \bibinfo{school}{University of Waterloo}.
\newblock


\bibitem[Landauer(1961)]%
        {Landauer61}
\bibfield{author}{\bibinfo{person}{Rolf Landauer}.} \bibinfo{year}{1961}\natexlab{}.
\newblock \showarticletitle{Irreversibility and heat generation in the computing process}.
\newblock \bibinfo{journal}{\emph{IBM journal of research and development}} \bibinfo{volume}{5}, \bibinfo{number}{3} (\bibinfo{year}{1961}), \bibinfo{pages}{183--191}.
\newblock
\href{https://doi.org/10.1147/rd.53.0183}{doi:\nolinkurl{10.1147/rd.53.0183}}


\bibitem[Li et~al\mbox{.}(2022)]%
        {LVHPWH22}
\bibfield{author}{\bibinfo{person}{Liyi Li}, \bibinfo{person}{Finn Voichick}, \bibinfo{person}{Kesha Hietala}, \bibinfo{person}{Yuxiang Peng}, \bibinfo{person}{Xiaodi Wu}, {and} \bibinfo{person}{Michael Hicks}.} \bibinfo{year}{2022}\natexlab{}.
\newblock \showarticletitle{Verified compilation of quantum oracles}.
\newblock \bibinfo{journal}{\emph{Proceedings of the ACM on Programming Languages}} \bibinfo{volume}{6}, \bibinfo{number}{OOPSLA2} (\bibinfo{year}{2022}), \bibinfo{pages}{589--615}.
\newblock
\href{https://doi.org/10.1145/3563309}{doi:\nolinkurl{10.1145/3563309}}


\bibitem[Linden and Popescu(1998)]%
        {LP98}
\bibfield{author}{\bibinfo{person}{Noah Linden} {and} \bibinfo{person}{Sandu Popescu}.} \bibinfo{year}{1998}\natexlab{}.
\newblock \bibinfo{title}{The halting problem for quantum computers}.
\newblock
\showeprint[arxiv]{quant-ph/9806054}~[quant-ph]


\bibitem[Liu et~al\mbox{.}(2023)]%
        {LWS+23}
\bibfield{author}{\bibinfo{person}{Chenxu Liu}, \bibinfo{person}{Meng Wang}, \bibinfo{person}{Samuel~A. Stein}, \bibinfo{person}{Yufei Ding}, {and} \bibinfo{person}{Ang Li}.} \bibinfo{year}{2023}\natexlab{}.
\newblock \bibinfo{title}{Quantum memory: a missing piece in quantum computing units}.
\newblock
\showeprint[arxiv]{2309.14432}~[quant-ph]


\bibitem[Lloyd et~al\mbox{.}(2014)]%
        {LMR14}
\bibfield{author}{\bibinfo{person}{Seth Lloyd}, \bibinfo{person}{Masoud Mohseni}, {and} \bibinfo{person}{Patrick Rebentrost}.} \bibinfo{year}{2014}\natexlab{}.
\newblock \showarticletitle{Quantum principal component analysis}.
\newblock \bibinfo{journal}{\emph{Nature physics}} \bibinfo{volume}{10}, \bibinfo{number}{9} (\bibinfo{year}{2014}), \bibinfo{pages}{631--633}.
\newblock


\bibitem[Low et~al\mbox{.}(2024)]%
        {LKS24}
\bibfield{author}{\bibinfo{person}{Guang~Hao Low}, \bibinfo{person}{Vadym Kliuchnikov}, {and} \bibinfo{person}{Luke Schaeffer}.} \bibinfo{year}{2024}\natexlab{}.
\newblock \showarticletitle{Trading {T} gates for dirty qubits in state preparation and unitary synthesis}.
\newblock \bibinfo{journal}{\emph{Quantum}}  \bibinfo{volume}{8} (\bibinfo{year}{2024}), \bibinfo{pages}{1375}.
\newblock
\href{https://doi.org/10.22331/q-2024-06-17-1375}{doi:\nolinkurl{10.22331/q-2024-06-17-1375}}


\bibitem[{L}ow and {W}iebe(2019)]%
        {LW19}
\bibfield{author}{\bibinfo{person}{Guang~Hao {L}ow} {and} \bibinfo{person}{Nathan {W}iebe}.} \bibinfo{year}{2019}\natexlab{}.
\newblock \bibinfo{title}{Hamiltonian simulation in the interaction picture}.
\newblock
\showeprint[arxiv]{1805.00675}~[quant-ph]


\bibitem[Lutz and Derby(1986)]%
        {LD86}
\bibfield{author}{\bibinfo{person}{Christopher Lutz} {and} \bibinfo{person}{Howard Derby}.} \bibinfo{year}{1986}\natexlab{}.
\newblock \showarticletitle{Janus: a time-reversible language}.
\newblock \bibinfo{journal}{\emph{Letter to Rolf Landauer}}  \bibinfo{volume}{2} (\bibinfo{year}{1986}).
\newblock


\bibitem[Miyadera and Ohya(2005)]%
        {MO05}
\bibfield{author}{\bibinfo{person}{Takayuki Miyadera} {and} \bibinfo{person}{Masanori Ohya}.} \bibinfo{year}{2005}\natexlab{}.
\newblock \showarticletitle{On halting process of quantum turing machine}.
\newblock \bibinfo{journal}{\emph{Open Systems \& Information Dynamics}} \bibinfo{volume}{12}, \bibinfo{number}{3} (\bibinfo{year}{2005}), \bibinfo{pages}{261--264}.
\newblock
\href{https://doi.org/10.1007/s11080-005-0923-2}{doi:\nolinkurl{10.1007/s11080-005-0923-2}}


\bibitem[{M}oore and {N}ilsson(2002)]%
        {MN02}
\bibfield{author}{\bibinfo{person}{Cristopher {M}oore} {and} \bibinfo{person}{Martin {N}ilsson}.} \bibinfo{year}{2002}\natexlab{}.
\newblock \showarticletitle{Parallel quantum computation and quantum codes}.
\newblock \bibinfo{journal}{\emph{SIAM J. Comput.}} \bibinfo{volume}{31}, \bibinfo{number}{2} (\bibinfo{year}{2002}), \bibinfo{pages}{799--815}.
\newblock
\href{https://doi.org/10.1137/S0097539799355053}{doi:\nolinkurl{10.1137/S0097539799355053}}


\bibitem[Myers(1997)]%
        {Myers97}
\bibfield{author}{\bibinfo{person}{John~M. Myers}.} \bibinfo{year}{1997}\natexlab{}.
\newblock \showarticletitle{Can a universal quantum computer be fully quantum?}
\newblock \bibinfo{journal}{\emph{Physical Review Letters}} \bibinfo{volume}{78}, \bibinfo{number}{9} (\bibinfo{year}{1997}), \bibinfo{pages}{1823}.
\newblock
\href{https://doi.org/10.1103/PhysRevLett.78.1823}{doi:\nolinkurl{10.1103/PhysRevLett.78.1823}}


\bibitem[Ofman(1963)]%
        {Ofman63}
\bibfield{author}{\bibinfo{person}{Yu Ofman}.} \bibinfo{year}{1963}\natexlab{}.
\newblock \showarticletitle{On the algorithmic complexity of discrete functions}. In \bibinfo{booktitle}{\emph{Sov. Math. Dokl.}}, Vol.~\bibinfo{volume}{7}. \bibinfo{pages}{589}.
\newblock


\bibitem[Ozawa(1998a)]%
        {Ozawa98}
\bibfield{author}{\bibinfo{person}{Masanao Ozawa}.} \bibinfo{year}{1998}\natexlab{a}.
\newblock \showarticletitle{Quantum nondemolition monitoring of universal quantum computers}.
\newblock \bibinfo{journal}{\emph{Physical Review Letters}} \bibinfo{volume}{80}, \bibinfo{number}{3} (\bibinfo{year}{1998}), \bibinfo{pages}{631}.
\newblock
\href{https://doi.org/10.1103/PhysRevLett.80.631}{doi:\nolinkurl{10.1103/PhysRevLett.80.631}}


\bibitem[Ozawa(1998b)]%
        {Ozawa98b}
\bibfield{author}{\bibinfo{person}{Masanao Ozawa}.} \bibinfo{year}{1998}\natexlab{b}.
\newblock \showarticletitle{Quantum {Turing} machines: local transition, preparation, measurement, and halting}.
\newblock In \bibinfo{booktitle}{\emph{Quantum Communication, Computing, and Measurement 2}}. \bibinfo{pages}{241--248}.
\newblock
\href{https://doi.org/10.1007/0-306-47097-7_32}{doi:\nolinkurl{10.1007/0-306-47097-7_32}}


\bibitem[Paradis et~al\mbox{.}(2021)]%
        {PBSV21}
\bibfield{author}{\bibinfo{person}{Anouk Paradis}, \bibinfo{person}{Benjamin Bichsel}, \bibinfo{person}{Samuel Steffen}, {and} \bibinfo{person}{Martin Vechev}.} \bibinfo{year}{2021}\natexlab{}.
\newblock \showarticletitle{Unqomp: synthesizing uncomputation in quantum circuits}. In \bibinfo{booktitle}{\emph{Proceedings of the 42nd ACM SIGPLAN International Conference on Programming Language Design and Implementation}}. \bibinfo{pages}{222--236}.
\newblock
\href{https://doi.org/10.1145/3453483.3454040}{doi:\nolinkurl{10.1145/3453483.3454040}}


\bibitem[Paradis et~al\mbox{.}(2024)]%
        {PBV24}
\bibfield{author}{\bibinfo{person}{Anouk Paradis}, \bibinfo{person}{Benjamin Bichsel}, {and} \bibinfo{person}{Martin Vechev}.} \bibinfo{year}{2024}\natexlab{}.
\newblock \showarticletitle{Reqomp: space-constrained uncomputation for quantum circuits}.
\newblock \bibinfo{journal}{\emph{Quantum}}  \bibinfo{volume}{8} (\bibinfo{year}{2024}), \bibinfo{pages}{1258}.
\newblock
\href{https://doi.org/10.22331/q-2024-02-19-1258}{doi:\nolinkurl{10.22331/q-2024-02-19-1258}}


\bibitem[Pius(2010)]%
        {Pius10}
\bibfield{author}{\bibinfo{person}{Einar Pius}.} \bibinfo{year}{2010}\natexlab{}.
\newblock \emph{\bibinfo{title}{Automatic parallelisation of quantum circuits using the measurement based quantum computing model}}.
\newblock \bibinfo{thesistype}{Master's\ thesis}. \bibinfo{school}{University of Edinburgh}.
\newblock


\bibitem[Rand et~al\mbox{.}(2019)]%
        {RPLZ19}
\bibfield{author}{\bibinfo{person}{Robert Rand}, \bibinfo{person}{Jennifer Paykin}, \bibinfo{person}{Dong-Ho Lee}, {and} \bibinfo{person}{Steve Zdancewic}.} \bibinfo{year}{2019}\natexlab{}.
\newblock \showarticletitle{{ReQWIRE}: reasoning about reversible quantum circuits}.
\newblock \bibinfo{journal}{\emph{Electronic Proceedings in Theoretical Computer Science}}  \bibinfo{volume}{287} (\bibinfo{year}{2019}), \bibinfo{pages}{299–312}.
\newblock
\href{https://doi.org/10.4204/eptcs.287.17}{doi:\nolinkurl{10.4204/eptcs.287.17}}


\bibitem[Raussendorf and Briegel(2001)]%
        {RB01}
\bibfield{author}{\bibinfo{person}{Robert Raussendorf} {and} \bibinfo{person}{Hans~J. Briegel}.} \bibinfo{year}{2001}\natexlab{}.
\newblock \showarticletitle{A one-way quantum computer}.
\newblock \bibinfo{journal}{\emph{Physical Review Letters}} \bibinfo{volume}{86}, \bibinfo{number}{22} (\bibinfo{year}{2001}), \bibinfo{pages}{5188}.
\newblock
\href{https://doi.org/10.1103/PhysRevLett.86.5188}{doi:\nolinkurl{10.1103/PhysRevLett.86.5188}}


\bibitem[{R}eif(1986)]%
        {Reif83}
\bibfield{author}{\bibinfo{person}{John~H. {R}eif}.} \bibinfo{year}{1986}\natexlab{}.
\newblock \showarticletitle{Logarithmic depth circuits for algebraic functions}.
\newblock \bibinfo{journal}{\emph{SIAM J. Comput.}} \bibinfo{volume}{15}, \bibinfo{number}{1} (\bibinfo{year}{1986}), \bibinfo{pages}{231--242}.
\newblock
\href{https://doi.org/10.1137/0215017}{doi:\nolinkurl{10.1137/0215017}}


\bibitem[Rosenthal(2023)]%
        {Rosenthal23}
\bibfield{author}{\bibinfo{person}{Gregory Rosenthal}.} \bibinfo{year}{2023}\natexlab{}.
\newblock \bibinfo{title}{Query and depth upper bounds for quantum unitaries via {Grover} search}.
\newblock
\showeprint[arxiv]{2111.07992}~[quant-ph]


\bibitem[Sabry et~al\mbox{.}(2018)]%
        {SVV18}
\bibfield{author}{\bibinfo{person}{Amr Sabry}, \bibinfo{person}{Beno{\^i}t Valiron}, {and} \bibinfo{person}{Juliana~Kaizer Vizzotto}.} \bibinfo{year}{2018}\natexlab{}.
\newblock \showarticletitle{From symmetric pattern-matching to quantum control}. In \bibinfo{booktitle}{\emph{Foundations of Software Science and Computation Structures: 21st International Conference, FOSSACS 2018}}. \bibinfo{pages}{348--364}.
\newblock
\href{https://doi.org/10.1007/978-3-319-89366-2_19}{doi:\nolinkurl{10.1007/978-3-319-89366-2_19}}


\bibitem[Selinger(2004)]%
        {Selinger04}
\bibfield{author}{\bibinfo{person}{Peter Selinger}.} \bibinfo{year}{2004}\natexlab{}.
\newblock \showarticletitle{Towards a quantum programming language}.
\newblock \bibinfo{journal}{\emph{Mathematical Structures in Computer Science}} \bibinfo{volume}{14}, \bibinfo{number}{4} (\bibinfo{year}{2004}), \bibinfo{pages}{527--586}.
\newblock
\href{https://doi.org/10.1017/S0960129504004256}{doi:\nolinkurl{10.1017/S0960129504004256}}


\bibitem[Shende et~al\mbox{.}(2005)]%
        {SBM05}
\bibfield{author}{\bibinfo{person}{Vivek~V. Shende}, \bibinfo{person}{Stephen~S. Bullock}, {and} \bibinfo{person}{Igor~L. Markov}.} \bibinfo{year}{2005}\natexlab{}.
\newblock \showarticletitle{Synthesis of quantum logic circuits}. In \bibinfo{booktitle}{\emph{Proceedings of the 2005 Asia and South Pacific Design Automation Conference}}. \bibinfo{pages}{272--275}.
\newblock
\href{https://doi.org/10.1109/TCAD.2005.855930}{doi:\nolinkurl{10.1109/TCAD.2005.855930}}


\bibitem[Shi(2002)]%
        {Shi02}
\bibfield{author}{\bibinfo{person}{Yu Shi}.} \bibinfo{year}{2002}\natexlab{}.
\newblock \showarticletitle{Remarks on universal quantum computer}.
\newblock \bibinfo{journal}{\emph{Physics Letters A}} \bibinfo{volume}{293}, \bibinfo{number}{5-6} (\bibinfo{year}{2002}), \bibinfo{pages}{277--282}.
\newblock
\href{https://doi.org/10.1016/S0375-9601(02)00015-4}{doi:\nolinkurl{10.1016/S0375-9601(02)00015-4}}


\bibitem[Smith et~al\mbox{.}(2017)]%
        {RMW17}
\bibfield{author}{\bibinfo{person}{Robert~S. Smith}, \bibinfo{person}{Michael~J. Curtis}, {and} \bibinfo{person}{William~J. Zeng}.} \bibinfo{year}{2017}\natexlab{}.
\newblock \bibinfo{title}{A practical quantum instruction set architecture}.
\newblock
\showeprint[arxiv]{1608.03355}~[quant-ph]


\bibitem[Sun et~al\mbox{.}(2023)]%
        {STYYZ23}
\bibfield{author}{\bibinfo{person}{Xiaoming Sun}, \bibinfo{person}{Guojing Tian}, \bibinfo{person}{Shuai Yang}, \bibinfo{person}{Pei Yuan}, {and} \bibinfo{person}{Shengyu Zhang}.} \bibinfo{year}{2023}\natexlab{}.
\newblock \showarticletitle{Asymptotically optimal circuit depth for quantum state preparation and general unitary synthesis}.
\newblock \bibinfo{journal}{\emph{IEEE Transactions on Computer-Aided Design of Integrated Circuits and Systems}} \bibinfo{volume}{42}, \bibinfo{number}{10} (\bibinfo{year}{2023}), \bibinfo{pages}{3301--3314}.
\newblock
\href{https://doi.org/10.1109/TCAD.2023.3244885}{doi:\nolinkurl{10.1109/TCAD.2023.3244885}}


\bibitem[{T}akahashi and {T}ani(2013)]%
        {TT13}
\bibfield{author}{\bibinfo{person}{Yasuhiro {T}akahashi} {and} \bibinfo{person}{Seiichiro {T}ani}.} \bibinfo{year}{2013}\natexlab{}.
\newblock \showarticletitle{Collapse of the hierarchy of constant{-}depth exact quantum circuits}. In \bibinfo{booktitle}{\emph{Proceedings of the 28th IEEE Conference on Computational Complexity}}. \bibinfo{pages}{168--178}.
\newblock
\href{https://doi.org/10.1109/CCC.2013.25}{doi:\nolinkurl{10.1109/CCC.2013.25}}


\bibitem[Tao et~al\mbox{.}(2022)]%
        {TSY+22}
\bibfield{author}{\bibinfo{person}{Runzhou Tao}, \bibinfo{person}{Yunong Shi}, \bibinfo{person}{Jianan Yao}, \bibinfo{person}{Xupeng Li}, \bibinfo{person}{Ali {Javadi-Abhari}}, \bibinfo{person}{Andrew~W. Cross}, \bibinfo{person}{Frederic~T. Chong}, {and} \bibinfo{person}{Ronghui Gu}.} \bibinfo{year}{2022}\natexlab{}.
\newblock \showarticletitle{Giallar: Push-button verification for the {Qiskit} quantum compiler}. In \bibinfo{booktitle}{\emph{Proceedings of the 43rd ACM SIGPLAN International Conference on Programming Language Design and Implementation}}. \bibinfo{pages}{641--656}.
\newblock
\href{https://doi.org/10.1145/3519939.3523431}{doi:\nolinkurl{10.1145/3519939.3523431}}


\bibitem[Terhal and DiVincenzo(2004)]%
        {TD04}
\bibfield{author}{\bibinfo{person}{Barbara~M. Terhal} {and} \bibinfo{person}{David~P. DiVincenzo}.} \bibinfo{year}{2004}\natexlab{}.
\newblock \showarticletitle{Adptive quantum computation, constant depth quantum circuits and arthur-merlin games}.
\newblock \bibinfo{journal}{\emph{Quantum Information \& Computation}} \bibinfo{volume}{4}, \bibinfo{number}{2} (\bibinfo{year}{2004}), \bibinfo{pages}{134–145}.
\newblock
\href{https://doi.org/10.26421/QIC4.2-5}{doi:\nolinkurl{10.26421/QIC4.2-5}}


\bibitem[Thomsen et~al\mbox{.}(2012)]%
        {TAG12}
\bibfield{author}{\bibinfo{person}{Michael~Kirkedal Thomsen}, \bibinfo{person}{Holger~Bock Axelsen}, {and} \bibinfo{person}{Robert Gl{\"u}ck}.} \bibinfo{year}{2012}\natexlab{}.
\newblock \showarticletitle{A reversible processor architecture and its reversible logic design}. In \bibinfo{booktitle}{\emph{Reversible Computation: Third International Workshop, RC 2011}}. \bibinfo{pages}{30--42}.
\newblock
\href{https://doi.org/10.1007/978-3-642-29517-1_3}{doi:\nolinkurl{10.1007/978-3-642-29517-1_3}}


\bibitem[{van den Hove}(2015)]%
        {vdH15}
\bibfield{author}{\bibinfo{person}{Gauthier {van den Hove}}.} \bibinfo{year}{2015}\natexlab{}.
\newblock \showarticletitle{On the origin of recursive procedures}.
\newblock \bibinfo{journal}{\emph{Comput. J.}} \bibinfo{volume}{58}, \bibinfo{number}{11} (\bibinfo{year}{2015}), \bibinfo{pages}{2892--2899}.
\newblock
\href{https://doi.org/10.1093/comjnl/bxu145}{doi:\nolinkurl{10.1093/comjnl/bxu145}}


\bibitem[Venev et~al\mbox{.}(2024)]%
        {VGDV24}
\bibfield{author}{\bibinfo{person}{Hristo Venev}, \bibinfo{person}{Timon Gehr}, \bibinfo{person}{Dimitar Dimitrov}, {and} \bibinfo{person}{Martin Vechev}.} \bibinfo{year}{2024}\natexlab{}.
\newblock \showarticletitle{Modular synthesis of efficient quantum uncomputation}.
\newblock \bibinfo{journal}{\emph{Proceedings of the ACM on Programming Languages}} \bibinfo{volume}{8}, \bibinfo{number}{OOPSLA2} (\bibinfo{year}{2024}), \bibinfo{pages}{2097--2124}.
\newblock
\href{https://doi.org/10.1145/3689785}{doi:\nolinkurl{10.1145/3689785}}


\bibitem[Vieri(1999)]%
        {Vieri99}
\bibfield{author}{\bibinfo{person}{Carlin~James Vieri}.} \bibinfo{year}{1999}\natexlab{}.
\newblock \emph{\bibinfo{title}{Reversible computer engineering and architecture}}.
\newblock \bibinfo{thesistype}{Ph.\,D. Dissertation}. \bibinfo{school}{Massachusetts Institute of Technology}.
\newblock


\bibitem[Voichick et~al\mbox{.}(2023)]%
        {VLRH23}
\bibfield{author}{\bibinfo{person}{Finn Voichick}, \bibinfo{person}{Liyi Li}, \bibinfo{person}{Robert Rand}, {and} \bibinfo{person}{Michael Hicks}.} \bibinfo{year}{2023}\natexlab{}.
\newblock \showarticletitle{Qunity: A unified language for quantum and classical computing}.
\newblock \bibinfo{journal}{\emph{Proceedings of the ACM on Programming Languages}}  \bibinfo{volume}{7} (\bibinfo{year}{2023}), \bibinfo{pages}{921--951}.
\newblock
\href{https://doi.org/10.1145/3571225}{doi:\nolinkurl{10.1145/3571225}}


\bibitem[Wang and Ying(2023)]%
        {WY23}
\bibfield{author}{\bibinfo{person}{Qisheng Wang} {and} \bibinfo{person}{Mingsheng Ying}.} \bibinfo{year}{2023}\natexlab{}.
\newblock \showarticletitle{Quantum random access stored-program machines}.
\newblock \bibinfo{journal}{\emph{J. Comput. System Sci.}}  \bibinfo{volume}{131} (\bibinfo{year}{2023}), \bibinfo{pages}{13--63}.
\newblock
\href{https://doi.org/10.1016/j.jcss.2022.08.002}{doi:\nolinkurl{10.1016/j.jcss.2022.08.002}}


\bibitem[Xu et~al\mbox{.}(2025)]%
        {XLD25}
\bibfield{author}{\bibinfo{person}{Shifan Xu}, \bibinfo{person}{Alvin Lu}, {and} \bibinfo{person}{Yongshan Ding}.} \bibinfo{year}{2025}\natexlab{}.
\newblock \showarticletitle{Fat-tree {QRAM}: A high-bandwidth shared quantum random access memory for parallel queries}. In \bibinfo{booktitle}{\emph{Proceedings of the 30th ACM International Conference on Architectural Support for Programming Languages and Operating Systems}} \emph{(\bibinfo{series}{ASPLOS '25}, Vol.~\bibinfo{volume}{2})}. \bibinfo{pages}{390–406}.
\newblock
\href{https://doi.org/10.1145/3676641.3716256}{doi:\nolinkurl{10.1145/3676641.3716256}}


\bibitem[Xu et~al\mbox{.}(2021)]%
        {XYV21}
\bibfield{author}{\bibinfo{person}{Zhaowei Xu}, \bibinfo{person}{Mingsheng Ying}, {and} \bibinfo{person}{Benoît Valiron}.} \bibinfo{year}{2021}\natexlab{}.
\newblock \bibinfo{title}{Reasoning about recursive quantum programs}.
\newblock
\showeprint[arxiv]{2107.11679}~[cs.LO]


\bibitem[Ying(2016)]%
        {Ying16}
\bibfield{author}{\bibinfo{person}{Mingsheng Ying}.} \bibinfo{year}{2016}\natexlab{}.
\newblock \bibinfo{booktitle}{\emph{Foundations of quantum programming}}.
\newblock \bibinfo{publisher}{Morgan Kaufmann}.
\newblock
\href{https://doi.org/10.1016/C2014-0-02660-3}{doi:\nolinkurl{10.1016/C2014-0-02660-3}}


\bibitem[Ying et~al\mbox{.}(2012)]%
        {YYF12}
\bibfield{author}{\bibinfo{person}{Mingsheng Ying}, \bibinfo{person}{Nengkun Yu}, {and} \bibinfo{person}{Yuan Feng}.} \bibinfo{year}{2012}\natexlab{}.
\newblock \bibinfo{title}{Defining quantum control flow}.
\newblock
\showeprint[arxiv]{1209.4379}~[quant-ph]


\bibitem[Ying and Zhang(2024)]%
        {YZ24}
\bibfield{author}{\bibinfo{person}{Mingsheng Ying} {and} \bibinfo{person}{Zhicheng Zhang}.} \bibinfo{year}{2024}\natexlab{}.
\newblock \bibinfo{title}{Verification of recursively defined quantum circuits}.
\newblock
\showeprint[arxiv]{2404.05934}~[quant-ph]


\bibitem[Yokoyama et~al\mbox{.}(2008)]%
        {YAG08}
\bibfield{author}{\bibinfo{person}{Tetsuo Yokoyama}, \bibinfo{person}{Holger~Bock Axelsen}, {and} \bibinfo{person}{Robert Gl{\"u}ck}.} \bibinfo{year}{2008}\natexlab{}.
\newblock \showarticletitle{Principles of a reversible programming language}. In \bibinfo{booktitle}{\emph{Proceedings of the 5th Conference on Computing Frontiers}}. \bibinfo{pages}{43--54}.
\newblock
\href{https://doi.org/10.1145/1366230.1366239}{doi:\nolinkurl{10.1145/1366230.1366239}}


\bibitem[Yokoyama and Gl{\"u}ck(2007)]%
        {YG07}
\bibfield{author}{\bibinfo{person}{Tetsuo Yokoyama} {and} \bibinfo{person}{Robert Gl{\"u}ck}.} \bibinfo{year}{2007}\natexlab{}.
\newblock \showarticletitle{A reversible programming language and its invertible self-interpreter}. In \bibinfo{booktitle}{\emph{Proceedings of the 2007 ACM SIGPLAN Symposium on Partial Evaluation and Semantics-Based Program Manipulation}}. \bibinfo{pages}{144--153}.
\newblock
\href{https://doi.org/10.1145/1244381.1244404}{doi:\nolinkurl{10.1145/1244381.1244404}}


\bibitem[Yuan and Carbin(2022)]%
        {YC22}
\bibfield{author}{\bibinfo{person}{Charles Yuan} {and} \bibinfo{person}{Michael Carbin}.} \bibinfo{year}{2022}\natexlab{}.
\newblock \showarticletitle{Tower: data structures in quantum superposition}.
\newblock \bibinfo{journal}{\emph{Proceedings of the ACM on Programming Languages}} \bibinfo{volume}{6}, \bibinfo{number}{OOPSLA2} (\bibinfo{year}{2022}), \bibinfo{pages}{259--288}.
\newblock
\href{https://doi.org/10.1145/3563297}{doi:\nolinkurl{10.1145/3563297}}


\bibitem[Yuan and Carbin(2024)]%
        {YC24}
\bibfield{author}{\bibinfo{person}{Charles Yuan} {and} \bibinfo{person}{Michael Carbin}.} \bibinfo{year}{2024}\natexlab{}.
\newblock \showarticletitle{The {T-complexity} costs of error correction for control flow in quantum computation}.
\newblock \bibinfo{journal}{\emph{Proceedings of the ACM on Programming Languages}} \bibinfo{volume}{8}, \bibinfo{number}{PLDI} (\bibinfo{year}{2024}), \bibinfo{pages}{492--517}.
\newblock
\href{https://doi.org/10.1145/3656397}{doi:\nolinkurl{10.1145/3656397}}


\bibitem[Yuan et~al\mbox{.}(2024)]%
        {YVC24}
\bibfield{author}{\bibinfo{person}{Charles Yuan}, \bibinfo{person}{Agnes Villanyi}, {and} \bibinfo{person}{Michael Carbin}.} \bibinfo{year}{2024}\natexlab{}.
\newblock \showarticletitle{Quantum control machine: The limits of control flow in quantum programming}.
\newblock \bibinfo{journal}{\emph{Proceedings of the ACM on Programming Languages}} \bibinfo{volume}{8}, \bibinfo{number}{OOPSLA1} (\bibinfo{year}{2024}), \bibinfo{pages}{1--28}.
\newblock
\href{https://doi.org/10.1145/3649811}{doi:\nolinkurl{10.1145/3649811}}


\bibitem[Yuan and Zhang(2023)]%
        {YZ23}
\bibfield{author}{\bibinfo{person}{Pei Yuan} {and} \bibinfo{person}{Shengyu Zhang}.} \bibinfo{year}{2023}\natexlab{}.
\newblock \showarticletitle{Optimal (controlled) quantum state preparation and improved unitary synthesis by quantum circuits with any number of ancillary qubits}.
\newblock \bibinfo{journal}{\emph{Quantum}}  \bibinfo{volume}{7} (\bibinfo{year}{2023}), \bibinfo{pages}{956}.
\newblock
\href{https://doi.org/10.22331/q-2023-03-20-956}{doi:\nolinkurl{10.22331/q-2023-03-20-956}}


\bibitem[Zhang et~al\mbox{.}(2022)]%
        {ZLY22}
\bibfield{author}{\bibinfo{person}{Xiao-Ming Zhang}, \bibinfo{person}{Tongyang Li}, {and} \bibinfo{person}{Xiao Yuan}.} \bibinfo{year}{2022}\natexlab{}.
\newblock \showarticletitle{Quantum state preparation with optimal circuit depth: Implementations and applications}.
\newblock \bibinfo{journal}{\emph{Physical Review Letters}} \bibinfo{volume}{129}, \bibinfo{number}{23} (\bibinfo{year}{2022}), \bibinfo{pages}{230504}.
\newblock
\href{https://doi.org/10.1103/PhysRevLett.129.230504}{doi:\nolinkurl{10.1103/PhysRevLett.129.230504}}


\bibitem[Zhang and Yuan(2024)]%
        {ZY24}
\bibfield{author}{\bibinfo{person}{Xiao-Ming Zhang} {and} \bibinfo{person}{Xiao Yuan}.} \bibinfo{year}{2024}\natexlab{}.
\newblock \showarticletitle{Circuit complexity of quantum access models for encoding classical data}.
\newblock \bibinfo{journal}{\emph{npj Quantum Information}} \bibinfo{volume}{10}, \bibinfo{number}{1} (\bibinfo{year}{2024}), \bibinfo{pages}{42}.
\newblock
\href{https://doi.org/10.1038/s41534-024-00835-8}{doi:\nolinkurl{10.1038/s41534-024-00835-8}}


\bibitem[Zhang et~al\mbox{.}(2021)]%
        {ZYY21}
\bibfield{author}{\bibinfo{person}{Xiao-Ming Zhang}, \bibinfo{person}{Man-Hong Yung}, {and} \bibinfo{person}{Xiao Yuan}.} \bibinfo{year}{2021}\natexlab{}.
\newblock \showarticletitle{Low-depth quantum state preparation}.
\newblock \bibinfo{journal}{\emph{Physical Review Research}} \bibinfo{volume}{3}, \bibinfo{number}{4} (\bibinfo{year}{2021}), \bibinfo{pages}{043200}.
\newblock
\href{https://doi.org/10.1103/PhysRevResearch.3.043200}{doi:\nolinkurl{10.1103/PhysRevResearch.3.043200}}


\bibitem[Zhang et~al\mbox{.}(2024)]%
        {ZWY24}
\bibfield{author}{\bibinfo{person}{Zhicheng Zhang}, \bibinfo{person}{Qisheng Wang}, {and} \bibinfo{person}{Mingsheng Ying}.} \bibinfo{year}{2024}\natexlab{}.
\newblock \showarticletitle{Parallel quantum algorithm for hamiltonian simulation}.
\newblock \bibinfo{journal}{\emph{Quantum}}  \bibinfo{volume}{8} (\bibinfo{year}{2024}), \bibinfo{pages}{1228}.
\newblock
\href{https://doi.org/10.22331/q-2024-01-15-1228}{doi:\nolinkurl{10.22331/q-2024-01-15-1228}}


\end{thebibliography}


\appendix

\newpage

\section{Quantum Recursive Programming Language \texorpdfstring{$\mathbf{RQC}^{++}$}{RQC++}}
\label{sec:the_high_level_language_QRPL}

In this appendix, we provide a more detailed introduction to the syntax and semantics of high-level language $\mathbf{RQC}^{++}$
for describing quantum recursive programs~\cite{YZ24}.
For the purpose of this paper, we also make some slight modifications and more illustrations, compared to the original definitions of $\mathbf{RQC}^{++}$ in~\cite{YZ24}.

Let us first describe several features of the language $\mathbf{RQC}^{++}$.
As aforementioned, it allows both quantum control flow and recursive procedure calls,
and therefore supports the quantum recursion. 
For simplicity, $\mathbf{RQC}^{++}$ is not explicitly typed.
A program in $\mathbf{RQC}^{++}$ describes a quantum circuit without measurements,
whose size is parameterised.
The alphabet of $\mathbf{RQC}^{++}$ contains classical and quantum variables, 
while classical variables solely serve for specifying the control of the programs.
The quantum control flow in $\mathbf{RQC}^{++}$ is fully managed by the $\mathbf{qif}$ statements:
quantum branches are only created by $\mathbf{qif}$,
and only merged by $\mathbf{fiq}$.

\subsection{Program Variables and Procedure Identifiers}
\label{sub:variables_and_procedures}

\subsubsection{Classical Variables}
\label{sub:classical_variables}

Classical variables in $\mathbf{RQC}^{++}$ are solely for specifying the control of programs.
They can be used to define formal parameters (of procedure declarations), store intermediate results, and express conditions (in $\mathbf{if}$ and $\mathbf{while}$ statements) and actual parameters (of procedure calls).
A program in $\mathbf{RQC}^{++}$ describes a quantum unitary transformation without measurements,
whose dimension can depend on classical parameters.
\textit{Classical variables are classical only in the eyes of the programmer};
or more specifically, in the eyes of the enclosing procedure.
Meanwhile, the procedure calls can be used in quantum superposition, 
e.g., within quantum branches created by the $\mathbf{qif}$ statements.
In the implementation, classical variables will be realised by the quantum hardware instead.

A classical variable $x$ has a type $T(x)$, which can be thought of as a set;
i.e., $x\in T\parens*{x}$.
In this paper, we will only consider three types of classical variables,
$\mathbf{Uint}$, $\mathbf{Int}$ and $\mathbf{Bit}$,
standing for unsigned integer type, integer type and bit type, respectively.
We use $\overline{x}=x_1x_2\ldots x_n$ to denote a list of classical variables.

\subsubsection{Quantum Variables}
\label{sub:quantum_variables}

Quantum variables are very different from classical variables.
The state of quantum variables can be in superposition,
and different quantum variables can be entangled.
An elementary quantum variable $q$ has a type $T(q)=\calH$,
which represents the corresponding Hilbert space of $q$.
The type (Hilbert space) of a list $\overline{q}$ of distinct quantum variables $q_1,\ldots,q_n$
is then the tensor product  $T(\overline{q})=\bigotimes_{i=1}^n \calH_i$,
where each $\calH_i=T(q_i)$ is the type of $q_i$.
In this paper, we will only consider two types of quantum variables,
$\mathbf{Qint}$ and $\mathbf{Qbit}$,
standing for quantum integer type and qubit type, respectively.

\subsubsection{Procedure Identifiers}
\label{sub:procedure_identifiers}

Procedure identifiers are the names of procedures and are of a designated type $\mathbf{Pid}$. 
For each procedure identifier $P$, 
we can associate with it a classical variable $P.\mathit{ent}$ of type $\mathbf{Uint}$,
storing the entry address of the procedure declaration of $P$.
The value of $P.\mathit{ent}$ is determined and static after the program is compiled and loaded into the memory.
In the compiled program, $P.\mathit{ent}$ will be used for handling procedure calls (see also \Cref{sub:mid_level_to_low_level_translation}).

\subsubsection{Arrays}
\label{sub:arrays}

Classical and quantum variables can all be generalised to array variables.
Procedure identifiers can be generalised to procedure arrays too.
An array can be subscripted by classical values.
For simplicity of presentation, in this paper,
we only consider one-dimensional arrays.
High-dimensional arrays can be easily simulated by one-dimensional arrays.

The type of an array depends on the type of its elements:
\begin{enumerate}
	\item 
		The type of a classical array $x$ is $\parens*{T\rightarrow X}\equiv X^{T}$,
		where $X$ is the type of the elements in $x$,
		and $T$ is the type of the subscript.
	\item
		The type of a quantum array $q$ is $\parens*{T\rightarrow \calH}\equiv \bigotimes_{t\in T} \calH_t$,\footnote{
		Here, we assume the elements in $T$ are ordered.}
		where $\calH_t=\calH$ is the type of the element,
		and $T$ is the type of the subscript.
	\item
		The type of a procedure array is $\parens*{T\rightarrow \mathbf{Pid}}\equiv \mathbf{Pid}^T$,
		where $T$ is the type of the subscript.
\end{enumerate}
Elements in a classical or quantum array are assigned contiguous addresses in the memory 
(see also \Cref{sub:symbol_table_and_memory_allocation_of_variables}),
and hence can be efficiently addressed.

Given an array variable, we can also write corresponding subscripted variables.
For classical array $x$ of type $T\rightarrow X$
and quantum array $q$ of type $T\rightarrow \calH$,
we can write subscripted variables $x[t]$, $q[t]$ for classical expression $t$ of type $T$, respectively.
They are of types $X$ and $\calH$, respectively.
For example, one can write subscripted quantum variable $q[5x+2y]$, where $x,y$ are classical variables.
Similarly, given a procedure array $P$, we can also write corresponding subscripted procedure identifiers $P[t]$.

For simplicity of presentation, 
we restrict the use of nested subscriptions.
In particular, for classical subscripted variable $x[t]$,
we require that $t$ contains no more subscripted variables.
For example, we do not allow subscripted variable $x[y[10]]$ for classical $x,y$,
but allow $q[x[7z]]$ for quantum $q$ and classical $x,z$.
The implementation of nested subscriptions can actually be treated in similar but more complicated ways.

Suppose $x$ is a classical array of type $T\rightarrow X$ with $T=\mathbf{Uint}$ or $\mathbf{Int}$.
We use the notation $x[k:l]$ to denote the restriction of $x$ to the interval $[k:l]=\braces*{k,k+1,\ldots,l-1,l}$.
Then, $x[k:l]$ is a variable of type $[k:l]\rightarrow X$.
The same convention applies to quantum and procedure arrays.


\subsubsection{Global Variables vs.\ Local Variables}
\label{sub:global_vs_local}

In $\mathbf{RQC}^{++}$, a variable is not declared before its use.
All variables are treated as global variables.
Local (classical) variables will be realised by the block statement 
\begin{equation*}
	\mathbf{begin}\ \mathbf{local}\ \overline{x}:=\overline{t}; \ldots\mathbf{end}.
\end{equation*}
Within the scope of the block,
the list of classical variables $\overline{x}$ are regarded as local variables,
initialised to new values specified by the list of expressions $\overline{t}$ at the begining of the block,
and restore their old values at the end of the block.
A consequence of this treatment is that in a procedure call,
the callee can use the variables setting up by the caller.

\subsection{Syntax}

The syntax of $\mathbf{RQC}^{++}$ is already summarised in \Cref{fig:syntax-QRPL}.
Here, we provide some more detailed explanations.

\subsubsection{Classical Assignment, If-Statement and Loop}

In $\mathbf{RQC}^{++}$,
the classical assignment, if-statement and loop are similar to their counterparts in classical programming languages.
\begin{itemize}
	\item 
		The assignment $\overline{x}:=\overline{t}$ 
		simultaneously assigns the values of the list of expressions $\overline{t}$ to the list of classical variables $\overline{x}$.
		Note that in $\overline{t}=t_1t_2\ldots t_n$, $t_i$ might contains variables in $\overline{x}=x_1x_2\ldots x_n$.
	\item
		The $\mathbf{if}\ b\ \mathbf{then}\ C_1\ \mathbf{else}\ C_2\ \mathbf{fi}$ statement chooses one of
		statements $C_1$ and $C_2$ to execute, depending on the value of the boolean expression $b$.
	\item
		The $\mathbf{while}\ b\ \mathbf{do}\ C\ \mathbf{od}$ statement repeatedly execute statement $C$,
		conditioned on that the value of the boolean expression $b$ is $1$.
\end{itemize}

\subsubsection{Block Statement}

The block statement is used to declare classical variables as local variables.
In particular, 
the statement 
\begin{equation*}
	\mathbf{begin}\ \mathbf{local}\ \overline{x}:=\overline{t};\ C\ \mathbf{end}
\end{equation*}
declares the list of variables $\overline{x}$ as local variables within the scope of $\mathbf{begin}\ldots \mathbf{end}$,
initialised to values of the list of expressions $\overline{t}$.
At the end of the block, $\overline{x}$ will restore the old values. 
The block statement is also useful in defining the semantics of procedure calls (see \Cref{sub:details_semantics}).

\subsubsection{Procedure Call}

The procedure call and procedure declaration are also similar to their classical counterparts.
\begin{itemize}
	\item 
		The procedure declaration $P\parens*{\overline{u}}\Leftarrow C$ declares a procedure identifier $P$
		with the list of formal parameters $\overline{u}$ and procedure body $C$.
	\item
		The procedure call $P\parens*{\overline{t}}$ calls the procedure with identifier $P$
		with the list of actual parameters $\overline{t}$.
\end{itemize}

\subsubsection{Quantum Unitary Gate}
\label{sub:quantum_unitary_gate}

The statement $U[\overline{q}]$ applies an elementary quantum unitary gate $U$ on the list $\overline{q}$ of quantum variables.
For simplicity of presentation, we restrict that $\overline{q}$ contains at most two quantum variables,
which are of type $\mathbf{Qbit}$.
The type of the unitary gate $U$ needs to be matched with $\overline{q}$.
If $\overline{q}=q_1$, then $U$ is a single-qubit unitary gate;
if $\overline{q}=q_1q_2$, then $U$ is a two-qubit unitary gate.
Here, for the two-qubit case, note that we need to promise $q_1$ and $q_2$ are distinct quantum variables, which will be formally stated in \Cref{sub:details_semantics}.
Additionally, we restrict that $U$ is chosen from a fixed set of elementary unitary gate set $\calG$ of size $\abs*{\calG}=O(1)$.
For example, $\calG=\braces*{H,T,\mathit{CNOT}}$.
More complicated unitaries, like parameterised rotation $R_X(\theta)$,
can be implemented by procedure call with classical parameters.

\subsubsection{Quantum Control Flow}
\label{sub:quantum_control_flow}

The quantum if-statement $\mathbf{qif}[q]\parens*{\ket{0}\rightarrow C_0}\square\parens*{\ket{1}\rightarrow C_1}\mathbf{fiq}$ in $\mathbf{RQC}^{++}$ explicitly manages quantum control flow.
Unlike the classical if-statement with classical control flow,
the control flow in the $\mathbf{qif}$ statement is essentially quantum.
Conditioned on the value ($\ket{0}$ or $\ket{1}$) of the external quantum coin $q$,
the statements $C_0$ and $C_1$ will be executed.
When $q$ is in a quantum superposition state,
the subsystems that $C_0$ and $C_1$ act on will be entangled with $q$.
The condition of the external quantum coin is to promise the physicality of the $\mathbf{qif}$ statement,
which is already stated in \Cref{cnd:qif-external-formal},
and will be further explained later when we come to the semantics of $\mathbf{RQC}^{++}$.

It is also worth pointing out that in $\mathbf{RQC}^{++}$, the quantum control flow 
is fully guarded by the $\mathbf{qif}$ statements.
In particular, quantum branches are only created by $\mathbf{qif}$,
and only merged by $\mathbf{fiq}$.

\subsection{Semantics}
\label{sub:details_semantics}

As aforementioned in \Cref{sub:semantics},
the operational semantics of $\mathbf{RQC}^{++}$
is defined in terms of transitions between configurations.
Let us fix a finite set of classical and quantum variables.
A configuration is represented by $\parens*{C, \sigma, \ket{\psi}}$,
where $C$ is the remaining statement to be executed or $C=\ \downarrow$ (standing for termination; and we denote $\downarrow; C'\equiv C'$),
$\sigma$ is the current classical state of all classical variables, and $\ket{\psi}$ is the current quantum state of all quantum variables.
Let $\calC$ be the set of configurations.
Then, the operational semantics is defined as  the transition relation $\rightarrow\ \in\calC\times \calC$,
of the form $\parens*{C,\sigma,\ket{\psi}}\rightarrow \parens*{C', \sigma',\ket{\psi'}}$, by a series of transition rules. 
The transition rules for defining the operational semantics of $\mathbf{RQC}^{++}$ were already presented in \Cref{fig:semantics-QRPL}.
Let us also further explain some of them as follows.

\begin{itemize}
    \item 
        In the (GA) rule, $\sigma(\overline{q})$ denotes the subsystem specified by $q$ with respect to the classical state $\sigma$.
        In particular, if $\overline{q}=q_1$, where $q_1$ is not subscripted, then $\sigma(\overline{q})=q_1$; if $\overline{q}=q_1$
        with $q_1=r[t_1]$ for quantum array $r$, then $\sigma(\overline{q})=r[\sigma(t_1)]$;
        if $\overline{q}=q_1q_2$ with $q_1=r[t_1]$ and $q_2=r'[t_2]$ for quantum arrays $r$ and $r'$, then $\sigma(\overline{q})=r[\sigma(t_1)]r'[\sigma(t_2)]$.
        The condition $\sigma\models \mathit{Dist}(\overline{q})$ means in the classical state $\sigma$, $\overline{q}$ is a list of distinct quantum variables. In particular, if $\overline{q}=q_1$,
        then $\mathit{Dist}(\overline{q})$ is always true; if $\overline{q}=q_1q_2$ with $q_1=r[t_1]$ and $q_2=r'[t_2]$ for quantum arrays $r$ and $r'$, then $\mathit{Dist}(\overline{q})$ is the logical formula $r=r'\rightarrow t_1\neq t_2$.
        For simplicity, we assume all programs considered always satisfy this condition.
    \item 
        In the (QIF) rule, we see how the condition of external quantum coin (\Cref{cnd:qif-external-formal},
        in which $\mathit{qv}$ will be defined in \Cref{def:qv} below) is used: it promises that $q$ is separated from the subsystems that $C_0$ and $C_1$ act on,
        and therefore the current state $\ket{\psi}$ can be written as $\alpha_0\ket{0}_{\sigma(q)}\ket{\theta_0}+\alpha_1\ket{1}_{\sigma(q)}\ket{\theta_1}$, given the classical state $\sigma$.
        
        As already pointed out in \Cref{sub:semantics}, another point to note is that the executions of $C_i$ in both quantum branches ($i=0,1$)
        terminate in the same classical state $\sigma$.
        This promises that classical variables are \textit{disentangled} from quantum variables
        after the $\mathbf{qif}$ statement.
        In \cite{YZ24}, originally, both quantum branches are only required to terminate in the same classical state $\sigma'$ that may differ from the initial $\sigma$.
        For simplicity of later implementation, here in \Cref{fig:semantics-QRPL},
        the requirement has been made slightly stricter, while remaining easy to meet in practice. 
        
        It is worth stressing again that classical variables in $\mathbf{RQC}^{++}$
        are classical only in the eyes of the programmer (or the enclosing procedure);
        they become quantum when the program is implemented on quantum hardware.
    \item
        In the (BS) rule, 
        the sequential composition $\overline{x}:=\overline{t}; C; \overline{x}:=\sigma\parens*{\overline{x}}$
        formalises the meaning of initialising $\overline{x}$ at the beginning of the block
        and restoring their old values (i.e., $\sigma\parens*{\overline{x}}$) at the end of the block.
    \item 
        In the (RC) rule, 
        the formal parameters $\overline{u}$ to be used by the callee will be locally replaced by the actual parameters $\overline{t}$ set by the caller. 
\end{itemize}

\subsubsection{Conditions for Well-defined Semantics}
\label{sub:qv-fcv}

A program written in the syntax in \Cref{fig:syntax-QRPL} is not yet promised to have well-defined semantics.
To guarantee that a program in $\mathbf{RQC}^{++}$ has well-defined semantics, recall that in \Cref{sub:conditions_for_well_defined_semantics} we have introduced three conditions (\Cref{cnd:qif-external-formal,cnd:no-fcv-qif,cnd:no-fcv-proc-body-formal}),
in which two notions remain to be formally defined.

The first notion is the quantum variables $\mathit{qv}\parens*{C,\sigma}$ in a statement $C$ with respect to a given classical state $\sigma$. It is used in \Cref{cnd:qif-external-formal}.
Before formally defining $\mathit{qv}$,
we need some additional notions.
We use $\calQ\calV$ to denote the set of all quantum variables.
Let $\calC'$ be the set of classical configurations of the form $(C,\sigma)$.
Let $\calF\equiv \calC' \rightarrow 2^{\calQ\calV}$ be the set of functions $f$
that maps a configuration $(C,\sigma)\in \calC'$ to a set of quantum variables $f(C,\sigma)\in 2^{\calQ\calV}$.
Then, we can define a partial order $\sqsubseteq$ on $\calF$,
such that for $f,g\in \calF$, $f\sqsubseteq g$ iff $\forall (C,\sigma)\in \calC'$, $f(C,\sigma)\subseteq g(C,\sigma)$.

\begin{definition}[Quantum variables]
	\label{def:qv}
	For a statement $C\in \mathbf{RQC}^{++}$ and a  classical state $\sigma$,
	the quantum variables of $C$ with respect to $\sigma$
    is denoted by $\mathit{qv}\parens*{C,\sigma}$, and $\mathit{qv}(\cdot,\cdot): \calC' \rightarrow 2^{\calQ\calV}$ is defined as the least element $f\in\calF$ (with respect to the order $\sqsubseteq$) that satisfies  the following conditions:
	\begin{enumerate}
		\item 
			$f\parens*{C,\sigma}\supseteq\emptyset$ if $C\equiv \mathbf{skip}\mid \overline{x}:=\overline{t}$.
		\item
			$f\parens*{U[\overline{q}],\sigma}\supseteq\overline{q}$.
		\item
			$f\parens*{C_1;C_2,\sigma}\supseteq f\parens*{C_1,\sigma}\cup f\parens*{C_2,\sigma'}$,
			where $\sigma'$ satisfies $\parens*{C_1,\sigma,\ket{\psi}}\rightarrow^* \parens*{\downarrow,\sigma',\ket{\psi}}$ for any quantum state $\ket{\psi}$.
		\item
			$f\parens*{\mathbf{if}\ b\ \mathbf{then}\ C_1\ \mathbf{else}\ C_2\ \mathbf{fi},\sigma}\supseteq
			\begin{cases}
				f\parens*{C_1,\sigma}, & \sigma\models b,\\
				f\parens*{C_2,\sigma}, & \sigma\models \neg b.
			\end{cases}$
		\item
			$f\parens*{\mathbf{while}\ b\ \mathbf{do}\ C\ \mathbf{od},\sigma}\supseteq
			\begin{cases}
				f\parens*{C;\mathbf{while}\ b\ \mathbf{do}\ C\ \mathbf{od},\sigma}, & \sigma\models b,\\
				\emptyset, &\sigma\models \neg b.
			\end{cases}$
		\item
            \label{stp:cnd-7}
			$f\parens*{\mathbf{qif}[q]\parens*{\ket{0}\rightarrow C_0}\square \parens*{\ket{1}\rightarrow C_1}\mathbf{fiq},\sigma}
			\supseteq\braces*{q}\cup f\parens*{C_0,\sigma}\cup f\parens*{C_1,\sigma}$.
		\item
			$f\parens*{\mathbf{begin}\ \mathbf{local}\ \overline{x}:=\overline{t};C\ \mathbf{end}, \sigma}\supseteq
			f\parens*{\overline{x}:=\overline{t};C;\overline{x}:=\sigma\parens*{\overline{x}},\sigma}$.
		\item
			$f\parens*{P\parens*{\overline{t}},\sigma}\supseteq f\parens*{\mathbf{begin}\ \mathbf{local}\ \overline{u}:=\overline{t};C\ \mathbf{end},\sigma}$, if $P\parens*{\overline{u}}\Leftarrow C\in \calP$.
	\end{enumerate}
\end{definition}

The following lemma shows that $\mathit{qv}(\cdot,\cdot)$ in \Cref{def:qv} is well-defined.

\begin{lemma}
    \label{lmm:least-ele} Let $S$ be the set of all functions $f\in\calF$ that  satisfy the conditions in \Cref{def:qv}. Then $S$ has the least element.
\end{lemma}

\begin{proof}
    It is easy to see $S\neq \emptyset$,
    because the constant function $h$ with $\forall (C,\sigma)\in \calC', h(C,\sigma)=\calQ\calV$
    is in $S$.  
    Let $g\in \calF$ be defined by $g(C,\sigma)=\bigcap_{f\in S} f(C,\sigma)$ for all $(C,\sigma)\in \calC'$. 
    Then it suffices to show that $g\in S$; that is, $g$ satisfies the conditions in \Cref{def:qv}. Here, we only prove that $g$ satisfies Condition~\ref{stp:cnd-7} in \Cref{def:qv}; other conditions can be verified similarly. 
    \begin{align*}
        g\parens*{\mathbf{qif}[q]\parens*{\ket{0}\rightarrow C_0}\square \parens*{\ket{1}\rightarrow C_1}\mathbf{fiq},\sigma}& =\bigcap_{f\in S}f\parens*{\mathbf{qif}[q]\parens*{\ket{0}\rightarrow C_0}\square \parens*{\ket{1}\rightarrow C_1}\mathbf{fiq},\sigma}\\ &\supseteq\bigcap_{f\in S}\parens*{\braces*{q}\cup f\parens*{C_0,\sigma}\cup f\parens*{C_1,\sigma}}\\ 
        &\supseteq\braces*{q}\cup\Big(\bigcap_{f\in S} f\parens*{C_0,\sigma}\Big)\cup\Big(\bigcap_{f\in S}f\parens*{C_1,\sigma}\Big)\\
        &=\braces*{q}\cup g\parens*{C_0,\sigma}\cup g\parens*{C_1,\sigma}
    \end{align*}

\end{proof}

The second notion is free changed (classical) variables $\mathit{fcv}\parens*{C,\sigma}$
in a statement $C$ with respect to a given classical state $\sigma$.
It is used in \Cref{cnd:no-fcv-qif,cnd:no-fcv-proc-body-formal}.
Similar to the case of defining $\mathit{qv}$, we also need some additional notions.
We use $\calC\calV$ to denote the set of all classical variables.
Let $\calG\equiv \calC'\rightarrow 2^{\calC\calV}$,
which also has a partial order $\sqsubseteq$ induced by the subset order $\subseteq$.

\begin{definition}[Free changed (classical) variables]
	\label{def:fcv}
	For a statement $C\in \mathbf{RQC}^{++}$ and a  classical state $\sigma$, the free changed (classical) variables in $C$ with respect to  $\sigma$ is 
	denoted by $\mathit{fcv}\parens*{C,\sigma}$, and $\mathit{fcv}(\cdot,\cdot): \calC'\rightarrow 2^{\calC\calV}$  is defined as the least element $f\in \calG$ (with respect to the order $\sqsubseteq$)  that satisfies  the following conditions:
	\begin{enumerate}
		\item 
			$f\parens*{C,\sigma}\supseteq \emptyset$ if $C\equiv \mathbf{skip}\mid U[\overline{q}]\mid P(\overline{t})$.
		\item
			$f\parens*{C,\sigma}\supseteq \overline{x'}$ if $C\equiv \overline{x}:=\overline{t}$,
			where $x'_i=y_i$ if $x_i=y_i$ for some basic classical variable $y_i$,
			and $x'_i=y_i\bracks*{\sigma\parens*{e_i}}$ if $x_i=y_i\bracks*{e_i}$ for some classical array variable $y_i$ and expression $e_i$.
		\item
			$f\parens*{C_1;C_2,\sigma}\supseteq f\parens*{C_1,\sigma}\cup f\parens*{C_2,\sigma'}$,
			where $\sigma'$ satisfies $\parens*{C_1,\sigma,\ket{\psi}}\rightarrow^* \parens*{\downarrow, \sigma',\ket{\psi}}$ for any quantum state $\ket{\psi}$.
		\item
			$f\parens*{\mathbf{if}\ b\ \mathbf{then}\ C_1\ \mathbf{else}\ C_2\ \mathbf{fi},\sigma}\supseteq 
			\begin{cases}
				f\parens*{C_1,\sigma}, & \sigma\models b,\\
				f\parens*{C_2,\sigma}, & \sigma\models \neg b.
			\end{cases}$
		\item
			$f\parens*{\mathbf{while}\ b\ \mathbf{do}\ C\ \mathbf{od},\sigma}\supseteq 
			\begin{cases}
				f\parens*{C;\mathbf{while}\ b\ \mathbf{do}\ C\ \mathbf{od},\sigma}, & \sigma\models b,\\
				\emptyset, &\sigma\models \neg b.
			\end{cases}$
		\item
			$f\parens*{\mathbf{qif}\bracks*{q}(\ket{0}\rightarrow C_0)\square (\ket{1}\rightarrow C_1)\mathbf{fiq},\sigma}\supseteq f\parens*{C_0,\sigma}\cup f\parens*{C_1,\sigma}$.
		\item
			$f\parens*{\mathbf{begin}\ \mathbf{local}\ \overline{x}:=\overline{t};\ C\ \mathbf{end}}\supseteq f\parens*{C}-\overline{x}$.
	\end{enumerate}
\end{definition}

We can also prove $\mathit{fcv}$ in \Cref{def:fcv} is well-defined, as shown in the following lemma.

\begin{lemma}
    Let $R$ be the set of all functions in $\calG$ that satisfy the conditions in \Cref{def:fcv}. Then $R$ has the least element.
\end{lemma}

\begin{proof}
    The proof is very similar to that of \Cref{lmm:least-ele} and thus omitted.
\end{proof}

\section{Details of Quantum Register Machine}
\label{sec:details_quantum_register_machine}

In this appendix, we present further details of the quantum register machine,
in addition to \Cref{sec:quantum_register_machine}.

\subsection{Elementary Operations on Registers}
\label{sub:elementary_operations_on_registers}

In \Cref{sub:quantum_registers} we have briefly mentioned that
the quantum register machine can perform a series of elementary operations on registers,
each taking time $T_{\textup{reg}}$.
Now let us make these elementary operations more specific as follows.

\begin{definition}[Elementary operations on registers]
	\label{def:ele-op-reg}
	The quantum register machine can perform the following elementary operations on registers.
	\begin{itemize}
		\item 
			(Reversible elementary arithmetic):
			For binary operator $*\in\braces*{\oplus,+,-}$, let
                \begin{equation*}
                    U_{*}\parens*{r_1,r_2}:=\sum_{x,y} \ket{x*y}\!\bra{x}_{r_1}\otimes \ket{y}\!\bra{y}_{r_2}
                \end{equation*}
			for distinct registers $r_1\neq r_2$.
            Here, $\oplus$ denotes the XOR operator.
			Also, let $U_{\textup{neg}}\parens*{r}:=\sum_{x}\ket{-x}\!\bra{x}_{r}$ for register $r$.
		\item
			(Reversible versions of possibly irreversible arithmetic):
			Let $\calO\calP$ be a fix set of concerned unary and binary arithmetic  operators 
			with $\abs*{\calO\calP}=O\parens*{1}$.
			For unary operator $\mathit{op}\in \calO\calP$, let
			$U_{\mathit{op}}\parens*{r_1,r_2}:=\sum_{x,y}\ket{x\oplus \parens*{\mathit{op}\ y}}\!\bra{x}_{r_1}\otimes
			\ket{y}\!\bra{y}_{r_2}$ for distinct registers $r_1,r_2$.

			For binary operator $\mathit{op}\in \calO\calP$, let
			$U_{\mathit{op}}\parens*{r_1,r_2,r_3}:=\sum_{x,y,z}\ket{x\oplus \parens*{y\ \mathit{op}\ z}}\!\bra{x}_{r_1}\otimes \ket{y}\!\bra{y}_{r_2}\otimes
			\ket{z}\!\bra{z}_{r_3}$ for distinct registers $r_1,r_2,r_3$.
		\item
			(Swap):
			Let the swap unitary
			$U_{\textup{swap}}\parens*{r_1,r_2}:=\sum_{x,y}\ket{y}\!\bra{x}_{r_1}\otimes\ket{x}\!\bra{y}_{r_2}$
			for distinct registers $r_1,r_2$.
		\item
			(Unitary gate):
			For unitary gate specified by $G\in \calG$ (see also \Cref{sub:quantum_unitary_gate}),
			let $U_G\parens*{r}$ be its corresponding unitary applied on register $r$.
		\item
			(Controlled operations):
			For a register $r$ and an elementary operation (unitary) $U$ not acting on $r$, let the controlled versions of $U$ be
			\begin{align*}
				\bullet\parens*{r}\textup{-}U:=\sum_{x\neq 0}\ket{x}\!\bra{x}_r \otimes U+\ket{0}\!\bra{0}\otimes \Id,\\
				\circ\parens*{r}\textup{-}U:=\sum_{x\neq 0}\ket{x}\!\bra{x}_r \otimes \Id +\ket{0}\!\bra{0}\otimes U.
			\end{align*}
	\end{itemize}
	Suppose that performing any of the above elementary operations takes time $T_{\textup{reg}}$.
\end{definition}

Elementary quantum operations in \Cref{def:ele-op-reg} will be eventually implemented at a lower level by standard quantum circuits composed of one- and two-qubit gates. Further analysis is shown in \Cref{sub:quantum_circuit_complexity_for_elementary_operations}.

\subsection{Details of the Low-Level Language QINS}
\label{sub:details_qins}

\begin{figure}
    \centering
    \includegraphics[width=0.3\textwidth]{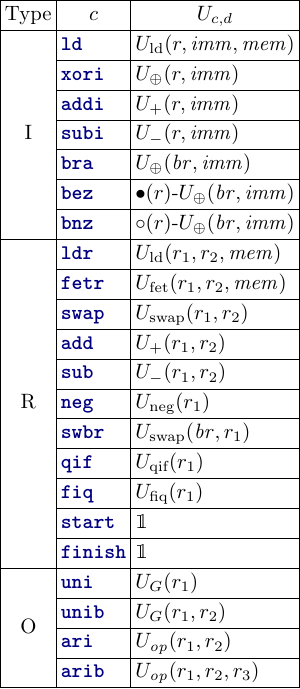}
    \caption{Full list of unitaries $U_{c,d}$ for implementing instructions in $\mathbf{QINS}$.}
    \label{fig:full-list-qins}
\end{figure}

In \Cref{sub:the_low_level_language_qins}, 
we have shown the list of instructions in the low-level language $\mathbf{QINS}$ 
and their classical effects in \Cref{fig:table-ins}.
We have also presented  several selected examples of instructions and their corresponding implementation unitaries in \Cref{fig:table-type-eg}. 
In this appendix, we provide further details.

Let us first have more discussions about the classical effects of instructions in \Cref{fig:table-ins}.

\begin{itemize}
	\item 
		Load instructions are used to retrieve information (including instructions and data) from the QRAM.
        In particular,
		\hlt{ld}\texttt{(r,i)} swaps the value $M_i$ at address $i$ specified by the immediate number and the value in register $r$.
		\hlt{ldr}\texttt{(r$_1$,r$_2$)} works similarly with the address specified by the value in register $r_2$.
        \hlt{fetr}\texttt{(r$_1$,r$_2$)} instead copies (via the XOR operator $\oplus$) the value $M_i$ at address $i$ specified by the immediate number into register $r$.
	\item
		Unitary gate instructions are used to apply elementary quantum gates 
		(chosen from the fixed set $\calG$, specified by $\mathbf{RQC}^{++}$ in \Cref{sub:quantum_unitary_gate}) on registers.
		In particular, \hlt{uni}\texttt{(G,r$_1$)} applies the single-qubit unitary gate specified by $G$ on  register $r_1$.
		\hlt{unib}\texttt{(G,r$_1$,r$_2$)} applies the two-qubit unitary gate specified by $G$ on registers $r_1$ and $r_2$.
	\item
		Arithmetic instructions are used to perform word-level elementary arithmetic operations,
		including reversible and (possibly) irreversible ones (see also \Cref{def:ele-op-reg}).

		Reversible arithmetic includes XOR (\hlt{xori} and \hlt{xor}),
		addition (\hlt{addi} and \hlt{add}), subtraction (\hlt{subi} and \hlt{sub}), negation (\hlt{neg})
		and swap (\hlt{swap}).
		These arithmetic operations can also take the immediate number $i$ as one of the inputs.

		A constant number of flexible choices of irreversible arithmetic operations is also allowed.
		In particular, \hlt{ari}\texttt{(op,r$_1$,r$_2$)} copies (by the XOR operator $\oplus$) the result of $\parens*{\mathit{op}\ r_2}$ into register $r_1$,
		where $\mathit{op}$ specifies an unary arithmetic operation.
		\hlt{arib}\texttt{(op,r$_1$,r$_2$,r$_3$)} works similarly with the $\mathit{op}$ specifying a binary arithmetic operation.
		In both cases, we assume $\mathit{op}\in \calO\calP$ for a fixed set $\calO\calP$ of $O(1)$ size (see also \Cref{def:ele-op-reg}).
		Here, the irreversible arithmetic operations are actually implemented reversibly by introducing garbage data.
	\item
		Branch instructions are used to realise transfer of control flow.
		In particular, \hlt{bra}\texttt{(i)} copies (by the XOR operator $\oplus$) 
        the branch offset $i$ specified by the immediate number $i$ into register $\mathit{br}$.
		\hlt{bez}\texttt{(r,i)} and \hlt{bnz}\texttt{(r,i)} are conditional versions of \hlt{bra}:
		the former performs \hlt{bra} when the value in register $r$ is zero;
		while the latter works in the non-zero case.
        \hlt{swbr}\texttt{(r)} swaps the values in register $r$ and register $\mathit{br}$.
	\item
		Qif instructions are used to update the registers $\mathit{qifv}$ and $\mathit{qifw}$.
        Their effects are already described in \Cref{sec:execution}.
	\item
		Special instructions are used to specify the start and finishing points of the main program.
\end{itemize}

Now we provide more details of implementing instructions in $\mathbf{QINS}$.
Recall that to execute an instruction stored in register $\mathit{ins}$,
the quantum register machine performs a unitary
\begin{equation*}
    U_{\textup{dec}} = \sum_{c} \ket{c}\!\bra{c} \otimes \sum_d \ket{d}\!\bra{d}\otimes U_{c,d},
\end{equation*}
which is previously defined in \Cref{fig:table-type-eg}.
Here, $c$ ranges over possible values of the section $\mathit{opcode}$ (see \Cref{fig:table-type-eg}) of $\mathit{ins}$, which corresponds to the name of the instruction in $\mathit{ins}$.
Depending on the type (I/R/O) of the instruction (indicated by $c$; see \Cref{fig:table-type-eg}), $d$ ranges over possible values of other sections in $\mathit{ins}$:
\begin{itemize}
    \item 
        If the instruction is of type I, then $d$ ranges over possible values of the section $\mathit{reg}$ in $\mathit{ins}$.
    \item
        If the instruction is of type R, then $d$ ranges over possible values of the sections
        $\mathit{reg}_1$ and $\mathit{reg}_2$ in $\mathit{ins}$.
    \item 
        If the instruction is of type O, then $d$ ranges over possible values of the sections $\mathit{para}$, $\mathit{reg}_1$, $\mathit{reg}_2$ and $\mathit{reg}_3$ in $\mathit{ins}$.
\end{itemize}

In \Cref{fig:full-list-qins}, we present
the full list of unitaries $U_{c,d}$ for implementing instructions in $\mathbf{QINS}$.
Most unitaries on the third columns are defined in \Cref{def:ele-op-reg}.
The unitaries $U_{\textup{qif}}$ and $U_{\textup{fiq}}$ for implementing instructions \hlt{qif} and \hlt{fiq} are already defined in \Cref{sec:execution}.

\section{Details of Compilation}
\label{sec:details_compilation}

In this appendix, we provide further details of the compilation of quantum recursive programs written in the language  $\mathbf{RQC}^{++}$.

\subsection{Details of High-Level Transformations}
\label{sub:details_high_level_transformations}

In \Cref{sub:high-level-transformations},
we have shown a series of high-level transformations to be performed in \Cref{fig:compdiagram},
and only selected the first two steps for illustration.
Now we delineate every step in the high-level transformations,
as well as the simplified program syntax after each step.
Note that the last four steps shown in \Cref{sub:remov_exp_cnd_para,sub:remove_simul_assignment,sub:remove_expr_in_subscription,sub:reduce_variables_in_expr} are rather simple and standard (see e.g., the textbook~\cite{AMSJ07}),
but we describe them in our context (adapted to support some features of the language $\mathbf{RQC}^{++}$) for completeness.

\subsubsection{Replacing Quantum Branches by Procedure Calls}

The first step is to replace the programs used in every quantum branch of the $\mathbf{qif}$ statements by procedure calls. We already showed this step in \Cref{sub:replace_qif_proc}.
For every $\mathbf{qif}$ statement, if $C_0,C_1$ are not procedure identifiers or $\mathbf{skip}$ statements, then we perform the replacement:
\begin{equation*}
	\mathbf{qif}[q](\ket{0}\rightarrow C_0)\square (\ket{1}\rightarrow C_1)\mathbf{fiq}
	\qquad\Rightarrow\qquad
	\mathbf{qif}[q](\ket{0}\rightarrow P_0)\square (\ket{1}\rightarrow P_1)\mathbf{fiq},
\end{equation*} 
where $P_0,P_1$ are fresh procedure identifiers,
and we add new procedure declarations $P_i\Leftarrow C_i$ (for $i\in \braces*{0,1}$) to $\calP$.
If only one of $C_0,C_1$ is procedure identifier or $\mathbf{skip}$,
then we only perform the replacement for the other branch.
Note that \Cref{cnd:qif-external-formal,cnd:no-fcv-qif,cnd:no-fcv-proc-body-formal} are kept after this step.
In particular, \Cref{cnd:no-fcv-qif} promises that $\mathit{fcv}(C_0,\sigma)=\mathit{fcv}(C_1,\sigma)=\emptyset$ for any classical state $\sigma$,
which therefore implies the newly introduced procedure declarations $P_i\Leftarrow C_i$ satisfy \Cref{cnd:no-fcv-proc-body-formal}.

After this step, the program $\calP=\braces*{P(\overline{u})\Leftarrow C}_P$ 
has the following simplified syntax:
\begin{equation}
    \label{eq:syntax-step-1}
	\begin{split}
		C::= &\ \mathbf{skip} \mid \overline{x}:=\overline{t}\mid U\bracks*{\overline{q}}\mid C_0; C_1\mid P(\overline{t})\mid \mathbf{if}\ b\ \mathbf{then}\ C_0\ \mathbf{else}\ C_1\ \mathbf{fi}\mid \mathbf{while}\ b\ \mathbf{do}\ C\ \mathbf{od}\\
				 &\mid \mathbf{begin}\ \mathbf{local}\ \overline{x}:=\overline{t}; C\ \mathbf{end}\mid\mathbf{qif}\bracks*{q}(\ket{0}\rightarrow P_0)\square (\ket{1}\rightarrow P_1)\mathbf{fiq}.
	\end{split}
\end{equation}

\subsubsection{Unrolling Nested Block Statements}

The second step is to unroll all nested block statements.
We already showed this step in \Cref{sub:unroll_nested_blocks}.
After this step, we are promised that the program no longer contains block statements,
but when the program is implemented, 
we need to perform uncomputation of classical variables at the end of every procedure body
for the program to preserve its original semantics.
To this end,
for every procedure declaration $P\parens*{\overline{u}}\Leftarrow C'\in \calP$
and every block statement $B$ appearing in $C'$,
we perform the replacement:
\begin{equation*}B\equiv 
	\mathbf{begin}\  \mathbf{local}\ \overline{x}:=\overline{t};\ C\ \mathbf{end}
	\qquad
	\Rightarrow
	\qquad
	\overline{x'}:=\overline{t};
	C\bracks*{x'/x},
\end{equation*}
where $\overline{x'}$ is a list of fresh variables, and $\bracks*{x'/x}$ stands for replacing variable $x$ by $x'$.
Moreover, we need to append $\overline{x'}:=\overline{0}$ 
at the beginning of $C'$.

After this step, the program $\calP=\braces*{P(\overline{u})\Leftarrow C}_P$ 
has the following simplified syntax:
\begin{equation}
    \label{eq:syntax-step-2}
	\begin{split}
		C::= &\ \mathbf{skip} \mid \overline{x}:=\overline{t}\mid U\bracks*{\overline{q}}\mid C_0; C_1\mid P(\overline{t})\mid \mathbf{if}\ b\ \mathbf{then}\ C_0\ \mathbf{else}\ C_1\ \mathbf{fi}\mid \mathbf{while}\ b\ \mathbf{do}\ C\ \mathbf{od}\\
				 &\mid \mathbf{qif}\bracks*{q}(\ket{0}\rightarrow P_0)\square (\ket{1}\rightarrow P_1)\mathbf{fiq}.
	\end{split}
\end{equation}
It is worth mentioning again that 
the program is technically in some new language with the same syntax as $\mathbf{RQC}^{++}$,
but whose semantics requires the uncomputation of classical variables at the end of every procedure body.
The uncomputation is to be automatically done in the high-to-mid-level translation (see \Cref{sub:high_level_to_mid_level_translation}). 

\subsubsection{Removing Expressions in Conditions and Parameters}
\label{sub:remov_exp_cnd_para}

The third step is to remove the expression $b$ in every $\mathbf{if}$ statement and every $\mathbf{while}$ loop,
and the list of expressions $\overline{t}$ in every procedure call $P\parens*{\overline{t}}$.
After this step, 
classical expressions only appear in the assignment statements $\overline{x}:=\overline{t}$
and the subscriptions of variables and procedure identifiers.
\begin{itemize}
	\item
		For every $\mathbf{if}$ statement, we perform the replacement:
		\begin{equation*}
                \mathbf{if}\ b\ \mathbf{then}\ C_0\ \mathbf{else}\ C_1\ \mathbf{fi}
			\qquad
			\Rightarrow 
			\qquad
                x:=b;\mathbf{if}\ x\ \mathbf{then}\ C_0\ \mathbf{else}\ C_1\ \mathbf{fi}
		\end{equation*}
		where $x$ is a fresh boolean variable.
	\item
		For every $\mathbf{while}$ statement, we perform the replacement:
		\begin{equation*}
            \mathbf{while}\ b\ \mathbf{do}\ C\ \mathbf{od}
			\qquad
			\Rightarrow
			\qquad
			x:=b;
			\mathbf{while}\ x\ \mathbf{do}\ C;x:=b\ \mathbf{od}
		\end{equation*}
		where $x$ is a fresh boolean variable.
	\item
		For every procedure call $P\parens*{\overline{t}}$, we perform the replacement:
		\begin{equation*}
			P\parens*{\overline{t}}
			\qquad\Rightarrow\qquad
			\overline{x}:=\overline{t}; P\parens*{\overline{x}},
		\end{equation*}
		where $\overline{x}$ is a list of fresh variables.
\end{itemize}

After this step, the program $\calP=\braces*{P(\overline{u})\Leftarrow C}_P$ 
has the following simplified syntax:
\begin{equation*}
	\begin{split}
		C::= &\ \mathbf{skip} \mid \overline{x}:=\overline{t}\mid U\bracks*{\overline{q}}\mid C_0; C_1\mid P(\overline{x})\mid \mathbf{if}\ x\ \mathbf{then}\ C_0\ \mathbf{else}\ C_1\ \mathbf{fi}\mid \mathbf{while}\ x\ \mathbf{do}\ C\ \mathbf{od}\\
				 &\mid \mathbf{qif}\bracks*{q}(\ket{0}\rightarrow P_0)\square (\ket{1}\rightarrow P_1)\mathbf{fiq}.
	\end{split}
\end{equation*}

\subsubsection{Removing Simultaneous Assignments}
\label{sub:remove_simul_assignment}

The fourth step is to remove every simultaneous assignment $\overline{x}:=\overline{t}$ 
by converting it to a series of unary assignments $x:=t$.
To do this, for every $\overline{x}:=\overline{t}$ with $\overline{x}\equiv x_1,x_2,\ldots,x_n$
and $\overline{t}\equiv t_1,t_2,\ldots,t_n$, we perform the replacement
\begin{equation*}
	\begin{split}
		&\overline{x}:=\overline{t}
	\end{split}
	\qquad
	\Rightarrow
	\qquad
	\begin{split}
		&x_1':=t_1;\ldots; x_n':=t_n;\\
		&x_1:=x_1';\ldots; x_n:=x_n';
	\end{split}
\end{equation*}
where $x_1',\ldots,x_n'$ are fresh variables.

After this step, the program $\calP=\braces*{P(\overline{u})\Leftarrow C}_P$ 
has the following simplified syntax:
\begin{equation*}
	\begin{split}
		C::= &\ \mathbf{skip} \mid x:=t\mid U\bracks*{\overline{q}}\mid C_0; C_1\mid P(\overline{x})\mid \mathbf{if}\ x\ \mathbf{then}\ C_0\ \mathbf{else}\ C_1\ \mathbf{fi}\mid \mathbf{while}\ x\ \mathbf{do}\ C\ \mathbf{od}\\
				 &\mid \mathbf{qif}\bracks*{q}(\ket{0}\rightarrow P_0)\square (\ket{1}\rightarrow P_1)\mathbf{fiq}.
	\end{split}
\end{equation*}

\subsubsection{Removing Expressions in Subscriptions}
\label{sub:remove_expr_in_subscription}

The fifth step is to remove expressions in the subscriptions of subscripted variables.
Note that subscripted quantum variables only appear in the $U[\overline{q}]$ statements 
and the $\mathbf{qif}[q]\parens*{\ket{0}\rightarrow C_0}\square\parens*{\ket{1}\rightarrow C_1}\mathbf{fiq}$ statements,
and subscripted procedure identifiers only appear in the procedure calls $P\parens*{\overline{x}}$.
By the previous steps of high-level transformations,
subscripted classical variables only appear in the assignments $x:=t$
and the subscriptions of quantum variables and procedure identifiers.

\begin{enumerate}
	\item
		For every unitary gate statement $U\bracks*{\overline{q}}$,
		suppose that subscripted quantum variables appearing in it are $q_1[e_1],\ldots,q_n[e_n]$ 
		(by our assumption in \Cref{sub:quantum_unitary_gate}, $n\leq 2$),
		where $q_1,\ldots,q_n$ are quantum arrays and $e_1,\ldots,e_n$ are expressions.
		We perform the replacement:
		\begin{equation*}
			\begin{split}
				&U\bracks*{\overline{q}}
			\end{split}
			\qquad\Rightarrow\qquad
			\begin{split}
				&x_1:=e_1;\ldots;x_n:=x_n;\\
				&\parens*{U\bracks*{q}}\bracks*{q_1[x_1]/q_1[e_1],\ldots,q_n[x_n]/q_n[e_n]},
			\end{split}
		\end{equation*}
		where $x_1,\ldots,x_n$ are fresh variables,
		and $\parens*{U\bracks*{\overline{q}}}\bracks*{q_1[x_1]/q_1[e_1],\ldots,q_n[x_n]/q_n[e_n]}$
		represents replacing every expression $e_i$ in the subscription of $q_i$ by variable $x_i$ in the statement $U\bracks*{\overline{q}}$.
    \item 
        For every $\mathbf{qif}[q]\parens*{\ket{0}\rightarrow C_0}\square\parens*{\ket{1}\rightarrow C_1}\mathbf{fiq}$ statement
        with $q=q'[e]$ for some quantum array $q'$ and expression $e$,
        we perform the replacement:
        \begin{equation*}
				\mathbf{qif}[q]\parens*{\ket{0}\rightarrow C_0}\square\parens*{\ket{1}\rightarrow C_1}\mathbf{fiq}
			\quad\Rightarrow\quad
            x:=e;\mathbf{qif}\bracks*{q'[x]}\parens*{\ket{0}\rightarrow C_0}\square\parens*{\ket{1}\rightarrow C_1}\mathbf{fiq},
		\end{equation*}
        where $x$ is a fresh variable. 
	\item
		For every procedure call $P\parens*{\overline{x}}$ 
		(note that at this point, actual parameters in procedure calls are variables)
		with $P=Q[t]$ for some procedure array $Q$ and expression $t$,
		we perform the replacement:
		\begin{equation*}
			P\parens*{\overline{x}}
			\qquad\Rightarrow\qquad
			y:=t; Q[y]\parens*{\overline{x}},
		\end{equation*}
		where $x$ is a fresh variable.
	\item 
		For every classical assignment $x:=t$,
		suppose that subscripted classical variables appearing in it are $y_1[e_1],y_2[e_2],\ldots,y_n[e_n]$,
		where $y_1,\ldots,y_n$ are classical arrays and $e_1,\ldots,e_n$ are expressions.
		We perform the replacement:
		\begin{equation*}
			\begin{split}
				&x:=t
			\end{split}
			\qquad\Rightarrow\qquad
			\begin{split}
				&z_1:=e_1;\ldots; z_n:=e_n;\\
				&\parens*{x:=t}\bracks*{y_1[x_1]/y_1[e_1],\ldots,y_n[x_n]/y_n[e_n]},
			\end{split}
		\end{equation*}
		where $z_1,\ldots,z_n$ are fresh variables,
		and $\parens*{x:=t}\bracks*{y_1[z_1]/y_1[e_1],\ldots,y_n[z_n]/y_n[e_n]}$
		represents replacing every expression $e_i$ in the subscription of $y_i$ by variable $z_i$ in the assignment $x:=t$.
		Note that here $e_1,\ldots,e_n$ no more contain subscripted classical variables
		because of our assumption in \Cref{sub:arrays}.
\end{enumerate}

After this step, every subscription in the program will be a classical variable.

\subsubsection{Reducing Variables in Expressions}
\label{sub:reduce_variables_in_expr}

The sixth step is to replace every assignments $x:=t$ with $t$ involving multiple variables
by a series of equivalent assignments
\begin{equation*}
	x_1:=t_1;\ \ldots\ x_n:=t_n;\ x:=t_{n+1},
\end{equation*}
such that every expression $t_i$ contains at most two variables;
that is, either $t_i\equiv a\ \mathit{op}\ b$, where $a$ and $b$ are variables or constants,
and $\mathit{op}\in \calO\calP$ (see also \Cref{def:ele-op-reg}) is an elementary binary operator;
or $t_i\equiv \mathit{op}\ a$, where $a$ is a variable or constant,
and $\mathit{op}\in \calO\calP$ is an elementary unary operator.
We will not bother describing the details of this standard conversion.

This step is the last step of the high-level transformations.
The transformed program is denoted by $\calP_h=\braces*{P\parens*{\overline{u}}\Leftarrow C}_P$,
and has the following simplified syntax:
\begin{equation*}
	\begin{split}
		C::= &\ \mathbf{skip} \mid x:=t\mid U\bracks*{\overline{q}}\mid C_0; C_1\mid P(\overline{x})\mid \mathbf{if}\ x\ \mathbf{then}\ C_0\ \mathbf{else}\ C_1\ \mathbf{fi}\mid \mathbf{while}\ x\ \mathbf{do}\ C\ \mathbf{od}\\
				 &\mid \mathbf{qif}\bracks*{q}(\ket{0}\rightarrow P_0)\square (\ket{1}\rightarrow P_1)\mathbf{fiq}.
	\end{split}
\end{equation*}
Moreover, every subscripted variable and procedure identifier has a basic classical variable as its subscription (e.g., $z[x], q[y], P[w]$),
and every expression has the form $t\equiv \mathit{op}\ x$ or $t\equiv x\ \mathit{op}\ y$.

\subsection{Details of High-Level to Mid-Level Translation}
\label{sub:further_details_of_high_level_to_mid_level_translation}

In \Cref{sub:high_level_to_mid_level_translation},
we have presented selected examples of the high-to-mid-level translation.
Now we provide further details of the translation,
as well as the full translation of the quantum multiplexor program previously shown in \Cref{fig:q-multiplexor} (part of the translation already presented in \Cref{fig:q-multiplexor-after}).

We use $\mathrm{mid}\braces*{D}$ to denote the high-to-mid-level translation of a statement (or declaration) $D$ in the high-level language $\mathbf{RQC}^{++}$.
The definition of $\mathrm{mid}\braces*{\cdot}$ is recursive, 
which further involves $\mathrm{init}\braces*{\cdot}$ and $\mathrm{uncp}\braces*{\cdot}$:
the former stands for the initialisation of formal parameters in procedure declarations,
and the latter stands for the uncomputation of classical variables 
(as required by the slightly changed program semantics after high-level transformations in \Cref{sub:details_high_level_transformations};
see also the beginning of \Cref{sub:high-level-transformations}).

In \Cref{fig:full-high-to-mid},
we present the full list of the high-to-mid-level translation $\mathrm{mid}\braces*{\cdot}$, as well as the initialisation $\mathrm{init}\braces*{\cdot}$ and the uncompuation $\mathrm{uncp}\braces*{\cdot}$. We further explain as follows.

\begin{itemize}
    \item 
        For $D\equiv \mathbf{skip}\mid C_1;C_2\mid U[q]\mid U[q_1q_2]$,
        the translation $\mathrm{mid}\braces*{D}$ and uncomputation $\mathrm{uncp}\braces*{D}$ of are self-explanatory. 
        Note that our uncomputation is only for classical variables.
    \item 
        For $D\equiv x:=\mathit{op}\ y\mid x:=y\ \mathit{op}\ z$,
        in the translation $\mathrm{mid}\braces*{D}$,
        the standard technique of introducing garbage data in reversible computing~\cite{Landauer61,Bennett73} is used.
        We first compute the new value of $x$ into a fresh variable $w$,
        then swap the values of $x$ and $w$. 
        Finally, the old value of $x$ will be pushed into the stack.
        The uncomputation $\mathrm{uncp}\braces*{D}$ is simply the inverse of $\mathrm{mid}\braces*{D}$.
    \item 
        For $D\equiv \mathbf{if}\ x\ \mathbf{then}\ C_1\ \mathbf{else}\ C_2\ \mathbf{fi}$, 
        the construction of $\mathrm{mid}\braces*{D}$
        is inspired by \cite{Axelsen11}.
        Note that by the high-level transformation in \Cref{sub:remov_exp_cnd_para},
        variable $x$ is not changed by $C_1$ and $C_2$.
        So, we can use the value of $x$ to determine which branch ($x=0,1$) the control flow has run through.
        Here, the mid-level instruction \hlt{brc} (a variant of \hlt{bra}) is used in pair with conditional branch instructions \hlt{bez} and \hlt{bnz}
        to indicate the variable $x$ for the condition,
        which will be helpful later in the mid-to-low-level translation.
        The uncomputation $\mathrm{uncp}\braces*{D}$ is similarly constructed.
    \item 
        For $D\equiv \mathbf{qif}[q]\parens*{\ket{0}\rightarrow C_0}\square\parens*{\ket{1}\rightarrow C_1}\mathbf{fiq}$,
        the translation $\mathrm{mid}\braces*{D}$ is similar to that of the $\mathbf{if}$ statement. 
        
        What is \textit{new} is our introduction of a pair of instructions \hlt{qif}\texttt{(q)} and \hlt{fiq}\texttt{(q)},
        which indicate the creation and join of quantum branching controlled by the quantum coin $q$.
        They are relevant to operations on the qif table in the partial evaluation 
        (described in \Cref{sec:partial-evaluation}) and the execution (described in \Cref{sec:execution}).
        Note that the uncomputation $\mathrm{uncp}\braces*{D}$ of the $\mathbf{qif}$ statement is $\emptyset$ due to \Cref{cnd:no-fcv-qif} in \Cref{sub:conditions_for_well_defined_semantics}:
        there are no free changed variables and therefore the uncomputation of classical variables is not required.
    \item
        For $D\equiv \mathbf{while}\ x\ \mathbf{do}\ C\ \mathbf{od}$,
        in the translation $\mathrm{mid}\braces*{D}$,
        a fresh variable $y$ is introduced to records the number of steps taken in the loop. Note that by the high-level transformation in \Cref{sub:remov_exp_cnd_para}, $x$ is not changed by $C$.
        So, we can use the value of $x$ and $y$ to determine where the control flow comes from.
        The construction of the uncomputation $\mathrm{uncp}\braces*{D}$
        exploits the same fresh variable $y$ as the one in $\mathrm{mid}\braces*{D}$.
    \item 
        For $D\equiv P(\overline{x})$, 
        in the translation $\mathrm{mid}\braces*{D}$,
        we first push $n$ actual parameters $x_1,\ldots,x_n$
        into the stack.
        Then, we branch to the declaration of the procedure $P$,
        whose address is specified by the variable $P.\mathit{ent}$.
        Finally, when the procedure call is returned,
        we pop the actual parameters from the stack. 
        Note that the procedure $P$ here can be subscripted,
        e.g., $P=Q[y]$ for some procedure array $Q$ and variable $y$.
        In this case, $P.\mathit{ent}=Q.\mathit{ent}[y]$ is also subscripted.
        The uncomputation $\mathrm{uncp}\braces*{D}=\emptyset$,
        due to \Cref{cnd:no-fcv-proc-body-formal} in \Cref{sub:conditions_for_well_defined_semantics}.
    \item 
        For $D\equiv P(\overline{u})\Leftarrow C$,
        we need to handle the initialisation of formal parameters $\overline{u}$ and the automatic uncomputation of classical variables at the end of the procedure body $C$ (see \Cref{sub:details_high_level_transformations}).
        The design of control flow in the translation $\mathrm{mid}\braces*{D}$
        is inspired by that in the translation of classical reversible languages~\cite{Axelsen11}.
        Specifically, the entry point and exit point to the declaration of procedure $P$ are the same, specified by $P.\mathit{ent}$.
        Instruction \hlt{swbr} is used to update the return offset register $\mathit{ro}$ (at both entry and exit), whose value will be temporarily pushed into the stack (see also \Cref{sub:QRAM-layout}) for supporting recursive procedure calls.
        
        What is \textit{new} here is our automatic uncomputation of classical variables, 
        done by $\mathrm{uncp}\braces*{C}$ at the end of the procedure body $C$.
    \item 
        The initialisation $\mathrm{init}\braces*{\overline{u}}$ of 
        formal parameters $\overline{u}$ is rather standard:
        we simply pop the actual parameters from the stack,
        then update the formal parameters,
        followed by pushing their old values into the stack. 
\end{itemize}

\begin{figure}
    \centering
    \includegraphics[width=0.73\linewidth]{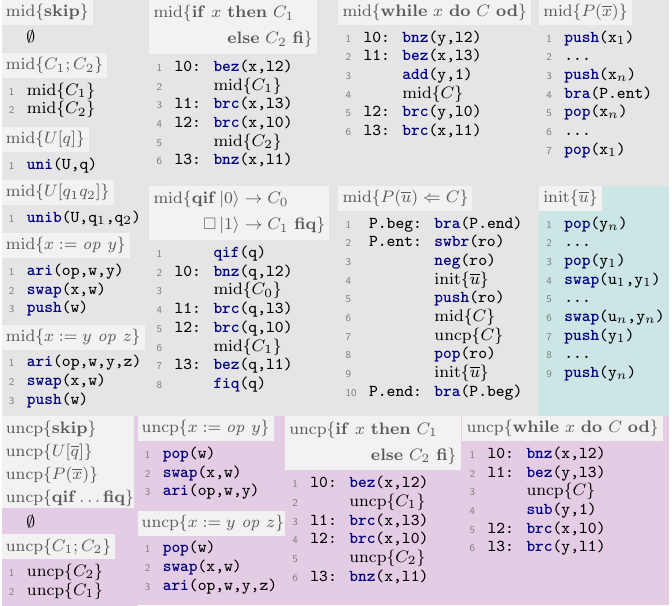}
    \caption{Full list of the high-to-mid-level translation $\mathrm{mid}\braces*{\cdot}$.
    Here, $\mathrm{init}\braces*{\cdot}$ and $\mathrm{uncp}\braces*{\cdot}$ denote the initialisation of formal parameters and uncomputation of classical variables,
    respectively.}
    \label{fig:full-high-to-mid}
\end{figure}

Now let us apply the high-to-mid-level translation to the transformed quantum multiplexor program on the LHS of \Cref{fig:q-multiplexor-after}.
This yields the translated program in \Cref{fig:q-multiplexor-full-mid},
which is annotated to show the correspondence with the program on the LHS of \Cref{fig:q-multiplexor-after}.
For simplicity of presentation, we have also done some manual optimisation on the translated program. 

\begin{figure}
    \centering
    \includegraphics[width=0.9\linewidth]{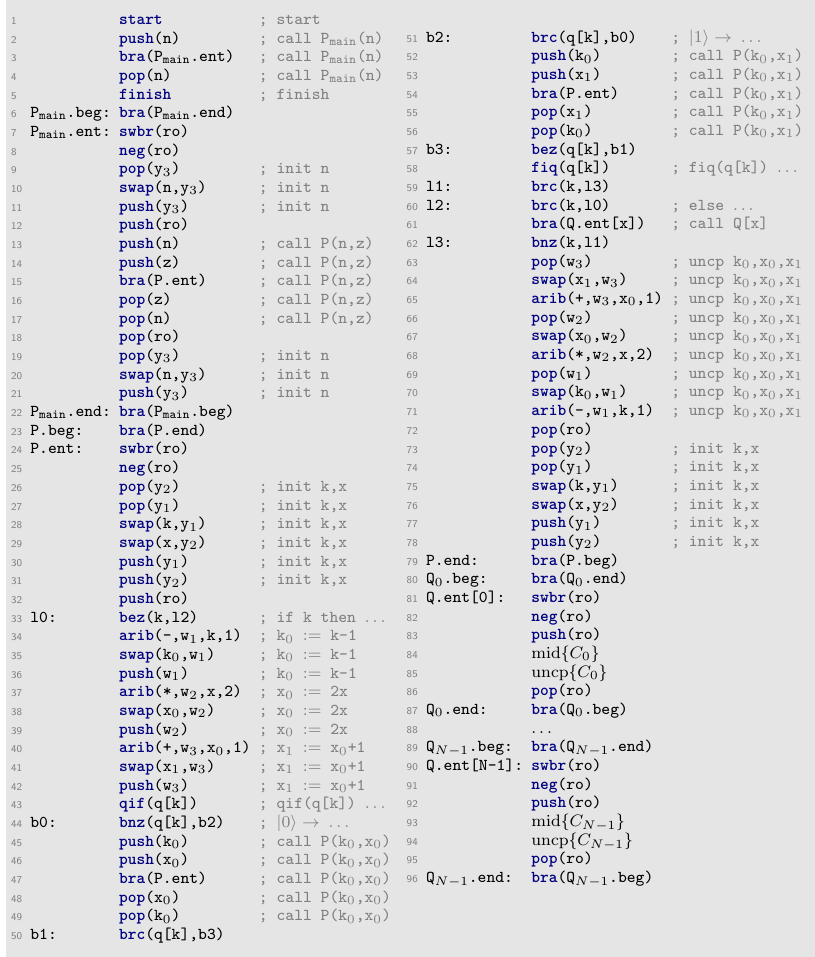}
    \caption{Full high-to-mid-level translation of the quantum multiplexor program.
    The original program is in \Cref{fig:q-multiplexor},
    with part of the high-to-mid-level translation previously presented in \Cref{fig:q-multiplexor-after}.}
    \label{fig:q-multiplexor-full-mid}
\end{figure}

\subsection{Symbol Table and Memory Allocation of Variables}
\label{sub:symbol_table_and_memory_allocation_of_variables}

In \Cref{sub:mid_level_to_low_level_translation},
we only briefly mentioned the memory allocation of variables and how the symbol table will be used at runtime.
In the following, we provide more details of the symbol table and memory allocation. 

After the high-level to mid-level translation, 
we can determine the names of all variables and procedure identifiers by scanning the whole program in the mid-level language,
because the mid-to-low-level translation will no more introduce new variables or procedure identifiers.
The addresses of all variables (corresponding to their names) 
will be stored in the symbol table and loaded into the QRAM for the addressing of variables at runtime.

For classical (either basic or array) variable $x$, 
we use $@ x$ to denote the address of the variable name $x$, which falls in the symbol table section of the QRAM (see \Cref{sub:QRAM-layout}).
The memory location at address $@ x$ stores $\& x$, which stands for the address of the variable $x$.
If $x$ is an array variable, then $\& x$ will be the base address of $x$;
i.e., $\& x=\& \parens*{x[0]}$.
As an address, $\& x$ falls in the classical variable section of the QRAM,
and stores the value of $x$.
The same convention applies to any quantum variable.

\begin{figure}
	\centering
	\includegraphics[width=0.9\textwidth]{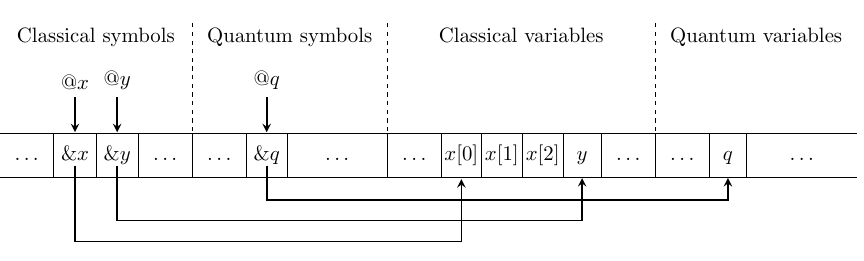}
	\caption{Example of the symbol table and memory allocation of variables in the QRAM.}
	\label{fig:symbol-table}
\end{figure}

An example of the symbol table is visualised in \Cref{fig:symbol-table}.
Here, the symbol table section is composed of classical and quantum symbols.
The memory location in the QRAM with address $@ x$ stores $\& x$,
which is the address of the memory location that stores $x[0]$, for classical array variable $x$.
Similarly, for quantum variable $q$,
the memory location with address $@ q$ stores $\& q$,
which is the address of the memory location that stores $q$.

Note that the sizes of classical and quantum arrays 
(except those classical variables $P.\mathit{ent}$ that correspond to procedure arrays $P$) 
are not determined at compile time
as the classical inputs are not yet given.
So, the values in the symbol table (i.e., addresses of all variables) 
need to be filled in after the classical inputs are considered 
(e.g., during the partial evaluation described in \Cref{sec:partial-evaluation}).
Still, for every classical variable $x$,
the address $@ x$ of the name $x$ is determined.
The same holds for every quantum variable.
Therefore, the compiled program will be completely determined and independent of classical and quantum inputs.

\subsection{Details of Mid-Level to Low-Level Translation}
\label{sub:further_details_mid_low_translation}

In \Cref{sub:mid_level_to_low_level_translation},
we have presented selected examples of the mid-to-low-level translation. In the following, we provide further details of the translation,
as well as the final compiled quantum multiplexor program (the original program was previously in \Cref{fig:q-multiplexor}).

Recall that instructions in the mid-level language can take variables and labels as inputs,
while the instructions in the low-level language $\mathbf{QINS}$ do not. 
To this end, we will use load instructions \hlt{ld}, \hlt{ldr} and \hlt{fetr} in $\mathbf{QINS}$
to first load variables from the QRAM into registers,
and then execute the instruction, followed by storing the results back to the QRAM.
Also, the mid-level language has additional instructions \hlt{push}, \hlt{pop} and \hlt{brc}
that need to be further translated.

\begin{figure}
    \centering
    \includegraphics[width=0.9\linewidth]{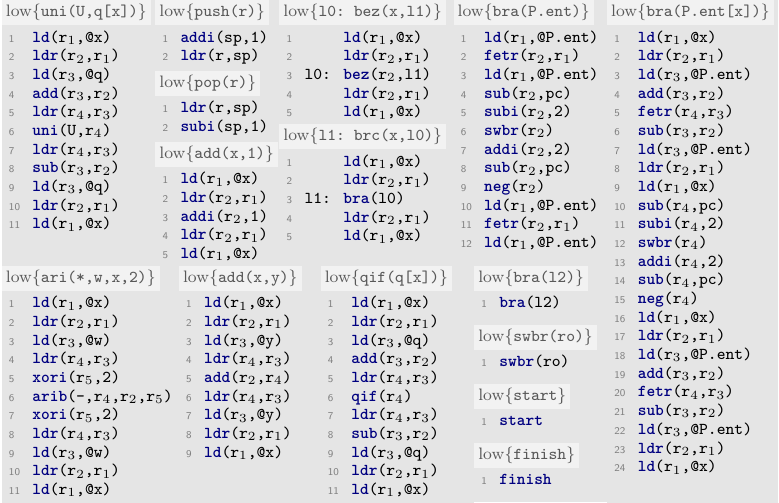}
    \caption{More examples of the mid-to-low-level translation. Here, all registers $r_i$ are free registers.}
    \label{fig:mid-to-low-full}
\end{figure}

We use $\mathrm{low}\braces*{i}$ to denote the mid-to-low-level translation of a mid-level instruction $i$. In \Cref{fig:mid-to-low-full}, we present more examples of the mid-to-low-level translation. 
Further explanations are as follows.
\begin{itemize}
    \item 
       In the translation $\mathrm{low}\braces*{\textup{\hlt{push}}\texttt{(r)}}$,
       Line~1 first increments the stack pointer $\mathit{sp}$ by $1$.
       Then, Line~2 loads the topmost stack element (whose address is specified by $\mathit{sp}$)
       into free register $r$.
       The translation $\mathrm{low}\braces*{\textup{\hlt{pop}}\texttt{(r)}}$
       works similarly.
    \item 
        In the translation $\mathrm{low}\braces*{\texttt{l0:\ }\textup{\hlt{bez}}\texttt{(x,l1)}}$,
        Lines~1--2 first load the value of variable $x$
        into free register $r_2$.
        Then, Line~3 branches to label $l0$,
        conditioned on $x\neq 0$. 
        Finally, Lines~4--5 reverse the effects of Lines~1--2.

        Note that instruction \texttt{l1:\ }\hlt{brc}\texttt{(x,l0)}
        is used in pair with instruction \texttt{l0:\ }\hlt{bez}\texttt{(x,l1)}.
        So, the free registers used in the 
        translations $\mathrm{low}\braces*{\texttt{l1:\ }\textup{\hlt{brc}}\texttt{(x,l0)}}$ and $\mathrm{low}\braces*{\texttt{l0:\ }\textup{\hlt{bez}}\texttt{(x,l1)}}$
        are also related (the same).
        This correspondence promises that registers $r_1,r_2$ are properly cleared (hence become free again) after being used.
    \item 
        In the translation $\mathrm{low}\braces*{\textup{\hlt{add}}\texttt{(x,1)}}$,
        Lines~1--2 first load the value of variable $x$
        into free register $r_2$.
        Then, Line~3 adds $1$ to $r_2$ to obtain $x+1$.
        Finally, Lines~4--5 reverse the effects of Lines~1--2.
        Note that we use \hlt{addi} in Line~3, because $1$ is an immediate number.
        
        In comparison, in the translation $\mathrm{low}\braces*{\textup{\hlt{add}}\texttt{(x,y)}}$,
        we use \hlt{add} in Line~5, because both $x,y$ are variables and loaded into registers.
    \item 
        In the translation $\mathrm{low}\braces*{\textup{\hlt{ari}}\texttt{(*,w,x,2)}}$,
        Lines~1--4 first load the values of $w$ and $x$ into free registers $r_2$ and $r_4$, respectively. Then, Line~5 prepares the immediate number $2$ in free register $r_5$.
        Line~6 performs the binary arithmetic operation $*$ and puts the result $x*2$
        into register $r_4$.
        Finally, Lines~7--11 reverse the effects of Lines~1--4
    \item 
        For instructions like \hlt{bra}\texttt{(l2)},
        \hlt{swbr}\texttt{(ro)}, \hlt{start} and \hlt{finish} that are already in $\mathbf{QINS}$, no further translation is needed.
    \item 
        In the translation $\mathrm{low}\braces*{\textup{\hlt{bra}}\texttt{(P.ent[x])}}$,
        Lines~1--2 first load the value of $x$ into free register $r_2$.
        Then, the base address $\& P.\mathit{ent}$ of classical array $P.\mathit{ent}$ is loaded into $r_3$ by Line~3,
        which is added to $r_2$ by Line~4 to obtain the address $\& P.\mathit{ent}[x]$ of the subscripted variable $P.\mathit{ent}[x]$. 
        Line~5 loads the value of $P.\mathit{ent}[x]$ to $r_4$, where \hlt{fetr} instead of \hlt{ldr} is used to preserve the copy of $P.\mathit{ent}[x]$ in the QRAM for supporting recursive procedure calls.
        Lines~6--9 reverse the effects of Lines~1--4.

        Lines~10--11 calculate in $r_4$ the offset of $P.\mathit{ent}[x]$ from the address of Line~12,
        which branches to the declaration of the subscripted procedure $P[x]$.
        Note that after Line~12, all previously used registers are cleared and become free again.
        Finally, Lines~13--24 reverse the effects of Lines~1--12.
\end{itemize}

Now let us apply the mid-to-low-level translation to the quantum multiplexor program in \Cref{fig:q-multiplexor-full-mid} after the high-to-mid-level translation. 
This yields the full compiled program in \Cref{fig:qmuxlow-full},
which is also annotated to show the correspondence with the mid-level program in \Cref{fig:q-multiplexor-full-mid}.
As usual, for simplicity of presentation,
manual optimisation has been done on the translated program.

\begin{figure}
    \centering
    \includegraphics[width=0.9\linewidth]{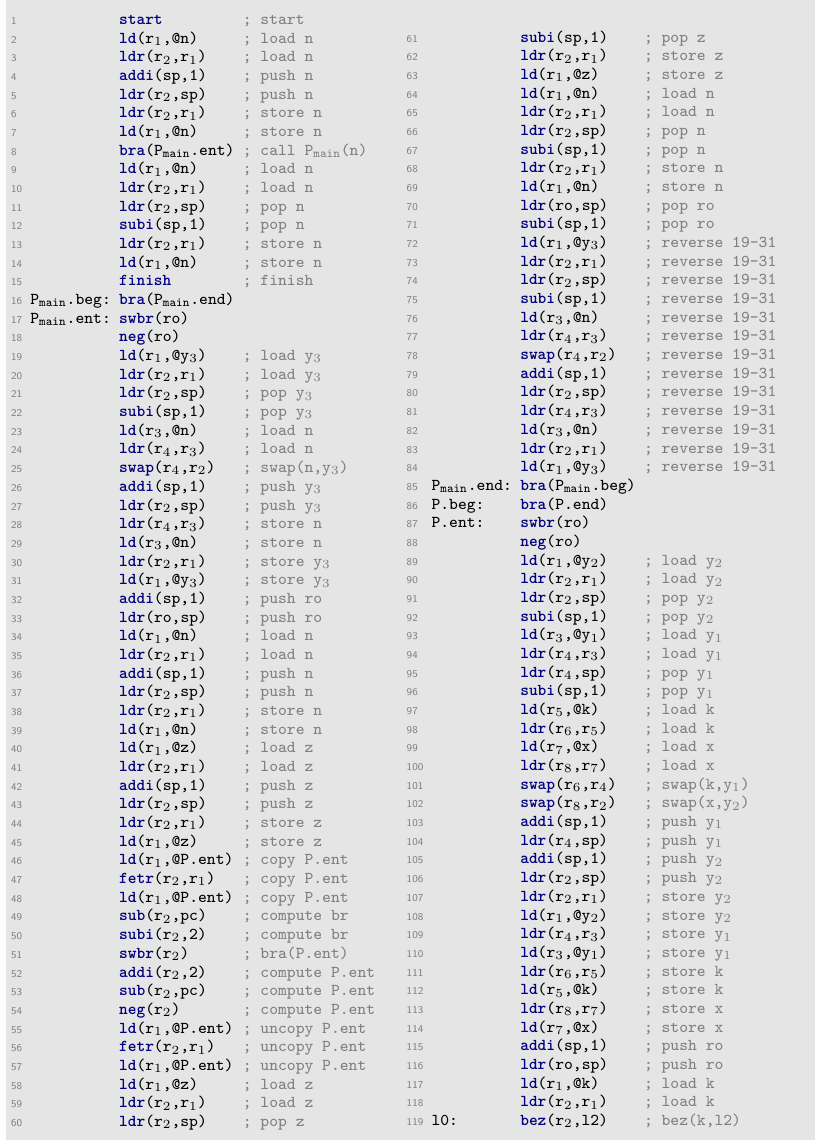}
    \caption{Full compiled quantum multiplexor program, after the mid-to-low-level translation.}
    \label{fig:qmuxlow-full}
\end{figure}

\begin{figure}
    \ContinuedFloat
    \centering
    \includegraphics[width=0.9\linewidth]{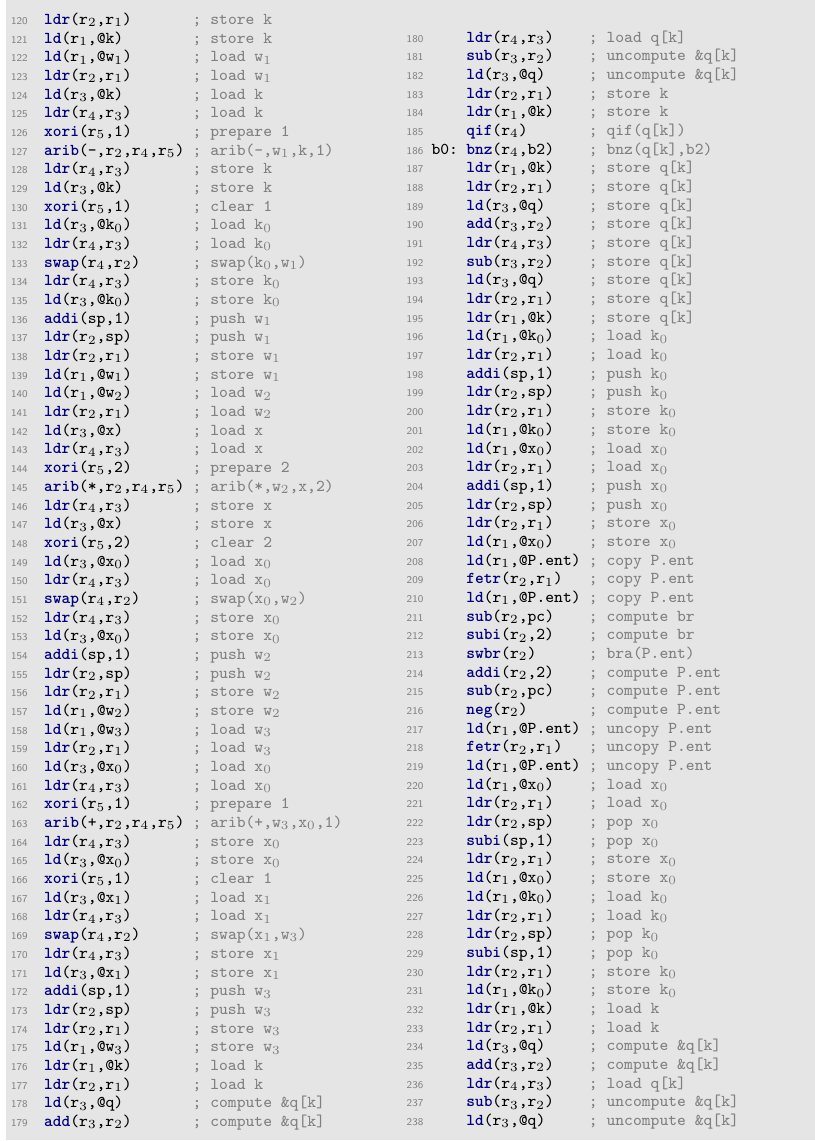}
    \caption{Full compiled quantum multiplexor program, after the mid-to-low-level translation (cont.)}
\end{figure}

\begin{figure}
    \ContinuedFloat
    \centering
    \includegraphics[width=0.9\linewidth]{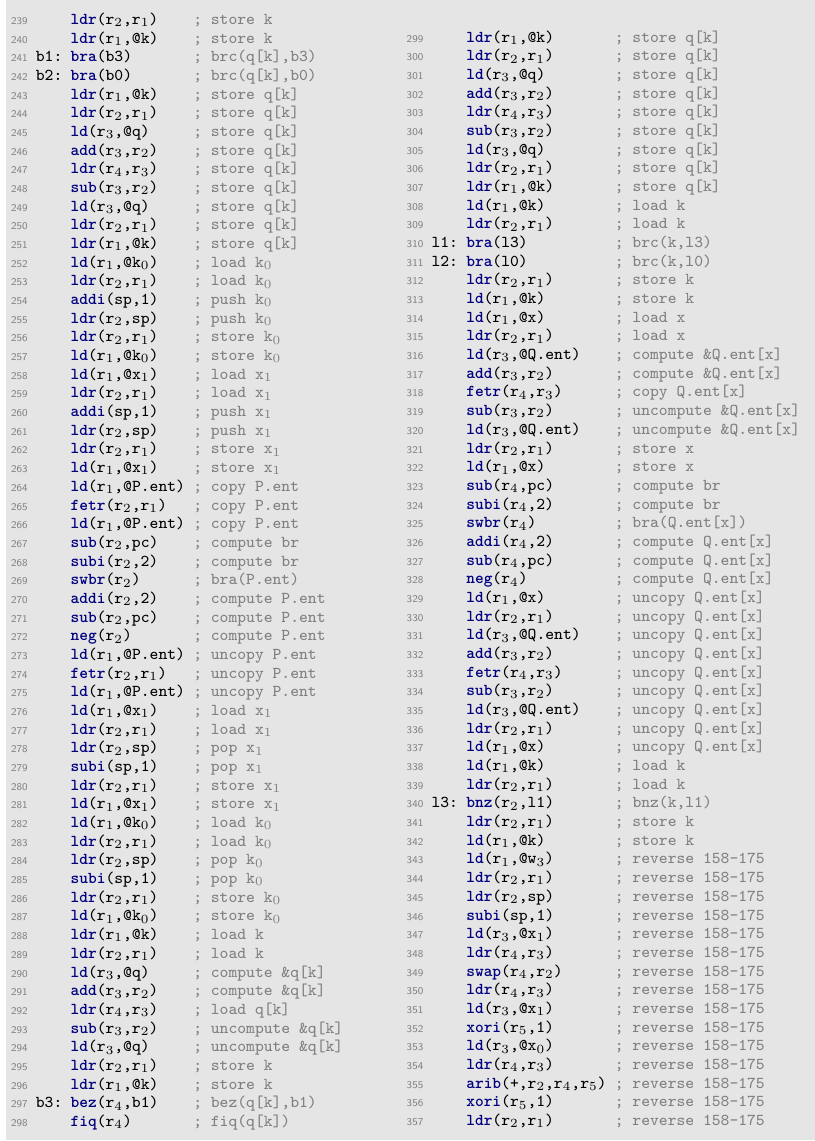}
    \caption{Full compiled quantum multiplexor program, after the mid-to-low-level translation (cont.)}
\end{figure}

\begin{figure}
    \ContinuedFloat
    \centering
    \includegraphics[width=0.9\linewidth]{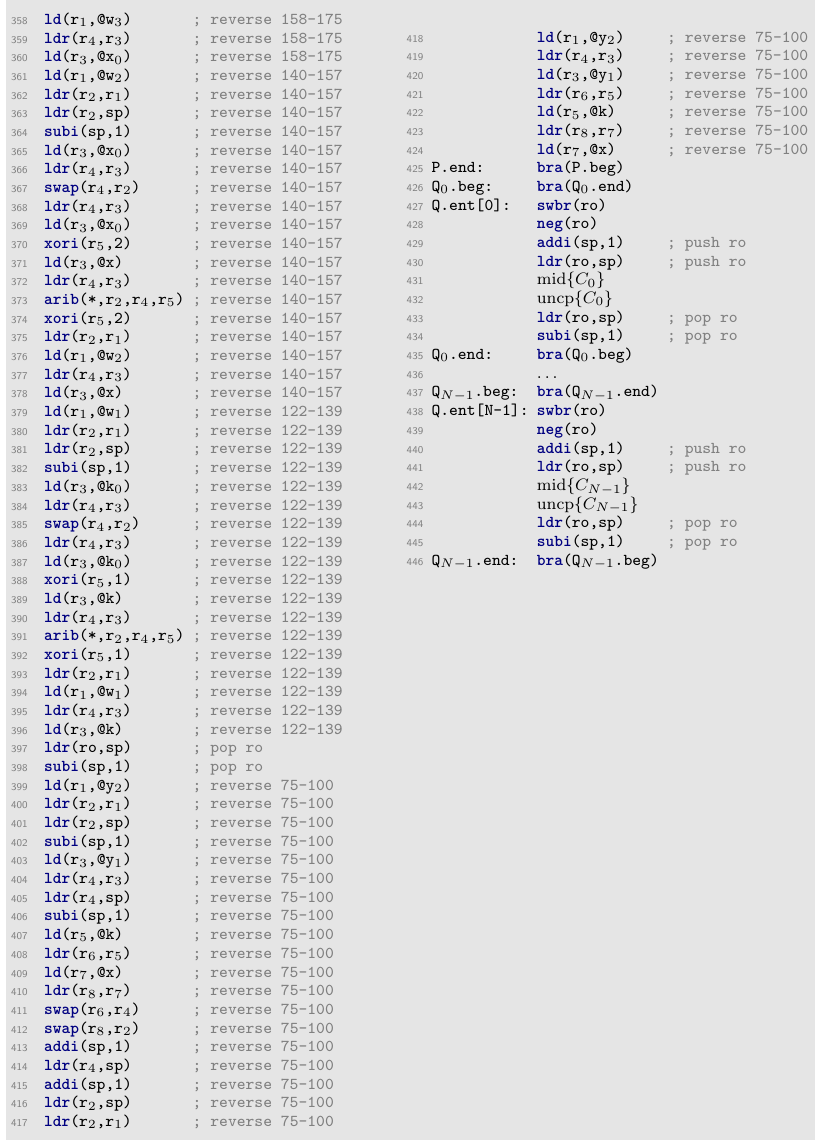}
    \caption{Full compiled quantum multiplexor program, after the mid-to-low-level translation (cont.)}
\end{figure}

\section{Details of Partial Evaluation}
\label{sec:details_qif_table}

In this appendix, we present further details of the partial evaluation of compiled quantum recursive programs,
in addition to \Cref{sec:partial-evaluation}.

\subsection{Memory Allocation of Qif Table}
\label{sub:memory_allocation_of_qif_table}

In \Cref{sub:the_synchronisation_problem}, we have identified the synchronisation problem
in executing quantum recursive programs, 
and defined in \Cref{sub:qif_table} a data structure called qif table that will be used in superposition at runtime to address this problem. 
The qif table generated from the partial evaluation needs to be loaded into the QRAM at runtime. 
To store a qif table in the QRAM,
we need an encoding of every node and its relevant information (links and counter).
The simplest way to encode is to gather all information of a node 
into a $9$-tuple 
\begin{equation}
    \parens*{v.w,v.\mathit{nx},v.\mathit{fc}_0,v.\mathit{fc}_1,v.\mathit{lc}_0,v.\mathit{lc}_1,v.\mathit{pr},v.\mathit{cf},v.\mathit{cl}}.
    \label{eq:qif-table-entry}
\end{equation}
Here, we assume every node $v$ is identified by its base address of the above tuple in the QRAM.
For nodes of type $\circ$, some of the entries in \Cref{eq:qif-table-entry} might be empty and set to $0$.
In this way, the whole qif table is encoded into an array of tuples, each corresponding to a node within. 
Accessing the information of a node can be done in a way similar to accessing an array element.
For example, if every entry in \Cref{eq:qif-table-entry} occupies a word,
then given the base address $a$ of a tuple that corresponds to a node $v$,
one can access the information $v.\mathit{nx}$ by the address $a+1$.
In \Cref{sec:execution}, the quantum register machine might need to access the value of $v.\mathit{nx}$ (see e.g., \Cref{alg:U-qif-fiq}),
which in this case can be fetched into a free register $r$ for later computation, by applying the unitary $U_{\mathit{fet}}(r, a+1, \mathit{mem})$.

It is worth mentioning that more compact ways (with smaller space complexity) of encoding the qif table other than the straightforward encoding in \Cref{eq:qif-table-entry} exist
but are not discussed here, for simplicity of presentation.

\subsection{Qif Table for the Quantum Multiplexor Program}
\label{sub:qif_table_qmux}

In \Cref{sub:qif_table}, we have presented an example of an artifitial program and its corresponding qif table in \Cref{fig:eg-qif-table}. For illustration, 
in \Cref{fig:qmux-qif-table}, we show the qif table for the quantum multiplexor program as another example. 
The original program written in $\mathbf{RQC}^{++}$ was already presented in \Cref{fig:q-multiplexor}, while the full compiled program was shown in \Cref{fig:qmuxlow-full}.
Here, we choose $n=3$.
As usual, we only show the links $\mathit{fc}_i$ (colored in black), $\mathit{lc}_i$ (colored in red and dashed) 
and $\mathit{nx}$ (colored in blue and squiggled),
while $w$, $\mathit{cf}$, $\mathit{cl}$ and $\mathit{pr}$ are omitted for simplicity.
For the quantum multiplexor program,
the qif table is just a tree of depth $3$,
because there are only nested instantiations of the $\mathbf{qif}$ statements. 

\begin{figure}
    \centering
    \includegraphics[width=0.5\linewidth]{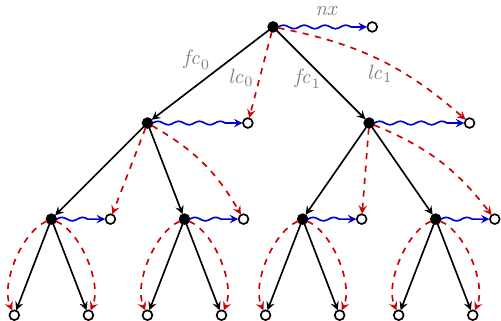}
    \caption{The qif table for the quantum multiplexor program in \Cref{fig:q-multiplexor},
    when $n=3$.}
    \label{fig:qmux-qif-table}
\end{figure}

\subsection{Details of Qif Table Generation}
\label{sub:details_qif_table_generation}

In \Cref{sub:generation_of_qif_table}, 
we have presented \Cref{alg:partial-evaluation} of the partial evaluation of quantum control flow and the generation of qif table. In the following, let us give a more detailed explanation of \Cref{alg:partial-evaluation}.

Recall that \Cref{alg:partial-evaluation} is run by multiple parallel processes.
We start with one initial process evaluating the function \textsc{QEva}.
At the beginning, the program counter $\mathit{pc}$ is initialised to the 
starting address of the compiled main program (i.e., the address of instruction \hlt{start}) and the counter $t$ is initialised to $0$.
An new node $v$ is created in the qif table as the initial node.  
In the generation of the qif table, all new nodes created are of type $\circ$,
which can later become nodes of type $\bullet$ if needed. 
The function \textsc{QEva} then repeats Lines~\ref{alstp:qeva-while-start}--\ref{alstp:qeva-while-end},
as long as the running time $t$ does not exceed the practical running time, i.e., $t\leq T_{\textup{prac}}\parens*{\calP}$.
The repetition consists of the following three stages,
which also appear in the execution on quantum register machine (see \Cref{sec:execution}).
\begin{enumerate}
	\item 
		\textbf{(Fetch)}:

		In this stage, register $\mathit{ins}$ is updated with the current instruction $M_{\mathit{pc}}$ whose address is specified by register $\mathit{pc}$.
		The counter $t$ then increments by $1$.
	\item
		\textbf{(Decode \& Execute)}:

		In this stage, we classically emulate the behaviour of the quantum register machine,
        given the classical inputs.
		In particular, we neglect unitary gate instructions \hlt{uni} and \hlt{unib}
        as they involve the quantum inputs.
		\begin{enumerate}
			\item 
				If $\mathit{ins}=\textup{\hlt{qif}}\texttt{(q)}$,
				then we meet a creation of quantum branching, controlled by the quantum coin $q$.
				In this case, two new nodes $v_0$ and $v_1$ are created in the qif table,
				and are then linked with the current node $v$.
				Since $v_0$ and $v_1$ are from the first enclosed nested instantiation of $\mathbf{qif}\ldots\mathbf{fiq}$ of $v$,
                by \Cref{def:qif-table}, 
                we have $v.\mathit{fc}_i=v.\mathit{lc}_i=v_i$ for $i=0,1$.
                The inverse links are created similarly.
                Next, we fork the current process \textsc{QEva} into 
				two sub-processes \textsc{QEva}$_0$ and \textsc{QEva}$_1$,
				which go into the quantum branches $q=0$ and $q=1$ (corresponding to $\ket{0}\rightarrow \ldots$ and $\ket{1}\rightarrow \ldots$ in the $\mathbf{qif}$ statement), respectively.
				For each \textsc{QEva}$_i$, the current node $v$ is updated by the children node $v_i$.
			\item
				If $\mathit{ins}=\textup{\hlt{fiq}}\texttt{(q)}$,
				then we meet a join of quantum branching, controlled by the quantum coin $q$.
				This implies the current process is a sub-process forked from some parent process.
				So, we wait for the pairing sub-process \textsc{QEva}$'$ 
				with the same parent node;
                i.e., $v'.\mathit{cl}=v.\mathit{cl}=\hat{v}$ for some $\hat{v}$,
                where $v'$ denotes the current node of \textsc{QEva}$'$,
                and $\hat{v}$ denotes the parent node of $v$ and $v'$.
				Let $\hat{t}$ be the maximum $\max\braces*{t,t'}$ of the numbers of instructions 
				in executing the two quantum branches (corresponding to $q=0,1$).
				Given $\hat{t}$, we can calculate $v.w=\hat{t}-t$ and $v'.w=\hat{t}-t'$, 
                the numbers of instruction cycles one needs to wait
				at the nodes $v$ and $v'$, respectively.
				Note that one of $v.w$ and $v'.w$ will be $0$.
				These wait counter information will be used to synchronise the two quantum branches
                at runtime (see \Cref{sec:execution}).

				After collecting the information in two quantum branches,
				the current process \textsc{QEva} will be merged with the pairing process \textsc{QEva}$'$
				into one process by updating $t$ with $\hat{t}$ and $v$ with $\hat{v}$.

				Finally, we need to create a new node $u$ in the qif table for the continuing quantum branch.
				We link the nodes $v$ and $u$ via $\mathit{nx}$ and $\mathit{pr}$.
				According to \Cref{def:qif-table}, 
                node $u$ will be updated as the new last children node for the parent node $\hat{u}$
                of the current node $v$.
				Then, we move the current node from $v$ to $u$.
				Note that the default type of $u$ is $\circ$,
				and it can become $\bullet$ if later we meet a continuing non-nested instantiation of $\mathbf{qif}\ldots\mathbf{fiq}$.
			\item
				If $\mathit{ins}=\textup{\hlt{finish}}$,
				then we have finished the execution of the program,
				and can return the counter $t$ as the actual running time.
			\item
				Otherwise, if $\mathit{ins}$ is any instruction other than the unitary gate instructions \hlt{uni} and \hlt{unib},
				we update the involved registers and memory locations in $M$ according to \Cref{fig:table-ins}.
		\end{enumerate}
	\item
		\textbf{(Branch)}:
		
		In this stage, we update the program counter $\mathit{pc}$ according to 
		the branch offset $\mathit{br}$.
		As previously mentioned in \Cref{sub:quantum_registers},
		if $\mathit{br}=0$, then we increment $\mathit{pc}$ by $1$;
		if $\mathit{br}\neq 0$, then we update $\mathit{pc}$ with $\mathit{pc}+\mathit{br}$.
\end{enumerate}

If the program $\calP$ does not terminate within the practical time $T_{\textup{prac}}\parens*{\calP}$,
\Cref{alg:partial-evaluation} returns a timeout error.
Otherwise, it returns the actual running time $t$,
denoted by $T_{\textup{exe}}\parens*{\calP}$,
which will be used in the execution on quantum register machine (see \Cref{sec:execution}).

It is worth stressing again that in the partial evaluation,
only those quantum variables $q$ involved in the quantum control flow 
(specifically, \hlt{qif}\texttt{(q)} and \hlt{fiq}\texttt{(q)}) and all classical variables need to be evaluated;
while all other quantum variables are ignored.
The parallel processes for running \Cref{alg:partial-evaluation} are completely classical.
One can also make use of this fact to dynamically maintain the classical registers and memory $M$
for every evaluation process, in order to save the space complexity,
but we will not discuss the details here for simplicity.

\subsection{Complexity Analysis of Partial Evaluation}
\label{sub:complexity_partial_evaluation}

Now let us analyse the time complexity of the partial evaluation of quantum control flow and generation of qif table described by \Cref{alg:partial-evaluation},
in terms of classical parallel elemetnary operations. 
Our goal is to show that \Cref{alg:partial-evaluation} can be implemented
using $O\parens*{T_{\textup{exe}}\parens*{\calP}}$ classical parallel elementary operations,
which will be explained as follows.

To clarify what is meant by ``elementary'' here,
recall that all registers, the memory and the qif table in \Cref{alg:partial-evaluation}
are simulated on a classical computer and stored in a classical RAM.
We regard any operation that involves a constant number of memory locations in the classical RAM
as an elementary operation.
For example, in \Cref{alg:partial-evaluation}, an operation on classical registers 
(including system registers $\mathit{pc},\mathit{ins},\mathit{br},\mathit{sp},v,t$
and user registers),
an access to the classical memory $M$,
or creation of a node in the qif table will all be regarded as 
a classical elementary operation.

The meaning of ``parallel'' here is the common one in parallel computing.
Multiple elementary operations performed simultaneously are counted as one parallel elementary operation.
The parallelism can appear within a single process;
e.g., when a process forks into two sub-processes,
its data can be copied in a parallel way, for initialising both sub-processes. 
The parallelism can also appear among multiple processes;
e.g., when two sub-processes run simultaneously before being merged,
they can perform operations in parallel.
We will not bother further going down to the rigorous details.

To prove our goal, note that \Cref{alg:partial-evaluation} terminates without timeout error
if and only if $t=T_{\textup{exe}}\parens*{\calP}$.
So, it suffices to verify that every step of a process (among other parallel processes)
executing \textsc{QEva} can be implemented using $O\parens*{1}$ classical elementary operations,
as follows.
\begin{itemize}
	\item 
		Consider those steps in \Cref{alg:partial-evaluation} only involving registers, the memory and the qif table;
		e.g, $t\gets 0$ in Line~\ref{alstp:qeva-tgets0},
		$\mathit{ins}\gets M_{\mathit{pc}}$ in Line~\ref{alstp:qeva-insgets},
		and the creation of nodes $v_0,v_1$ in Line~\ref{alstp:qeva-create-qif}.
		It is easy to verify that each of them only involves a constant number of memory locations in the classical RAM,
		and therefore can be implemented by $O\parens*{1}$ elementary operations.
	\item
		Consider those steps in \Cref{alg:partial-evaluation} involving forking and merging of sub-processes,
		in particular, Line~\ref{alstp:qeva-fork} and Line~\ref{alstp:qeva-merge}.
		For the forking of a parent process into two sub-processes,
		we need to create a copy of all registers and the memory of the parent process.
		This can be done using $O\parens*{1}$ classical parallel elementary operations.
		Similarly, for the merging of two pairing sub-processes,
		taking the data from any of the two sub-processes
		and updating registers according to Line~\ref{alstp:qeva-merge}
        can be done using $O(1)$ classical parallel elementary operations.
\end{itemize}


\section{Details of Execution on Quantum Register Machine}
\label{sec:details_execution}

In this appendix, we present further details of execution on quantum register machine,
in addition to \Cref{sec:execution}.
Recall that the execution is the last step of implementing quantum recursive programs
and the only step that concerns the quantum inputs and requires quantum hardware.
At this point, the original program written in the language $\mathbf{RQC}^{++}$
is compiled into a low-level program $\calP$ composed of instructions in $\mathbf{QINS}$.
We are promised that $\calP$ terminates and has running time $T_{\textup{exe}}\parens*{\calP}$,
and we have generated its corresponding qif table that can be used to solve the synchronisation problem during the execution. Along the way, we have also set up a symbol table
and allocated memory locations for classical and quantum variables.
All these instructions and data are now loaded into the QRAM,
according to the layout described in \Cref{sub:QRAM-layout}.

\subsection{Details of Unitary \texorpdfstring{$U_{\textup{cyc}}$}{U\_{cyc}} and Unitary \texorpdfstring{$U_{\textup{exe}}$}{U\_{exe}}}
\label{sub:details_uni_cyc_and_exe}

In \Cref{sub:cycle_unitary_u_cyc_and_execution_unitary_u_exe}, we have presented in \Cref{alg:exe-qrm} the construction of the unitary $U_{\textup{cyc}}$ applied repeatedly by the quantum register machine for each instruction cycle at runtime. We have also briefly explained unitary $U_{\textup{exe}}$ as a subroutine of $U_{\textup{cyc}}$.
Now let us provide a more detailed explanation of \Cref{alg:exe-qrm}.

The execution on quantum register machine
described by \Cref{alg:exe-qrm} consists of repeated cycles,
each performing the unitary $U_{\textup{cyc}}$.
From the partial evaluation in \Cref{alg:partial-evaluation}, 
we already obtain the running time $T_{\textup{exe}}\parens*{\calP}$ of the compiled program $\calP$,
so we can fix the number of repetitions to be $T_{\textup{exe}}=T_{\textup{exe}}\parens*{\calP}$. 
Before the repetitions of $U_{\textup{cyc}}$,
we also need to initialise system registers $\mathit{pc},\mathit{sp},\mathit{qifv}$
according to \Cref{sub:quantum_registers}.
All other registers are initialised to $\ket{0}$.

For each instruction cycle, in $U_{\textup{cyc}}$, 
we need to decide whether to wait (i.e., skip the current instruction cycle) or execute,
according to the wait counter information stored in the current node 
(specified by register $\mathit{qifv}$) in the qif table.
To reversibly implement this procedure,
we exploit two additional registers $\mathit{qifw}$ and $\mathit{wait}$,
and unitary $U_{\textup{cyc}}$ is designed to consist of the following three stages.

\begin{enumerate}
	\item 
		\textbf{(Set wait flag)}:
		
		In this stage, we check whether the value $w$ in the wait counter register $\mathit{qifw}$ is $>0$, which records the number of cycles to be skipped at the current node. 
		At this point, we are promised that the wait flag register $\mathit{wait}$ is in state $\ket{0}$.
		If the wait counter $w>0$, then we set a wait flag in register $\mathit{wait}$
            by flipping it to state $\ket{1}$,
            which indicates that the machine needs to wait at the current node. 
	\item
		\textbf{(Execute or wait)}:
	
		In this stage, conditioned on the wait flag in $\mathit{wait}$,
		we decide whether to wait (i.e., skip the current cycle),
		or apply the unitary $U_{\textup{exe}}$ (defined in~\Cref{alg:exe-qrm} and explained later) to execute.
		If the wait flag is $0$, then we apply $U_{\textup{exe}}$.
		If the wait flag is $1$, then we decrement the wait counter $w$ in register $\mathit{qifw}$ by $1$.
	\item
		\textbf{(Clear wait flag)}:

		In this stage, we clear the wait flag in register $\mathit{wait}$,
		conditioned on the value $w$ in register $\mathit{qifw}$ 
		and the counter information $v.w$,
		where $v$ is the value of register $\mathit{qifv}$ that specifies the current node in the qif table.
		Specially, if $w<v.w$, which implies that we have set the wait flag in the first stage,
		then we need to clear the wait flag.
		Note that $v.w$ is stored in the qif table, 
		and we need to first apply the unitary $U_{\textup{fet}}$ (see \Cref{def:access_QRAM}) to fetch it into some free register before using it,
		of which details are omitted for simplicity.
		After this stage, we are guaranteed that the wait flag register $\mathit{wait}$ is in state $\ket{0}$.
\end{enumerate}

\begin{figure}
	\centering
	\begin{subfigure}[b]{0.4\textwidth}
		\centering
		\includegraphics[width=\textwidth]{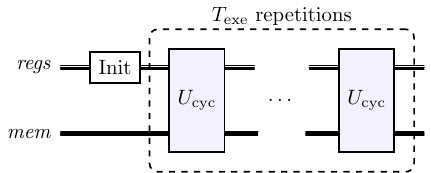}
		\caption{$U_{\textup{main}}$. Here, $\mathit{regs}$ stands for the collection of system and user registers.}
	\end{subfigure}
	~
	\begin{subfigure}[b]{0.5\textwidth}
		\centering
		\includegraphics[width=\textwidth]{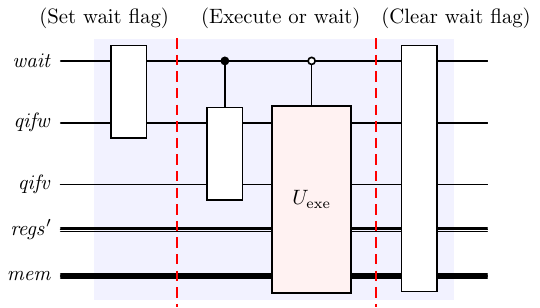}
		\caption{$U_{\textup{cyc}}$. Here, $\mathit{regs}'= \mathit{regs}- \mathit{wait}-\mathit{qifw}-\mathit{qifv}$.}
	\end{subfigure}
	\vskip\baselineskip
	\begin{subfigure}[b]{0.52\textwidth}
		\centering
		\includegraphics[width=\textwidth]{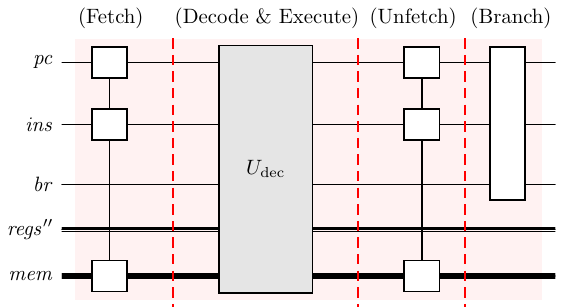}
		\caption{$U_{\textup{exe}}$. Here, $\mathit{regs}''= \mathit{regs}-\mathit{wait}-\mathit{pc}-\mathit{ins}-\mathit{br}$.}
	\end{subfigure}
	~
	\begin{subfigure}[b]{0.36\textwidth}
		\centering
		\includegraphics[width=0.75\textwidth]{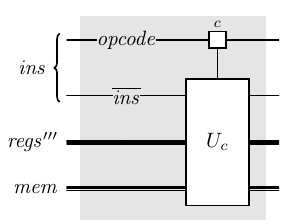}
		\caption{$U_{\textup{dec}}$. Here, $\mathit{regs}'''=\mathit{regs}-\mathit{wait}-\mathit{ins}$;
            $\mathit{opcode}$ is the first several qubits of $\mathit{ins}$
			(see \Cref{sub:the_low_level_language_qins}),
		and $\overline{\mathit{ins}}=\mathit{ins}-\mathit{opcode}$.}
	\end{subfigure}
	\caption{Visualisation of the execution on quantum register machine in the form of quantum circuits.
	Here, unitaries $U_{\textup{main}}$, $U_{\textup{cyc}}$ and $U_{\textup{exe}}$ are defined in \Cref{alg:exe-qrm},
	while $U_{\textup{dec}}$ is defined in \Cref{fig:table-type-eg}. 
	A single wire in these quantum circuits stands for a quantum word, or a part of a quantum word,
	while a bundled wire stands for multiple quantum words.}
	\label{fig:qrm-qcircuit}
\end{figure}

As a subroutine of $U_{\textup{cyc}}$,
the unitary $U_{\textup{exe}}$ consists of the following four stages.
\begin{enumerate}
	\item
		\textbf{(Fetch)}:

		In this stage, we fetch into register $\mathit{ins}$ the current instruction,
            whose address is specified by the program counter $\mathit{pc}$.
		At this point, we are promised that $\mathit{ins}$ is in state $\ket{0}$.
		So, we simply apply the QRAM fetch unitary $U_{\textup{fet}}\parens*{\mathit{ins},\mathit{pc},\mathit{mem}}$ defined in \Cref{def:access_QRAM}.
	\item
		\textbf{(Decode \& Execute)}:

		In this stage, we decode the current instruction in $\mathit{ins}$,
		and execute it by applying the unitary 
		\begin{equation*}
			U_{\textup{dec}}=\sum_{c}\ket{c}\!\bra{c}\otimes \sum_d \ket{d}\!\bra{d}\otimes U_{c,d}
		\end{equation*}
        previously defined in \Cref{fig:table-type-eg}.
		The full list of unitaries $U_{c,d}$ for implementing different instructions are already shown in \Cref{fig:full-list-qins}.
		What remains unspecified in \Cref{fig:full-list-qins} are the unitaries $U_{\textup{qif}}$ and $U_{\textup{fiq}}$ for qif instructions \hlt{qif} and \hlt{fiq},
		which are defined in \Cref{alg:U-qif-fiq} and will be explained in detail later in \Cref{sub:details_uni_qif_and_fiq}.
	\item
		\textbf{(Unfetch)}:

		In this stage, we clear the register $\mathit{ins}$ by applying the unitary 
		$U_{\textup{fet}}\parens*{\mathit{ins},\mathit{pc},\mathit{mem}}$ again.
		After this stage, register $\mathit{ins}$ is guaranteed to be in state $\ket{0}$.
	\item
		\textbf{(Branch)}:

		In this stage, we update the program counter $\mathit{pc}$
		according to the branch offset register $\mathit{br}$.
		Specifically, if the value $y$ in $\mathit{br}$ is non-zero,
		then we add the offset $y$ to the value in $\mathit{pc}$;
		otherwise, we simply increment the value in $\mathit{pc}$ by $1$.
\end{enumerate}

To sum up, in \Cref{fig:qrm-qcircuit},
we visualise the execution on quantum register machine (described by \Cref{alg:exe-qrm}) in the form of quantum circuits, which could be more familiar to the quantum computing community.

\subsection{Details of Unitary \texorpdfstring{$U_{\textup{qif}}$}{U\_{qif}} and \texorpdfstring{$U_{\textup{fiq}}$}{U\_{fiq}}}
\label{sub:details_uni_qif_and_fiq}

In \Cref{sub:unitaries_for_executing_qif_instructions},
we have presented in \Cref{alg:U-qif-fiq} the construction of the unitaries $U_{\textup{qif}}$ and $U_{\textup{fiq}}$ for executing qif instructions, which are used in defining the unitary $U_{\textup{dec}}$
in \Cref{fig:table-type-eg} (see also the full list of unitaries $U_{c,d}$ for implementing different instructions in \Cref{fig:full-list-qins}). Now we provide further explanation of \Cref{alg:U-qif-fiq}.

\begin{itemize}
	\item 
		\textbf{Unitary $U_{\textup{qif}}(q)$}.

        When the current instruction is \hlt{qif}\texttt{(q)},
		we first move the current node specified by $\mathit{qifv}$
		to its first children node corresponding to the quantum coin $q$.
		Specifically, conditioned on the value $x$ of the quantum coin $q$, 
		for the current node $v$ in register $\mathit{qifv}$,
		we first compute into some free register $r$ its first children $v.\mathit{fc}_x$ (corresponding to $x$);
		then use the inverse link $\mathit{cf}$ to clear the garbage data $v$.
		A following swap unitary $U_{\textup{swap}}$ finishes this move of the current node.
		Note that the free register $r$ is cleared at this point.

        Then, we update the wait counter register $\mathit{qifw}$ with the counter information $v.w$ stored in the qif table, where $v$ is the current node in register $\mathit{qifv}$.
		Here, we are promised that $\mathit{qifw}$ is initially in state $\ket{0}$,
		because $U_{\textup{qif}}$ is used as a subroutine of $U_{\textup{exe}}$,
		which will only be performed by $U_{\textup{cyc}}$ when the wait counter $\mathit{qifw}$ is in state $\ket{0}$.

		Note that in the above, the information $\mathit{fc}_x$, $\mathit{cf}$ and $w$ are stored in the qif table,
		and need to be fetched using $U_{\textup{fet}}$ into free registers before being used 
        (see also \Cref{sub:memory_allocation_of_qif_table}),
		of which details are omitted for simplicity.
	\item
		\textbf{Unitary $U_{\textup{fiq}}(q)$}.
  
        When the current instruction is \hlt{fiq}\texttt{(q)},
		we first move the current node specified by $\mathit{qifv}$
		to its parent node.
		Specifically, similar to the construction of $U_{\textup{qif}}$,
		for the current node $v$ in register $\mathit{qifv}$,
		we first compute into some free register $r$ its parent $v.\mathit{cl}$;
		then conditioned on the value $x$ of the quantum coin $q$,
		clear the garbage data $v$ using the link $\mathit{lc}_x$.
		A following swap unitary $U_{\textup{swap}}$ finishes this first move of the current node.

        Then, we continue to move the current node to its next node.
		Specifically, we first compute into some free register $r$ the next node $v.\mathit{nx}$,
        where $v$ is the current node in register $\mathit{qifv}$.
		Next, we clear the garbage data $v$ using the inverse link $\mathit{pr}$,
		followed by a swap unitary $U_{\textup{swap}}$.
		After the above procedure, the register $r$ is also cleared.

		Similar to the case of $U_{\textup{qif}}$, the information $\mathit{cl}$, $\mathit{lc}_x$, $\mathit{nx}$ and $\mathit{pr}$ are stored in the qif table,
		and need to be fetched using $U_{\textup{fet}}$ into free registers before being used (see also \Cref{sub:memory_allocation_of_qif_table}),
		of which details are omitted for simplicity.
\end{itemize}

\subsection{Complexity Analysis of Execution}
\label{sub:complexity_execution}

Now let us analyse the time complexity of the execution on quantum register machine described by \Cref{alg:exe-qrm},
in terms of quantum elementary operations. 
Our goal is to show that the unitary $U_{\textup{main}}$ for executing a compiled program $\calP$ can be implemented
using $O\parens*{T_{\textup{exe}}\parens*{\calP}}$ elementary operations on registers (see \Cref{def:ele-op-reg}) 
and elementary QRAM accesses (see \Cref{def:access_QRAM}).

To prove our goal, first note that the time complexity of $U_{\textup{main}}$
is dominated by the repeated applications of unitary $U_{\textup{cyc}}$.
It suffices to show that unitary $U_{\textup{cyc}}$ can be implemented using
$O\parens*{1}$ elementary operations on registers and QRAM acesses.

Implementing $U_{\textup{cyc}}$ consists of the following steps.
Let us analyse each of them.
\begin{enumerate}
	\item 
		(Set wait flag) in Line~\ref{alstp:exe-set-wait-flag}.

		It is easy to see that the unitary
        \begin{equation*}
        	\sum_{w,z}\ket{w}\!\bra{w}_{\mathit{qifw}}\otimes\ket{z\oplus \bracks*{w>0}}\!\bra{z}_{\mathit{wait}}
        \end{equation*}	
		can be performed using $O\parens*{1}$ elementary operations on registers.
	\item
		(Execute or wait) in Line~\ref{alstp:exe-exe-or-wait}.

		To implement the unitary
        \begin{equation*}
			\ket{0}\!\bra{0}_{\mathit{wait}}\otimes U_{\textup{exe}} 
			+ \sum_{z\neq 0,w}\ket{z}\!\bra{z}_{\mathit{wait}}\otimes\ket{w-1}\!\bra{w}_{\mathit{qifw}}\otimes \Id,
        \end{equation*}
		we first apply $\bullet\parens*{\mathit{wait}}$-$U_{-}\parens*{\mathit{qifw}, \mathit{r}}$,
		where $U_-$ and $\bullet(\cdot)$-$U$ are defined in \Cref{def:ele-op-reg},
        and $\mathit{r}$ is a free register initialised to $\ket{1}$.
		Then, we apply $\circ\parens*{\mathit{wait}}$-$U_{\textup{exe}}$,
		where $\circ(\cdot)$-$U$ is defined in \Cref{def:ele-op-reg},
        and $U_{\textup{exe}}$ will be shown to be implementable using $O\parens*{1}$
		elementary operations on registers and QRAM accesses below.
		Clearing garbage data is simple and also takes $O(1)$ elementary operations.
	\item
		(Clear wait flag) in Line~\ref{alstp:exe-clear-wait-flag}.

		To implement the unitary
		\begin{equation*}
			\sum_{w,v,z}\ket{w}\!\bra{w}_{\mathit{qifw}}\otimes \ket{v}\!\bra{v}_{\mathit{qifv}}
            \ket{z\oplus \bracks*{w<v.w}}\!\bra{z}_{\mathit{wait}},
		\end{equation*}
		suppose that we use the simplest memory allocation of qif table aforementioned in \Cref{sub:memory_allocation_of_qif_table}.
		We first apply $U_{\textup{fet}}\parens*{r,\mathit{qifv},\mathit{mem}}$ (defined in \Cref{def:access_QRAM}) to fetch the information $v.w$ into a free register $r$,
		where $v$ is the value in register $\mathit{qifv}$.
		Given the information $v.w$ in $r$, it is easy to see that using $O\parens*{1}$ elementary operations on registers we can compute $z\oplus [w<v.w]$ in register $\mathit{wait}$.
		Clearing garbage data is simple and also takes $O(1)$ elementary operations.
\end{enumerate}

Next, we show that the unitary $U_{\textup{exe}}$ can be implemented using 
$O\parens*{1}$ elementary operations on registers and QRAM accesses.
Implementing $U_{\textup{exe}}$ consists of the following steps.
\begin{enumerate}
	\item 
		(Fetch) in Line~\ref{alstp:exe-fetch}.

		This is a simple application of $U_{\textup{fet}}$.
	\item
		(Decode \& Execute) in Line~\ref{alstp:exe-dec-exe}.

		This is an application of $U_{\textup{dec}}$,
		which will be shown to be implementable using $O\parens*{1}$
		elementary operations on registers and QRAM accesses below.
	\item
		(Unfetch) in Line~\ref{alstp:exe-unfetch}.
		
		This is a simple application of $U_{\textup{fet}}$ again.
	\item
		(Branch) in Line~\ref{alstp:exe-branch}.

		This step can be done by first applying the unitary $U_{+}$ defined in \Cref{def:ele-op-reg},
		followed by $\circ\parens*{\mathit{br}}$-$U_{+}\parens*{\mathit{pc},r}$,
		where $r$ is a free register initialised to $\ket{1}$.
		Clearing garbage data is simple and also takes $O(1)$ elementary operations.
\end{enumerate}

It remains to show that $U_{\textup{dec}}$ can be implemented using 
$O\parens*{1}$ elementary operations on registers and QRAM accesses.
Recall that $U_{\textup{dec}}$
is a quantum multiplexor $\sum_{c}\ket{c}\!\bra{c}\otimes \sum_d\ket{d}\!\bra{d}\otimes U_{c,d}$ (see \Cref{fig:table-type-eg}),
where the full list of unitaries $U_{c,d}$ for all instructions was presented in \Cref{fig:full-list-qins}.
It it easy to verify that every $U_{c,d}$ can be implemented by using
$O\parens*{1}$ elementary operations on registers and QRAM accesses.
The unitary $U_{\textup{dec}}$ can then be implemented by a sequential composition of
controlled-$U_{c,d}$:
\begin{equation}
	\ket{c}\!\bra{c}\otimes\ket{d}\!\bra{d}\otimes U_{c,d} +\sum_{c'\neq c, d'\neq d} \ket{c'}\!\bra{c'}\otimes \ket{d'}\!\bra{d'}\otimes \Id
	\label{eq:c-U_c}
\end{equation}
for all $c,d$,
each using $O\parens*{1}$ elementary operations on registers and QRAM accesses.
Here, $c$ ranges over possible values of the section $\mathit{opcode}$ in $\mathit{ins}$ (i.e., the names of all instructions), 
and $d$ ranges over possible values of other sections in $\mathit{ins}$. 
Since there are only $22$ distinct instructions,
$O(1)$ distinct (user and system) registers, and $O(1)$ distinct parameters 
(in the section $\mathit{para}$ of $\mathit{ins}$; see \Cref{fig:table-type-eg})
due to $\abs*{\calG}=O(1)$ (see \Cref{sub:quantum_unitary_gate}) and $\abs*{\calO\calP}=O(1)$ (see \Cref{def:ele-op-reg}),
the number of unitaries~\Cref{eq:c-U_c} in the sequential composition is also $O\parens*{1}$.
Hence, we conclude that unitary $U_{\textup{dec}}$ can be implemented using $O(1)$ elementary operations. 

\section{Details of Efficiency and Automatic Parallelisation}
\label{sec:details_computational_efficiency}

In this appendix, we provide further details of the efficiency of implementing quantum recursive programs on quantum register machine
as well as the automatic parallelisation, in addition to \Cref{sec:computational_efficiency_and_algorithmic_speed_up}.
We also present the full proof of \Cref{thm:parallel-qmux} that shows the exponential parallel speed-up (over the straightforward implementation) for implementing the quantum multiplexor, obtained from the automatic parallelisation.

\subsection{Quantum Circuit Complexity for Elementary Operations}
\label{sub:quantum_circuit_complexity_for_elementary_operations}

In \Cref{sub:complexity_partial_evaluation},
we have shown that the partial evaluation for a compiled program $\calP$ can be done using $O\parens*{T_{\textup{exe}}(\calP)}$ classical parallel elementary operations.
In \Cref{sub:complexity_execution},
we have shown that the execution of $\calP$ on quantum register machine can be done
using $O\parens*{T_{\textup{exe}}(\calP)}$ quantum elementary operations,
including on registers (see \Cref{def:ele-op-reg}) and QRAM accesses (see \Cref{def:access_QRAM}). 
These costs are in terms of (parallel) elementary operations.
An immediate result now is that the overall parallel time complexity for implementing $\calP$
is 
\begin{equation}
    \label{eq:overall-time-complexity}
    O\parens*{T_{\textup{exe}}\parens*{\calP}\cdot \parens*{T_{\textup{reg}}+T_{\textup{QRAM}}}},
\end{equation}
where $T_{\textup{reg}}$ and $T_{\textup{QRAM}}$
are the parallel time complexities for implementing an elementary operations on register and QRAM access, respectively.
Here, we use the assumption that classical elementary operations are cheaper than their quantum counterparts. 

The parallel time complexities $T_{\textup{reg}}$ and $T_{\textup{QRAM}}$
are further determined by how quantum elementary operations on registers and QRAM accesses are implemented at the lower level. 
Let us consider the more common quantum circuit model,
and determine $T_{\textup{reg}}$ and $T_{\textup{QRAM}}$
when the quantum register machine is further  
implemented by quantum circuits. 

In particular, we present the following two lemmas.
The first lemma shows the quantum circuit complexity for implementing elementary operations on registers,
which implies $T_{\textup{reg}}=O(\log^2 L_{\textup{word}})$
if we consider the parallel time complexity.

\begin{lemma}[Quantum circuit complexity for elementary operations on registers]
	\label{lmm:qcirc-comp-ele-op}
	Suppose that $\calO\calP$ is any subset of the following set of operators:
	addition, subtraction, multiplication, division, cosine,
	sine, arctangent, exponentiation, logarithm, maximum, minimum, factorial.\footnote{
		Here, since we assume every number is stored in a word,
		each composed of $L_{\textup{word}}$ bits,
		these operators are actually approximated with error $2^{-L_{\textup{word}}}$.
		See \Cref{lmm:parallel-q-circ-ele-ari} for more details.
	}
	Then, every elementary operation on registers defined in \Cref{def:ele-op-reg}
	can be implemented by a quantum circuit of depth $O\parens*{\log^2 L_{\textup{word}}}$
	and size $O\parens*{L_{\textup{word}}^4}$.
\end{lemma}

\begin{proof}
	Let us verify every elementary operation on registers in \Cref{def:ele-op-reg},
	one by one.
	For simplicity of presentation, we change the order of items in \Cref{def:ele-op-reg}.
	\begin{itemize}
		\item
			(Reversible versions of possibly irreversible arithmetic):
			According to \Cref{lmm:parallel-q-circ-ele-ari} to be presented in            \Cref{sec:parallel_quantum_circuits_for_elementary_arithmetics},
			elementary arithmetic in $\calO\calP$ can be 
			implemented by quantum circuits with the desired properties in \Cref{lmm:qcirc-comp-ele-op}.
		\item
			(Swap):
			The implementation of $U_{\textup{swap}}$ is trivially a single layer of $O\parens*{L_{\textup{word}}}$ parallel swap gates.
		\item
			(Unitary gate):
			The unitary gates in the fixed set $\calG$ only act on one or two qubits,
			and can be trivially implemented.
		\item 
			(Reversible elementary arithmetic):
			\begin{itemize}
				\item 
					The implementation of $U_{\oplus}$ is simply by a single layer of $O\parens*{L_{\textup{word}}}$ parallel CNOT gates.
				\item
					The unitary $U_{+}\parens*{r_1,r_2}$ can be implemented by 
					first performing the mapping
					\begin{equation}
						\label{eq:U_plus}
						\ket{x}_{r_1}\ket{y}_{r_2}\ket{z}_{a}\mapsto \ket{x}_{r_1} \ket{y}_{r_2}\ket{z\oplus\parens*{x+y}}_{a},
					\end{equation}
					with $a$ being an ancilla register,
					followed by $U_{\textup{swap}}\parens*{r_1,a}$,
					and finally performing the mapping
					\begin{equation}
						\label{eq:U_minus}
						\ket{x}_{r_1}\ket{y}_{r_2}\ket{z}_{a}\mapsto \ket{x}_{r_1}\ket{y}_{r_2}\ket{z\oplus \parens*{x-y}}_a.
					\end{equation}

					Note that \Cref{eq:U_plus,eq:U_minus} are elementary operations
					of type (Reversibly versions of possibly irreversible arithmetic) discussed above,
					and therefore can be implemented by a quantum circuit of depth $O\parens*{\log^2 L_{\textup{word}}}$
					and size $O\parens*{L_{\textup{word}}^4}$.
					We also have already shown that $U_{\textup{swap}}$ can be implemented by a quantum circuit
					of depth $O\parens*{1}$ and size $O\parens*{L_{\textup{word}}}$.
					Hence, $U_{+}$ can be implemented by a quantum circuit with the desired properties in \Cref{lmm:qcirc-comp-ele-op}.
				\item
					The implementation of $U_{-}$ is similar to that of $U_{+}$.
				\item
					For $U_{\textup{neg}}$, if we encode an integer by recording its sign in the first bit,
					then implementing $U_{\textup{neg}}\parens*{r}$ is trivially applying an $X$ gate on the first qubit of $r$.
			\end{itemize}
		\item
			(Controlled versions):
			Recall that we only have a constant number of registers.
			Suppose that an elementary operation $U$ can be implemented by a quantum circuit $Q$
			of depth $O\parens*{\log^2 L_{\textup{word}}}$ and size $O\parens*{L_{\textup{word}}^4}$.
			Then, the controlled version $\circ\parens*{r}$-$U$
			can be implemented by first
			applying $\circ\parens*{r}$-$X\parens*{a}$,
			with $a$ being an ancilla qubit,
			followed by the single-controlled $c\parens*{a}$-$U$,
			and finally $\circ\parens*{r}$-$X\parens*{a}$ again.
			Here, $\circ\parens*{r}$-$X\parens*{a}$ is actually a multi-controlled Pauli $X$ gate,
			which is known to be implementable by a quantum circuit of depth $O\parens*{\log L_{\textup{word}}}$
			and size $O\parens*{L_{\textup{word}}}$ (similar to, e.g., Corollary 2.5 in \cite{ZWY24}).
			The single-controlled $c\parens*{a}$-$U$
			can be implemented by replacing every gate in $Q$ by its single-controlled version,
			which gives a quantum circuit of depth $O\parens*{\log^2 L_{\textup{word}}}$ and size $O\parens*{L_{\textup{word}}^4}$.
			The implementation of $\bullet\parens*{r}$-$U$ is similar.

			Since we only have $O\parens*{1}$ registers,
			it is easy to see by induction that the controlled versions of any elementary operation
			can be implemented by a quantum circuit with the desired properties in \Cref{lmm:qcirc-comp-ele-op}.
	\end{itemize}
\end{proof}

The second lemma shows the quantum circuit complexity for implementing elementary QRAM accesses, which implies $T_{\textup{QRAM}}=O\parens*{\log N_{\textup{QRAM}}+\log L_{\textup{word}}}$ if we consider the parallel time complexity. 
Using quantum circuits to implement QRAM is actually a well studied topic, 
for which the readers are referred to \cite{JR23} for a detailed review.

\begin{lemma}[Quantum circuit complexity for elementary QRAM accesses]
	\label{lmm:qcirc-comp-ele-qram}
	The elementary QRAM accesses defined in \Cref{def:access_QRAM}
	can be implemented by quantum circuits of depth $O\parens*{\log N_{\textup{QRAM}}+\log L_{\textup{word}}}$
	and size $O\parens*{L_{\textup{word}}\cdot N_{\textup{QRAM}}\log N_{\textup{QRAM}}}$.
\end{lemma}

\begin{proof}
	We first show how to implement $U_{\textup{ld}}$.
	Suppose that $N=N_{\textup{QRAM}}$ and $L=L_{\textup{word}}$.
	Existing circuit QRAM constructions (e.g., \cite{GLM08,HLGJ21}) show that
	a QRAM of $N$ qubits can be implemented by a quantum circuit
	of depth $O\parens*{\log N}$
	and size $O\parens*{N\log N}$.
	To adapt this to a QRAM of $N$ quantum words,
	we can use $L$ circuit QRAMs in parallel.

	Specifically, the implementation of $U_{\textup{ld}}\parens*{r,a,\mathit{mem}}$ is as follows.
	\begin{enumerate}
		\item 
			We first perform the unitary mapping
			\begin{equation}
				\label{eq:copy-L}
				\ket{x}_a\ket{y_1}_{a_1}\ldots\ket{y_{L-1}}_{a_{L-1}}
				\mapsto \ket{x}_a\ket{y_1\oplus x}_{a_1}\ldots \ket{y_{L-1}\oplus x}_{a_{L-1}},
			\end{equation}
			where each of $a_1,\ldots,a_{L-1}$ is composed of $L$ ancilla qubits
			and initialised to $\ket{0}$.
			Here, \Cref{eq:copy-L} prepares $L$ copies of the address in $a$ (in the computational basis),
			and can be implemented by a quantum circuit of depth $O\parens*{\log L}$
			and size $O\parens*{L^2}$ (see e.g., \cite{CW00}).
		\item
			Suppose that $\mathit{mem}_i$ is composed of the $i^{\textup{th}}$ qubit
			of every quantum word in our QRAM $\mathit{mem}$ (composed of quantum words).
			Let $r_i$ be the $i^{\textup{th}}$ qubit of the target register $r$,
			and let $a=a_0$.
			Then, we apply $U_{\textup{ld}}\parens*{r_i,a_i,\mathit{mem}_i}$ for $i=0$ to $L-1$ in parallel,
			each of which is implemented by a circuit QRAM composed of qubits, from previous works (e.g., \cite{GLM08,HLGJ21}).
			In total, this step can be implemented by a quantum circuit of depth $O\parens*{\log N}$
			and size $O\parens*{LN\log N}$.
		\item
			Finally, we perform the unitary mapping \Cref{eq:copy-L} again to clear the garbage data.
	\end{enumerate}
	It is easy to verify that the above implementation satisfies the properties in \Cref{lmm:qcirc-comp-ele-qram}.
\end{proof}

It turns out that $T_{\textup{reg}}$ and $T_{\textup{QRAM}}$
given by \Cref{lmm:qcirc-comp-ele-op,lmm:qcirc-comp-ele-qram}
are actually small enough compared to the overall parallel time complexity for implementing quantum recursive programs.  
In the next section, we will make this more concrete by analysing the example of the quantum multiplexor program.

\subsection{Proof of \texorpdfstring{\Cref{thm:parallel-qmux}}{Theorem 1}}
\label{sub:proof-main-theorem}

Now let us analyse the complexity for implementing the quantum multiplexor program on quantum register machine,
and prove \Cref{thm:parallel-qmux} (whose proof sketch is already presented in \Cref{sec:computational_efficiency_and_algorithmic_speed_up}). 
Recall that \Cref{thm:parallel-qmux} shows that 
we can obtain an exponential parallel speed-up (over the straightforward implementation) for implementing the quantum multiplexor. 

\begin{proof}[Proof of \Cref{thm:parallel-qmux}]
    To determine the parallel time complexity, 
    i.e., quantum circuit depth for implementing the quantum multiplexor program (see \Cref{fig:q-multiplexor})
    on quantum register machine,
    it suffices to determine the terms $T_{\textup{exe}}(\calP)$,
    $T_{\textup{reg}}$ and $T_{\textup{QRAM}}$ in \Cref{eq:overall-time-complexity},
    where $\calP$ denotes the compiled program in \Cref{fig:qmuxlow-full}.
    \begin{itemize}
        \item 
            To determine $T_{\textup{exe}}(\calP)$,
            let us analyse the partial evaluation (described by \Cref{alg:partial-evaluation})
            of $\calP$. 
            Note that Lines 1--424 of $\calP$ are independent of $n$,
            in the sense that no matter how large $n$ is, this part of the program text is fixed.
            
            Before branching to each procedure $Q[x]$ (that describes unitary $U_x$ in the quantum multiplexor \Cref{eq:qmux-uni}) for $x\in [N]$, 
            the value of the counter $t$ is only determined by the number of nesting layers of $\mathbf{qif}\ldots\mathbf{fiq}$ instantiations, up to a constant factor.
            As the qif table of the quantum multiplexor program is a tree, as shown in \Cref{fig:qmux-qif-table},
            the number of nesting layers is exactly the depth of the tree, which is equal to $n$.

            After branching to procedures $Q[x]$ for $x\in [N]$,
            the value of the counter $t$ is determined by the time complexity of executing every $Q[x]$. 
            Recall that \Cref{thm:parallel-qmux} assumes each $U_x$ is composed of $T_x$ elementary unitary gates, i.e.,
            the procedure body $C_x$ of $Q[x]$ is composed of $T_x$ unitary gate instructions.
            So, the number of cycles to execute every $C_x$ is $O(T_x)$. As \Cref{alg:partial-evaluation} is run by parallel processes,
            each going into a single quantum branch, 
            the overall parallel running time is $O\parens*{\max_{x\in [N]} T_x}$.
            
            Combining the above together, 
            we have 
            \begin{equation*}
                T_{\textup{exe}}(\calP)=O\parens*{\max_{x\in [N]}T_x+n}.
            \end{equation*}
        \item 
            To determine $T_{\textup{reg}}$ and $T_{\textup{QRAM}}$,
            it suffices to determine the word length $L_{\textup{word}}$ and the QRAM size $N_{\textup{QRAM}}$ for executing the compiled program $\calP$. Let us first determine how large $N_{\textup{QRAM}}$ is needed, by calculating the sizes (counted by the number of quantum words) of all sections in the QRAM, as follows. 
            \begin{itemize}
                \item 
                    Program section:
                    The compiled program $\calP$ is already presented in \Cref{fig:qmuxlow-full},
                    which consists of $\Theta\parens*{\sum_{x\in [N]} T_x}$ instructions. 
                \item 
                    Symbol table section:
                    The size of the symbol table is upper bounded by the number of variables.
                    By our assumption, every $C_x$ is composed of unitary gate instructions, 
                    so the number of variables is further upper bounded by $\Theta\parens*{\sum_{x\in [N]} T_x}$, which is the number of instructions in $\calP$. 
                \item 
                    Variable section:
                    The size of the variable section is upper bounded by
                    $\Theta\parens*{\sum_{x\in [N]} T_x}$, as mentioned above. 
                \item 
                    Qif table section:
                    The size of the qif table is easily seen to be $\Theta\parens*{2^{n}}$, 
                    by noting that the qif table of $\calP$ is a simple tree shown in \Cref{fig:qmux-qif-table}.
                \item 
                    Stack section:
                    The size of the stack is upper bounded by the number of variables and the number of nesting layers of $\mathbf{qif}\ldots\mathbf{fiq}$ instantiations, up to a constant factor. 
                    The former is upper bounded by $\Theta\parens*{\sum_{x\in [N]} T_x}$, and the latter is upper bounded by $\Theta(n)$.
                    Hence, the overall size is upper bounded by $\Theta\parens*{\sum_{x\in [N]} T_x}$.
            \end{itemize}
            To summarise, taking $N_{\textup{QRAM}}=\Theta\parens*{\sum_{x\in [N]}T_x}$
            is sufficient to implement $\calP$.
            Now for the word length $L_{\textup{word}}$,
            we only need it to be large enough to store an address of a memory location in the QRAM, as required by the formats of instructions in $\mathbf{QINS}$ (see \Cref{fig:table-type-eg}).
            So, taking $L_{\textup{word}}=O\parens*{\log\parens*{N_{\textup{QRAM}}}}$ is enough.
            Finally, by \Cref{lmm:qcirc-comp-ele-op,lmm:qcirc-comp-ele-qram}, we can calculate 
            \begin{align*}
    		T_{\textup{reg}}&=O\parens*{\log^2\log N_{\textup{QRAM}}}= O\parens*{\log^2 T_{\textup{QRAM}}},\\
    		T_{\textup{QRAM}}&=O\parens*{\log N_{\textup{QRAM}}}= O\parens*{n+\log\parens*{\max_{x\in [N]}T_x}}.
	    \end{align*}
    \end{itemize}
    By inserting the above $T_{\textup{exe}}(\calP)$, $T_{\textup{reg}}$
    and $T_{\textup{QRAM}}$ into \Cref{eq:overall-time-complexity},
    we obtain the overall parallel time complexity 
    $\widetilde{O}\parens*{n\cdot \max_{x\in [N]} T_x+n^2}$, as desired.
\end{proof}

\section{Parallel Quantum Circuits for Elementary Arithmetic}
\label{sec:parallel_quantum_circuits_for_elementary_arithmetics}

In this appendix,
we restate Lemma 2.6 in~\cite{ZWY24},
which translates several classical results in parallel computing~\cite{Ofman63,Reif83,BCH86} into the quantum setting:
elementary arithmetic functions can be efficiently approximated by classical circuits of low depth.
The translation is by introducing garbage data as in standard reversible computing~\cite{Landauer61,Bennett73}.
Although these classical results are more refined,
the following lemma only provides simple upper bounds for simplicity.

\begin{lemma}[Parallel quantum circuits for elementary arithmetics, restating Lemma 2.6 in~\cite{ZWY24}]
	\label{lmm:parallel-q-circ-ele-ari}
	Suppose $f$ is one of the following elementary arithmetic functions:
	addition, subtraction, multiplication, division, modulo, cosine,
	sine, arctangent, exponentiation, logarithm, maximum, minimum, factorial.
	Then, the unitary 
	\begin{equation*}
		\sum_{\widetilde{x},\widetilde{y},z=0}^{2^{L}-1} 
		\ket{\widetilde{x},\widetilde{y}}\!\bra{\widetilde{x},\widetilde{y}}\otimes 
		\ket{z\oplus \widetilde{f}\parens*{\widetilde{x},\widetilde{y}}}\!\bra{z}
	\end{equation*}
	can be implemented by a quantum circuit of depth $O\parens*{\log^2 L}$
	and size $O\parens*{L^4}$,
	where $\widetilde{x},\widetilde{y}$ are certain proper representations of $x,y$,
	and $\widetilde{f}\parens*{\widetilde{x},\widetilde{y}}$ is an approximation of $f\parens*{x,y}$
	with error $2^{-L}$.
	Here, for unary $f$, the input $y$ is omitted.
	Further, if $f$ is addition, subtraction or multiplication (modulo $2^L$),
	the depth can be $O\parens*{\log L}$
	and the error can be $0$.
\end{lemma}

It is worth mentioning the quantum circuit in \Cref{lmm:parallel-q-circ-ele-ari}
can use ancilla qubits.
The number of ancilla qubits can be trivially upper bounded by the size $O\parens*{L^4}$ of the quantum circuit.

\section{More Examples of Quantum Recursive Programs}
\label{sec:more_examples}

In this appendix, we provide more examples of quantum recursive programs for further illustrating the compilation process.
For simplicity of presentation, 
we will only show the high-level transformations and the high-to-mid-level translation of these examples,
while omitting the rather lengthy mid-to-low-level translation.
The high-level programs for these examples are from~\cite{YZ24}.

The first example is the generation of the Greenberger-Horne-Zeilinger (GHZ) state~\cite{GHZ89}.
This simple example involves recursive procedure calls but no $\mathbf{qif}$ statements. 

\begin{example}[Generation of GHZ states]
    The (generalised) Greenberger-Horne-Zeilinger (GHZ) state~\cite{GHZ89} of $n$ qubits is defined as:
    \begin{equation*}
        \ket{\textup{GHZ}(n)}=\frac{1}{\sqrt{2}}\parens*{\ket{0}^{\otimes n}+\ket{1}^{\otimes n}}.
    \end{equation*}
    We can use the following program in $\mathbf{RQC}^{++}$ to generate the GHZ state~\cite{YZ24}:
    \begin{equation}
        \label{eq:GHZ}
        \begin{split}
            P_{\textup{main}}(n)\ \Leftarrow\ & \mathbf{if}\ n=1\\
            & \quad \mathbf{then}\ H\bracks*{q[n]}\\
            &\quad \mathbf{else}\ P_{\textup{main}}(n-1); \mathit{CNOT}\bracks*{q[n-1],q[n]}\\
            & \mathbf{fi}.
        \end{split}
    \end{equation}

    The GHZ state generation program after the high-level transformations is presented in \Cref{fig:GHZ-after-high}, of which the high-to-mid-level translation is further shown in \Cref{fig:GHZ-full-mid}. As usual, we have done some manual optimisation to improve the presentation.

    \begin{figure}
        \centering
        \begin{equation*}
            \begin{split}
                P_{\textup{main}}(n)\ \Leftarrow \ & n_0:=n-1;\\
                & \mathbf{if}\ n_0\ \mathbf{then}\\
                & \qquad P_{\textup{main}}(n_0);\\
                & \qquad \mathit{CNOT}\bracks*{q[n_0],q[n]}\\
                & \quad\mathbf{else}\ H[q[n]]\\
                &\mathbf{fi}
            \end{split}
        \end{equation*}
        \caption{The GHZ state generation program after the high-level transformations.}
        \label{fig:GHZ-after-high}
    \end{figure}

    \begin{figure}
        \centering
        \includegraphics[width=0.95\linewidth]{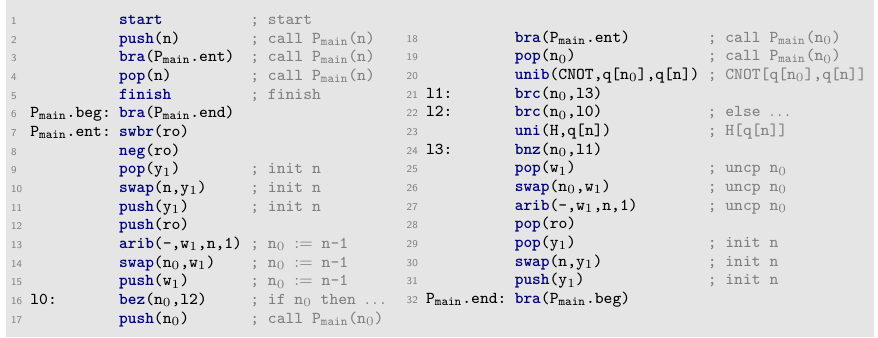}
        \caption{Full high-to-mid-level translation of the GHZ state generation program.}
        \label{fig:GHZ-full-mid}
    \end{figure}
\end{example}

The second example is the multi-controlled gate. This example is still simple,
but illustrates both $\mathbf{qif}$ statements and recursive procedure calls. 

\begin{example}[Multi-controlled Gate]
    Let $U$ be an elementary unitary gate.
    Then the multi-controlled $U$ gate with $n-1$ control qubits is defined by:
    \begin{equation*}
        C^{(*)}(U)\ket{i_1,\ldots,i_{n-1}}\ket{\psi}
        =\begin{cases}
            \ket{i_1,\ldots,i_{n-1}} U \ket{\psi}, & i_m=\ldots= i_{n-1}=1,\\
            \ket{i_1,\ldots,i_{n-1}} \ket{\psi}, & o.w.
        \end{cases}
    \end{equation*}
    for all $i_1,\ldots,i_{n-1}\in \braces*{0,1}$ and quantum state $\ket{\psi}$.
    We can describe $C^{(*)}(U)$ by the following program in $\mathbf{RQC}^{++}$~\cite{YZ24}:
    \begin{equation}
        \label{eq:controlled-gate}
		\begin{split}
            P_{\textup{main}}(n)\ \Leftarrow\ & P(1,n)\\
            P(m,n)\ \Leftarrow\ &\mathbf{if}\ m=n\\
                                &\ \ \ \mathbf{then}\ U[q[n]]\\
                                &\ \ \ \mathbf{else}\ \mathbf{qif}[q[m]]\ket{0}\rightarrow\mathbf{skip}\\ 
                                &\quad\ \qquad\qquad\square\ \ \ \ket{1}\rightarrow P\parens*{m+1,n}\\ 
                                &\qquad\quad\ \mathbf{fiq}\\ 
                                &\mathbf{fi}.
		\end{split}
	\end{equation}

    In \Cref{fig:multi-controlled-gate-after-high},
    we present the multi-controlled gate program after applying the high-level transformations to \Cref{eq:controlled-gate}.
    Further, the full high-to-mid-level translation of the program is shown in \Cref{fig:multi-controlled-gate-full-mid}, which is annotated to show the correspondence with the program in \Cref{fig:multi-controlled-gate-after-high}.
    
    \begin{figure}
        \centering
        \begin{equation*}
            \begin{split}
                P_{\textup{main}}(n)\ \Leftarrow\ & y:=1; P(y,n)\\
                P(m,n)\ \Leftarrow\ & x:=n-m;\\
                                                &\mathbf{if}\ x\ \mathbf{then}\\ 
                                                &\qquad m_0:=m+1;\\
                                                &\qquad\mathbf{qif}[q[m]]\ket{0}\rightarrow\mathbf{skip}\\ 
                                		      &\qquad\qquad\square\ \ \ \ket{1}\rightarrow P\parens*{m_0,n}\\   
                                                &\qquad\mathbf{fiq}\\ 
                                                &\quad\mathbf{else}\ U[q[n]]\\
                                		      &\mathbf{fi}
            \end{split}
        \end{equation*}
        \caption{The multi-controlled gate program after the high-level transformations.}
        \label{fig:multi-controlled-gate-after-high}
    \end{figure}

    \begin{figure}
        \centering
        \includegraphics[width=0.9\linewidth]{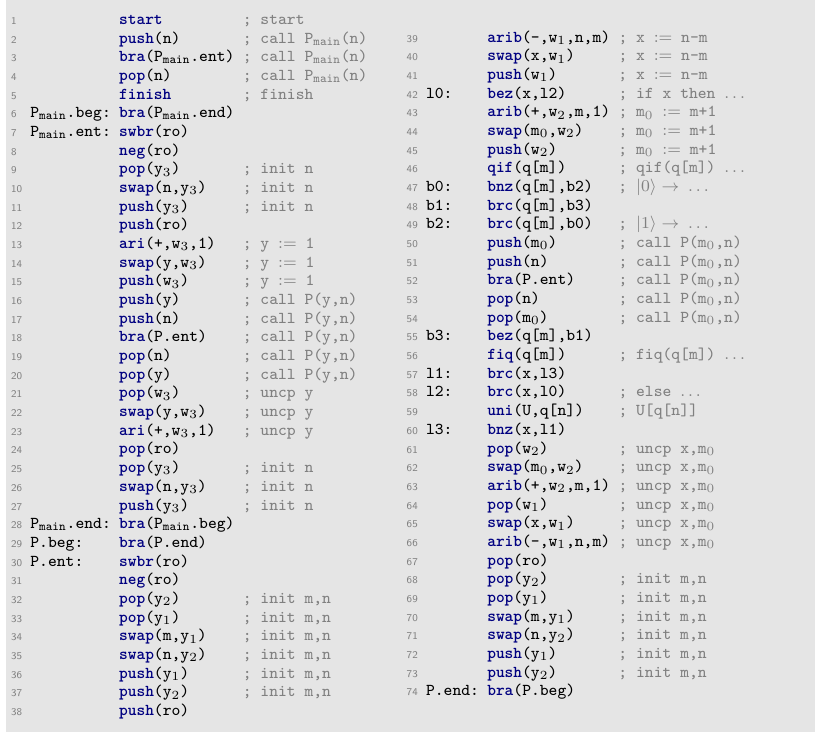}
        \caption{Full high-to-mid-level translation of the multi-controlled gate program.}
        \label{fig:multi-controlled-gate-full-mid}
    \end{figure}
\end{example}

The third example is the quantum state preparation, which is a common subroutine in many quantum algorithms, e.g., Hamiltonian simulation~\cite{BCCKS15,BCK15,Kothari14}, quantum machine learning~\cite{LMR14,KP17}, and solving quantum linear system of equations~\cite{HHL09,CKS17}. 
This example involves both $\mathbf{qif}$ statements and recursive procedure calls.
It also offers automatic parallelisation,
resulting in a quantum circuit of poly-logarithmic depth for implementing the quantum state preparation (of which the detailed analysis is omitted),
although whose complexity is slightly worse than the state-of-the-art algorithms~\cite{Rosenthal23,STYYZ23,ZLY22,YZ23}.

\begin{example}[Quantum State Preparation]
    The task of quantum state preparation is to generate the $n$-qubit quantum state
    \begin{equation}
        \ket{\psi} = \sum_{j\in [N]} \alpha_j \ket{j},
    \end{equation}
    where $N=2^n$ and $\alpha_j\in \Co$ satisfy $\sum_{j\in [N]}\abs*{\alpha_j}^2=1$.
    
    For $j\in [N]$, let us define $\theta_j\in [0,2\pi)$ such that $\alpha_j=e^{i\theta_j} \abs*{\alpha_j}$. For $l,r\in [N]$, let $S_{l,r}=\sum_{j=l}^r \abs*{\alpha_j}^2$.
    We define a single qubit unitary $U_{k,x}$ such that:
    \begin{equation*}
        U_{k,x}\ket{0}=\sqrt{\gamma_x}\ket{0} + e^{i\beta_x}\sqrt{1-\gamma_x}\ket{1}, 
    \end{equation*}
    where $\gamma_x=\frac{S_{u,w-1}}{S_{u,v-1}}$, $\beta_x=\theta_w-\theta_u$, $u=2^{n-k}x$, $v=2^{n-k}(x+1)$ and $w=(u+v)/2$.
    
    Now we can use the following program in $\mathbf{RQC}^{++}$ to generate the state $\ket{\psi}$~\cite{YZ24}:
    \begin{equation}
        \label{eq:QSP}
        \begin{split}
            P_{\textup{main}}(n)\ \Leftarrow \ & P(0,n,0)\\
            P(k,n,x)\ \Leftarrow\ & \mathbf{if}\ k\neq n\ \mathbf{then} \\
            &\quad Q(k,x);\\
            &\quad \mathbf{qif}\bracks*{q\bracks*{k}} \ket{0}\rightarrow P\parens*{k+1,n,2x}\\
            &\quad\qquad\square\ \ \ket{1}\rightarrow P\parens*{k+1,n,2x+1}\\
            &\quad\mathbf{fiq}\\
            &\mathbf{fi}\\
            Q(k,x)\ \Leftarrow \ & C,
        \end{split}
    \end{equation}
    where $C$ is a quantum circuit that performs $U_{k,x}[q[k+1]]$. 
    (In practice, when the elementary gate set is simple, e.g., $\braces*{H,S,T,\mathit{CNOT}}$, what $C$ performs is only an approximation of $U_{k,x}$, and $C$ depends on the explicit $\alpha_j$ (for $j\in [N]$).) 

    In \Cref{fig:QSP-after-high},
    we show the quantum state preparation program after applying the high-level transformations to \Cref{eq:QSP}.
    The full high-to-mid-level translation of the program is then presented in \Cref{fig:QSP-full-mid},
    where the annotation shows the correspondence with the program in \Cref{fig:QSP-after-high}.

    \begin{figure}
        \centering
        \begin{equation*}
            \begin{split}
                P_{\textup{main}}(n)\ \Leftarrow \ & z:=0; P(z,n,z)\\
                P(k,n,x)\ \Leftarrow\ & y:=n-k;\\
                & \mathbf{if}\ y\ \mathbf{then}\\
                &\quad Q(k,x);\\
                &\quad k_0:=k+1;\\
                &\quad x_0:=2x;\\
                &\quad x_1:=x_0+1;\\
                &\quad\mathbf{qif}\bracks*{q\bracks*{k}} \ket{0}\rightarrow P\parens*{k_0,n,x_0}\\
                &\quad\qquad\square\ \ \ket{1}\rightarrow P\parens*{k_0,n,x_1}\\
                &\quad\mathbf{fiq}\\
                &\mathbf{fi}\\
                Q(k,x)\ \Leftarrow \ & C,
            \end{split}
        \end{equation*}
        \caption{The quantum state preparation program after the high-level transformations.}
        \label{fig:QSP-after-high}
    \end{figure}

     \begin{figure}
        \centering
        \includegraphics[width=0.9\linewidth]{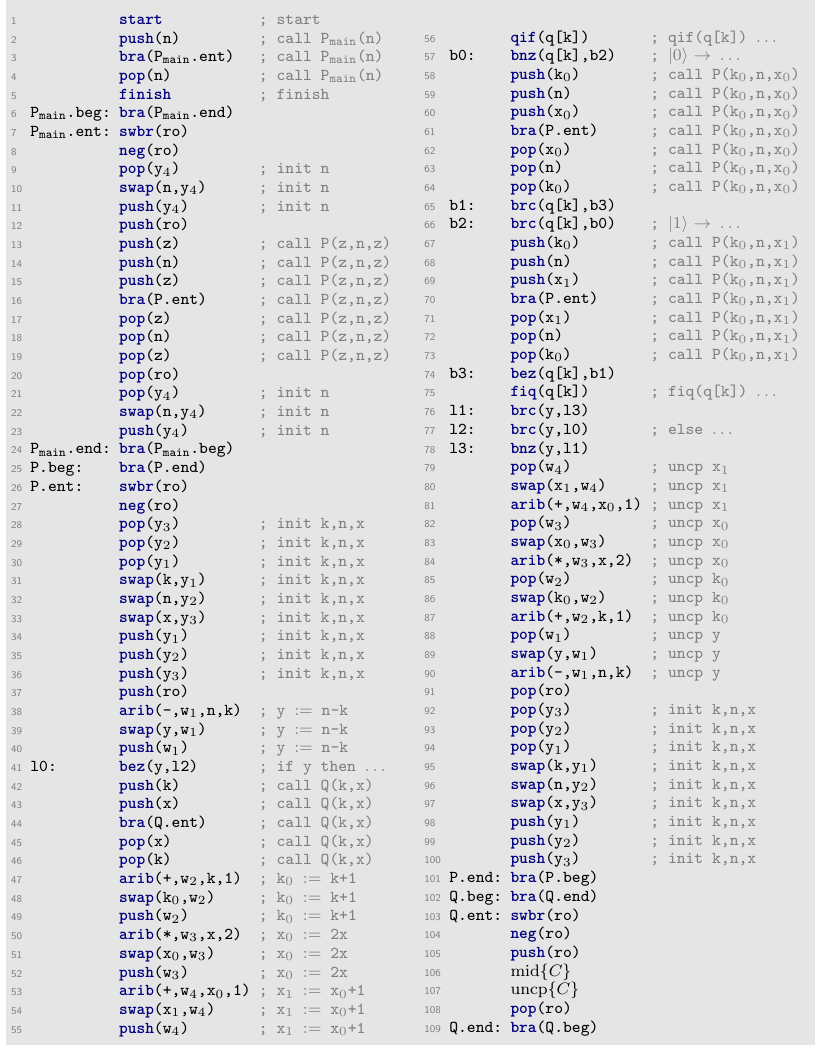}
        \caption{Full high-to-mid-level translation of the quantum state preparation program.}
        \label{fig:QSP-full-mid}
    \end{figure}
\end{example}

\section{Further Discussions on Related Work}
\label{sec:further_related_work}

\subsection{The Synchronisation Problem}
\label{sub:further_synchronisation_problem}

In \Cref{sec:partial-evaluation}, we have described the partial evaluation of quantum control flow, which is used to address (within a practical running time) the synchronisation problem (see \Cref{sub:the_synchronisation_problem}) for executing programs with quantum control flow. In this appendix, we discuss some related works to the synchronisation problem.

The synchronisation problem dates back to Deutsch's definition of quantum Turing machine~\cite{Deutsch85}.
In particular, Myers~\cite{Myers97} pointed out a problem in Deutsch's original definition: for a universal quantum Turing machine that allows inputs being in quantum superposition,
different quantum branches are asynchronous and can have different termination times,
where some branch might even never halt.
The consequence is that one cannot check the halting of quantum Turing machine (according to Deutsch's definition) by directly measuring a flag qubit because the measurement can destroy the superposition of different quantum branches. 

A later definition of quantum Turing machine~\cite{BV93} given by Bernstein and Vazirani circumvented this problem by an explicit halting scheme. It requires all quantum branches to be synchronised:
if any quantum branch terminates the computation at some time $T$, then all other branches also terminate at $T$.
There is also a number of subsequent discussions~\cite{Ozawa98,LP98,Ozawa98b,Shi02,MO05,WY23} about the synchronisation problem and the halting schemes of quantum Turing machine.
Recently, in \cite{YVC24} they studied the synchronisation problem in general transition systems and formalised synchronisation as a condition.

Addressing (or circumventing) the synchronisation problem is often by restricting the inputs to be a specific subset such that different quantum branches always synchronise. 
For example, in~\cite{BV93}, the synchronisation condition is explicitly  stated in their definition of quantum Turing machine.
In~\cite{WY23}, they extend the definition of quantum Turing machine in~\cite{BV93}, and show a conversion of a quantum Turing machine from the extended definition to the standard one by inserting meaningless symbols to synchronise different quantum branches. 
In~\cite{YVC24}, they impose structures on the program in their low-level language
by borrowing the techniques of paired branch instructions from the classical reversible languages~\cite{Vieri99,Frank99,AGY07,TAG12}, 
and manually insert nop (no operation) into the low-level programs to synchronise different quantum branches.

To address the synchronisation problem, in our implementation of quantum recursive programs,
we exploit the structures of the programs imposed by the high-level language $\mathbf{RQC}^{++}$ (in particular, the $\mathbf{qif}$ statement for managing quantum control flow; see \Cref{sec:background}),
and our compilation promises that the compiled low-level programs inherit the structures (via paired instructions \hlt{qif} and \hlt{fiq} in $\mathbf{QINS}$; see \Cref{sub:the_low_level_language_qins}). 
Then, we use the partial evaluation of quantum control flow to generate the qif table (see \Cref{sec:partial-evaluation}), a data structure that records the history information of quantum branching within a practical running time, which is used later (in quantum superposition) to synchronise different quantum branches at runtime (see \Cref{sub:unitaries_for_executing_qif_instructions}).
Note that our whole implementation is automatic. The technique of manually inserting nop (like in previous works~\cite{WY23,YVC24}) is not extendable to handle our case of quantum recursive programs, because the length of (dynamic) computation generated by quantum recursion cannot be pre-determined from the (static) program text.  

\subsection{Implementation of Quantum Multiplexor}
\label{sub:further_qmux}

In \Cref{thm:parallel-qmux}, we have shown that via the quantum register machine, we can automatically obtain a parallel implementation of the quantum multiplexor. In this appendix, we further compare it with previous efficient implementations in \cite{ZLY22,ZY24}. 

Recall that a quantum multiplexor can be described by the unitary
\begin{equation*}
    U=\sum_{x\in [N]} \ket{x}\!\bra{x}\otimes U_x,
\end{equation*}
where $N=2^n$ and $n$ is the number of control qubits.
Let us assume all $U_x$ act on the last $m$ data qubits. 
For general $U_x$, it is shown in \cite[Algorithm 4,5]{ZLY22} and \cite[Lemma 7]{ZY24} that $U$ can be implemented by a quantum circuit of depth $O\parens*{n+\max_{x\in [N]}T_x}$ and size $O\parens*{\sum_{x\in [N]}T_x}$, where each $T_x$ is the time for implementing controlled-$U_x$.
The idea of their construction is similar to that of the bucket-brigade QRAM~\cite{GLM08,GLM08b,HZZCSGJ19,HLGJ21}.

In particular, the bucket-brigade QRAM uses a binary tree structure to route the address qubits (see \Cref{sub:quantum_random_access_memory}) to the corresponding memory location.
The construction in \cite{ZLY22,ZY24} uses a similar structure to route the $n$ control qubits and a single data qubit to the location specified by the control qubits. Such structure is repeated $m$ times in parallel to route the total $m$ data qubits, and the first $n$ control qubits also needs to be copied (in the computational basis) $m$ times in parallel. At each location $x\in [N]$, the quantum circuit for implementing $U_x$ is placed.
Finally, a reverse routing processes is applied to retrieve the control and data qubits.
In this way, if the $n$ control qubits are in superposition, then the $m$ data qubits can travel through a superposition of paths to pass through unitaries $U_x$ and therefore realise the quantum multiplexor $U$.

Unlike in~\cite{ZLY22,ZY24} where data flows in quantum superposition,
in our implementation of the quantum multiplexor, 
the $m$ data qubits do not move.
The control structure in $U$ is completely  captured by the quantum multiplexor program $\calP$ in \Cref{fig:q-multiplexor},
and the quantum register machine essentially realises the quantum control flow: it executes all quantum programs $C_x$ (that describe unitaries $U_x$) in quantum superposition.
%
Indeed, the quantum register machine also exploits the circuit QRAM~\cite{GLM08,GLM08b,HZZCSGJ19,HLGJ21} (see also \Cref{sub:quantum_circuit_complexity_for_elementary_operations}) at the low-level,
to which the high-level programmer is oblivious.

From the perspective of design (as aforementioned in \Cref{sub:motivating_example,sub:bonus}),
compared to the manual design of the rather involved quantum circuits in~\cite{ZLY22,ZY24},
our \Cref{thm:parallel-qmux} is automatically obtained from our implementation of quantum recursive programs. In our framework, the programmer only needs to design at high-level (in particular, the quantum multiplexor program in~\Cref{fig:q-multiplexor}), and needs not to know the explicit construction of QRAM.
Not only restricted to the quantum multiplexor, the automatic parallelisation from our framework also works for general quantum recursive programs.

\end{document}